%% file: main.tex
\documentclass[11pt, twoside, sort, a4paper]{article}
\usepackage{graphicx}
\usepackage{subcaption}
\usepackage{hyperref, bm}
\usepackage{mathtools}

\usepackage{dsfont}
\newcommand{\myblue}[1]{\color{black}{#1}}
\newcommand{\apnum}[1]{\color{black}{#1}}
\newcommand{\apnumb}[1]{\color{black}{#1}}

\newcommand{\zblue}[1]{\color{black}{#1}}

\newcommand{\myD}{\mathbb{D}}
\newcommand{\myDin}{\myD_{\myin}}
\newcommand{\myDinc}{\myD_{\myout}^{\myinf}}
\usepackage{geometry}
\usepackage{empheq}
 \usepackage{mathrsfs}
 \usepackage{tikz}
\newcommand{\myout}{\scalebox{0.7}{\text{out}}}

\newcommand{\rhoh}{\hat{\rho}}

\newcommand{\mysum}{\sideset{}{^{\boldsymbol{*}}}\sum}
\AtBeginDocument{
\addtolength{\abovedisplayskip}{-0.4ex}
\addtolength{\abovedisplayshortskip}{-0.4ex}
\addtolength{\belowdisplayskip}{-0.4ex}
\addtolength{\belowdisplayshortskip}{-0.4ex}
}
\makeatletter
\newcommand{\subalign}[1]{%
  \vcenter{%
    \Let@ \restore@math@cr \default@tag
    \baselineskip\fontdimen10 \scriptfont\tw@
    \advance\baselineskip\fontdimen12 \scriptfont\tw@
    \lineskip\thr@@\fontdimen8 \scriptfont\thr@@
    \lineskiplimit\lineskip
    \ialign{\hfil$\m@th\scriptstyle##$&$\m@th\scriptstyle{}##$\crcr
      #1\crcr
    }%
  }
}

\usepackage{amsmath, amsthm, amssymb}
\usepackage{caption}
\usepackage{tikz}
\usetikzlibrary{matrix, positioning, calc}
\usepackage{multirow, nccmath}
\usepackage{algorithmic}
\usepackage{algorithm}

\usepackage{comment}
\usepackage{color}
\usepackage[titletoc,title]{appendix}

\usepackage[running, mathlines]{lineno}

\newcommand{\EQ}{\begin{equation}}
\newcommand{\EN}{\end{equation}}
\newcommand{\EQS}{\begin{equation*}}
\newcommand{\ENS}{\end{equation*}}
\newcommand{\EQA}{\begin{eqnarray}}
\newcommand{\ENA}{\end{eqnarray}}
\newcommand{\EQAS}{\begin{eqnarray*}}
\newcommand{\ENAS}{\end{eqnarray*}}
\newcommand{\ds}{\displaystyle}

\usepackage{accents}
\newlength{\dhatheight}

\newcommand{\myinf}{\scalebox{0.6}{$\infty$}}

\newcommand{\errorf}{\scalebox{0.9}{$\mathcal{E}_{f}$}}

\newcommand{\errorb}{\scalebox{0.9}{$\mathcal{E}_{b}$}}
\newcommand{\errorc}{\scalebox{0.9}{$\mathcal{E}_{c}$}}

\newcommand{\myin}{\scalebox{0.7}{\text{in}}}

\newcommand{\mymin}{\scalebox{0.55}{$\min$}}
\newcommand{\mymax}{\scalebox{0.55}{$\max$}}
\newcommand{\myx}{\scalebox{0.6}{$x$}}
\newcommand{\myy}{\scalebox{0.6}{$y$}}

\newcommand{\dtau}{\Delta \tau}

\newcommand{\taus}{\tau_{m}}

\newcommand{\x}{{\bf{x}}}

\newcommand{\Lcal}{\mathcal{L}}

\newcommand{\Ocal}{\mathcal{O}}
\newcommand{\Jcal}{\mathcal{J}}

\newcommand{\U}{{\bf{u}}}

\newcommand{\V}{{\bf{v}}}

\newcommand{\gb}{{\bf{g}}}

\newcommand{\Rbb}{\mathbb{R}}

\newcommand{\Nbb}{\mathbb{N}}
 \newcommand{\Jbb}{\mathbb{J}}
\newcommand{\Nd}{\mathbb{N}^{\dagger}}

\newcommand{\Jd}{\mathbb{J}^{\dagger}}

\newcommand{\Ebb}{\mathbb{E}}

\newcommand{\vh}{\hat{v}}

\newcommand{\Oinf}{\Omega^{\scalebox{0.5}{$\infty$}}}
\numberwithin{equation}{section}
\numberwithin{table}{section}
\numberwithin{figure}{section}

\newtheorem{definition}{Definition}
\newtheorem{theorem}{Theorem}
\newtheorem{lemma}{Lemma}
\newtheorem{remark}{Remark}
\newtheorem{assumption}{Assumption}

\newtheorem{corollary}{Corollary}

\numberwithin{definition}{section}
\numberwithin{theorem}{section}
\numberwithin{lemma}{section}
\numberwithin{remark}{section}
\numberwithin{assumption}{section}
\numberwithin{condition}{section}
\numberwithin{property}{section}
\numberwithin{proposition}{section}
\numberwithin{corollary}{section}
\numberwithin{algorithm}{section}

\def\r{\right}
\def\l{\left}
\def\f{\frac}
\def\s{\sqrt}

\newcommand{\bfs}{{\boldsymbol{s}}}
\newcommand{\bfx}{{\bf{x}}}

\newcommand{\bfz}{{\boldsymbol{z}}}

\newcommand{\tmu}{{\widetilde{\mu}}}
\newcommand{\tsig}{{\widetilde{\sigma}}}

\newcommand{\bmeta}{{\boldsymbol{\eta}}}
\newcommand{\bmbeta}{{\boldsymbol{\beta}}}
\newcommand{\bmxi}{{\boldsymbol{\xi}}}

\newcommand{\bmS}{{\boldsymbol{S}}}

\newcommand{\bmtmu}{{\boldsymbol{\widetilde{\mu}}}}

\newcommand{\bfC}{\boldsymbol{C}}

\newcommand{\bfCI}{\bfC^{-1}}

\newcommand{\bfbz}{\bmbeta+\bfz}
\newcommand{\bfA}{\boldsymbol{A}}
\newcommand{\bfal}{\boldsymbol{\alpha}}
\newcommand{\md}{d}

\makeatletter
\renewcommand*\env@matrix[1][c]{\hskip -\arraycolsep
  \let\@ifnextchar\new@ifnextchar
  \array{*\c@MaxMatrixCols #1}}
\makeatother

\makeatletter
\def\thm@space@setup{\thm@preskip=3pt
\thm@postskip=3pt}
\makeatother
\usepackage{titlesec}
\titlespacing*{\section}
  {0pt}{0.7\baselineskip}{0.2\baselineskip}
\titlespacing*{\subsection}
  {0pt}{0.7\baselineskip}{0.2\baselineskip}

\usepackage[normalem]{ulem}
\usepackage{cancel}

\begin{document}

\title{Numerical analysis of American option pricing in a two-asset jump-diffusion model
}

\author{
Hao Zhou \thanks{School of Mathematics and Physics, The University of Queensland, St Lucia, Brisbane 4072, Australia,
email: \texttt{h.zhou3@uq.net.au}
}
\and
Duy-Minh Dang\thanks{School of Mathematics and Physics, The University of Queensland, St Lucia, Brisbane 4072, Australia,
email: \texttt{duyminh.dang@uq.edu.au}
}
}
\date{\today}
\maketitle
\begin{abstract}
This paper addresses an {\apnum{important}} gap in rigorous numerical treatments for pricing American options under correlated two-asset jump-diffusion models using the viscosity solution {\apnum{framework}}, with a particular focus on the Merton model.  The pricing of these options is governed by complex two-dimensional (2-D) variational inequalities
that incorporate cross-derivative terms and nonlocal integro-differential terms due to the presence of jumps. Existing numerical methods, primarily based on finite differences, often struggle with preserving monotonicity in the approximation of cross-derivatives-a key requirement for ensuring convergence to the viscosity solution. In addition, these methods face challenges in accurately discretizing 2-D jump integrals.

We introduce a novel approach to effectively tackle the aforementioned variational inequalities {\apnum{while seamlessly handling}} cross-derivative terms and nonlocal integro-differential terms through an efficient and straightforward-to-implement monotone integration scheme. Within each timestep, our approach {\apnum{explicitly enforces the inequality constraint}}, resulting in a 2-D Partial Integro-Differential Equation (PIDE) to solve. Its solution is then expressed as a 2-D convolution integral involving the Green's function of the PIDE. We derive an infinite series representation of this Green's function, where each term is strictly positive and computable. This series facilitates the numerical approximation of the PIDE solution through a monotone integration method, such as the composite quadrature rule. To further enhance efficiency, we {\apnum{develop an efficient implementation}} of this monotone integration scheme via Fast Fourier Transforms, exploiting the Toeplitz matrix structure.

The proposed method is proved to be both $\ell_{\infty} $-stable and consistent in the viscosity sense, ensuring its convergence to the viscosity solution of the variational inequality. Extensive numerical results validate the effectiveness and robustness of our approach, highlighting its practical applicability and theoretical soundness.
\vspace{.1in}

\vspace{.1in}

\end{abstract}

\input{section_1_Introduction}

\input{section_2_VI_formulation}
\input{section_3_Green_function}
\input{section_4_Numerical_methods}
\input{section_5_Convergence}

\input{section_6_Numerical_experiments}
\input{section_7_Conclusion}

{\apnum{\section*{Acknowledgments}
The authors are grateful to the two anonymous referees for their constructive comments and suggestions, which have significantly improved the quality of this work.}}


\section*{Appendices}
\appendix
\input{Appendix}
\end{document}

%% file: section_1_Introduction.tex
\section{Introduction}
\label{sec:intro}
{\apnumb{In stochastic control problems, such as those arising in financial mathematics, value functions are often non-smooth, prompting the use of viscosity solutions \cite{crandall_ishii_lions1992, crandall1983viscosity, fleming2006controlled, pham1998optimal}. The framework for provable convergence numerical methods, established by Barles and Souganidis in \cite{barles-souganidis:1991}, requires
them  to be (i) $\ell_{\infty}$-stable, (ii) consistent, and (iii) monotone in the viscosity sense, assuming a strong comparison principle holds. Achieving monotonicity is often the most difficult criterion, and non-monotone schemes can fail to converge to viscosity solutions, leading to violations of the fundamental no-arbitrage principle in finance \cite{Oberman2006, pooley2003}.

While monotonicity is a well-established sufficient condition for convergence to the viscosity solution \cite{barles-souganidis:1991}, its necessity remains an open question. In some specialized settings (see, for example, \cite{warin2016, bokanowski2010convergence}), non-monotone schemes can still converge, yet no general result analogous to Barles–Souganidis \cite{barles-souganidis:1991} exists for such methods. Determining whether monotonicity is strictly required for convergence remains an important open problem in viscosity-solution-based numerical analysis. Nonetheless, monotone schemes remain the most reliable approach for ensuring convergence in complex settings.

Monotone finite difference schemes are typically constructed using positive coefficient discretization techniques \cite{wang08}, and rigorous convergence results exist for one-dimensional (1D) models, both with and without jumps \cite{yann04, forsyth_2002a}.
However, extending these results to multi-dimensional settings presents significant challenges, particularly when the underlying assets are correlated. In such cases, the local coordinate rotation of the computational stencil  improves stability and accuracy, but this technique is fairly complex and introduces significant computational overhead \cite{MaForsyth2015, clift2008numerical, bonnans2003consistency}. Moreover, accurate discretization of the nonlocal integro-differential terms arising from jumps remains a difficult task.

Fourier-based integration methods often rely on an analytical expression for the Fourier transform of the underlying transition density function or of an associated Green's function, as demonstrated in various studies \cite{Pavel2016, alonso-garcca_wood_ziveyi_2018, Huang2018, ForsythLabahn2017, lu2024semi, Fang2008, ruijter2012two, Ruijter2013}. Among these, the Fourier-cosine (COS) method \cite{Fang2008, Ruijter2013, ruijter2012two} is particularly notable for achieving high-order convergence in piecewise smooth problems. However, for general stochastic optimal control problems, such as asset allocation \cite{zhang2023monotone, DangForsyth2014}, which often involve complex and non-smooth structures, this high-order convergence is generally unattainable \cite{ForsythLabahn2017, lu2024semi}.

When applicable, Fourier-based integration methods offer advantages such as eliminating timestepping errors between intervention dates (e.g.\ rebalancing dates in asset allocation)  and naturally handling complex dynamics like jump-diffusion and regime-switching. However, a key drawback is their potential loss of monotonicity, which can lead to violations of the no-arbitrage principle in numerical value functions. This poses significant challenges in stochastic optimal control, where accuracy is crucial for optimal decision-making \cite{ForsythLabahn2017}. In the same vein of research, recent works on $\epsilon$-monotone Fourier methods for control problems in finance merit attention \cite{ForsythLabahn2017, online23, lu2024semi, LuDang2023}.

These challenges have motivated the development of strictly monotone numerical integration methods, such as the approach in \cite{zhang2023monotone} for asset allocation under jump-diffusion models within an impulse control framework with discrete rebalancing. Here, the stock index follows the 1D Merton \cite{merton1975} or Kou \cite{kou01} jump-diffusion dynamics, while the bond index evolves deterministically with a constant drift rate. Since the bond index dynamics lack a diffusion or jump operator, the asset allocation problem reduces to a weakly coupled system of one-dimensional (1D) equations linked through impulse control.

Akin to other Fourier-based  integration methods, the approach in \cite{zhang2023monotone} leverages a known closed-form expression for the Fourier transform of an associated 1D conditional density function. However, a key feature of this method is the derivation of an infinite series representation of this density function, in which each term is nonnegative and explicitly computable. This representation is used to advance each 1D equation in time, ensuring strict monotonicity. However, the method remains primarily suited for 1D problems within a discrete rebalancing setting where the other dimension is weakly coupled.

This work generalizes the strictly monotone integration methodology of \cite{zhang2023monotone} to a fully 2-D jump-diffusion setting, addressing the key challenge of cross-derivative terms and 2-D jumps. Furthermore, we establish its convergence to the viscosity solution in continuous time. We apply this framework to American option pricing under a correlated two-asset jump-diffusion model—an important yet computationally demanding problem in financial mathematics. While American options exhibit relatively weak nonlinearity, they provide a suitable setting to develop and assess strictly monotone numerical integration techniques in a fully 2-D framework, without the added complexity of more advanced financial applications, such as that of fully 2-D asset allocation.

\subsection{American options}
American options play a crucial role in financial markets and are widely traded for both hedging and speculative purposes. Unlike European options, which can only be exercised at expiration, American options allow exercise at any time before expiration. While this flexibility enhances their appeal across equities, commodities, and bonds, it also introduces substantial mathematical and computational challenges due to the absence of closed-form solutions in most cases.

From a mathematical perspective, American option pricing is formulated as a variational inequality due to its optimal stopping feature. The problem becomes particularly challenging when incorporating jumps, as the governing equations involve nonlocal integro-differential terms \cite{zhang1997numerical, pham1997optimal}. These complexities necessitate advanced numerical methods to achieve accurate valuations.

A common approach to tackling the variational inequality in American option pricing is to reformulate it as a partial (integro-)differential complementarity problem \cite{christara_2010, for98d, clift2008numerical, forsyth_2002a, boen2020operator}. This reformulation captures the early exercise feature through time-dependent complementarity conditions, effectively handling the boundary between the exercise and continuation regions. Finite difference techniques are widely used to solve these problems, resulting in nonlinear discretized equations that require iterative solvers such as the projected successive over-relaxation method (PSOR) \cite{cryer1971solution} and penalty methods \cite{for98d}. For models incorporating jumps, fixed-point iterations can be used to address the integral terms, as shown in \cite{clift2008numerical} for two-asset jump-diffusion models. Additionally, efficient operator splitting schemes—including implicit-explicit and alternating direction implicit types—have recently been proposed for American options under the two-asset Merton jump-diffusion model \cite{boen2020operator}, while the method of lines has been applied to the same problem under the two-asset Kou model \cite{hout2024efficient}.

An alternative approach to pricing American options is to approximate them using Bermudan options, a discrete-time counterpart, and analyze the convergence of the solution as the early exercise intervals shrink to zero. In this context, Fourier-based COS methods \cite{ruijter2012two} have been applied to  2\mbox{-}D Bermudan options.
However, for this application, the COS method requires particularly careful boundary tracking despite its high-order convergence for piecewise smooth problems.
Notably, as observed in \cite{ForsythLabahn2017, du2024fourier}, the COS method may produce negative option prices for short maturities, akin to cases where early exercise opportunities become more frequent. This phenomenon is closely tied to its potential loss of monotonicity, a limitation highlighted in previous discussions. Consequently, this approach may struggle to ensure convergence from Bermudan to American options.}}

\subsection{Main contributions}
This paper bridges the gap in existing numerical methods by introducing an efficient, straightforward-to-implement monotone integration scheme for the variational inequalities governing American options under the two-asset Merton jump-diffusion model. Our approach simultaneously handles cross-derivative terms and nonlocal integro-differential terms, simplifying the design of monotone schemes and ensuring convergence to the viscosity solution. In doing so, we address key computational challenges in current numerical techniques.

The main contributions of our paper are outlined below.
\begin{itemize}

\item[(i)]
We present the localized variational inequality for pricing American options under the two-asset Merton jump-diffusion model, formulated on an infinite domain {\apnum{in log-price variables}} with a finite interior and artificial boundary conditions. Using a probabilistic technique, we demonstrate that  the difference between the solutions of the localized and full-domain variational inequalities decreases exponentially {\apnum{with respect to the log-price domain size}}. In addition, we establish that the localized variational inequality satisfies a comparison result.

\item[(ii)]
We develop a monotone scheme for the variational inequality that explicitly enforces the inequality constraint. We solve a 2-D Partial Integro-Differential Equation (PIDE) at each timestep to approximate the continuation value, followed by an intervention action applied at the end of the timestep. Using the closed-form Fourier transform of the Green's function, we derive an infinite series representation where each term is non-negative. This enables the direct approximation of the PIDE's solutions via 2-D convolution integrals, using a monotone numerical integration method.

\item[(iii)]
We implement the monotone integration scheme efficiently by exploiting the Toeplitz matrix structure and using Fast Fourier Transforms (FFTs) combined with circulant convolution. The implementation process includes expanding the inner summation's convolution kernel into a circulant matrix, followed by transforming the double summation kernel into a circulant block structure. This allows the circulant matrix-vector product to be efficiently computed as a circulant convolution using 2-D FFTs.

\item[(iv)]
We prove that the proposed monotone scheme is both $\ell_{\infty}$-stable and consistent in the viscosity sense, ensuring pointwise convergence to the viscosity solution of the variational inequality as the discretization parameter approaches zero.

\item[(v)]
Extensive numerical results demonstrate strong agreement with benchmark solutions from published test cases, including those obtained via operator splitting methods,  establishing our method as a reference for validating numerical techniques.

\end{itemize}
{\apnum{While this work focuses on American option pricing under a correlated two-asset Merton jump-diffusion model, the core methodology—particularly the infinite series representation of the Green's function, where each term is non-negative—can be extended to other stochastic control problems. One such application is asset allocation with a stock index and a bond index. For discrete rebalancing, the 2-D extension is straightforward, as time advancement can be handled using a similar infinite series representation. For continuous rebalancing, a natural approach is to start from the discrete setting and leverage insights from \cite{zhang2023monotone} to analyze the limit as the rebalancing interval approaches zero.

The extension to the 2-D Kou model presents significant additional challenges due to its piecewise-exponential structure, leading to complex multi-region double integrals. We leave this extension for future work and refer the reader to Subsection~\ref{ssc:Koutwo}, where we discuss these difficulties and outline a potential neural network-based approach.
}}

The remainder of the paper is organized as follows. In Section~\ref{sc:VIs_VS}, we provide an overview of the two-asset Merton jump-diffusion model and present the corresponding variational inequality. We then define a localized version of this problem, incorporating boundary conditions for the sub-domains. Section~\ref{sc:green} introduces the associated Green's function and its infinite series representation. In Section~\ref{section:num}, we describe a simple, yet effective, monotone integration scheme based on a composite 2-D quadrature rule. Section~\ref{sc:conv} establishes the mathematical convergence of the proposed scheme to the viscosity solution of the localized variational inequality. Numerical results are discussed in Section~\ref{sec:num_test}, and finally, Section~\ref{sc:conclude} concludes the paper and outlines directions for future research.

%% file: section_2_VI_formulation.tex
\section{Variational inequalities and viscosity solution}
\label{sc:VIs_VS}
We consider a complete filtered probability space $(\mathfrak{S}, \mathfrak{F}, \mathfrak{F}_{0 \le t \le T}, \mathfrak{Q})$, which includes a sample space $\mathfrak{S}$, a sigma-algebra $\mathfrak{F}$, a filtration $\mathfrak{F}_{0 \le t \le T}$ for a finite time horizon $T > 0$, and a risk-neutral measure $\mathfrak{Q}$. For each $t \in [0, T]$,  $X_t$ and $Y_t$ represent the prices of two distinct underlying assets.
These price processes are modeled under the risk-neutral measure to follow jump-diffusion dynamics given by
\begin{linenomath}
 \begin{subequations}
 \label{eq:dynamics}
\begin{empheq}[left={\empheqlbrace}]{alignat=3}
&\frac{dX_t}{X_t} = \left(r - \lambda \kappa_{\myx} \right) dt + \sigma_{\myx} dW^{\myx}_t + d\l(\sum_{\iota=1}^{\pi_t} (\xi^{(\iota)}_{\myx}-1)\r), &\qquad ~X_0 = x_0>0,
\\
&\frac{dY_t}{Y_t} = \left(r - \lambda \kappa_{\myy} \right) dt + \sigma_{\myy} dW^{\myy}_t + d\l(\sum_{\iota=1}^{\pi_t} (\xi^{(\iota)}_{\myy}-1)\r), &\qquad Y_0 = y_0>0.
\label{eq:Z_dynamics*}
\end{empheq}
\end{subequations}
\end{linenomath}
Here, $r > 0$ denotes the risk-free interest rate, and $\sigma_x>0$ and $\sigma_y > 0$  represent the instantaneous volatility of the respective underlying asset. The processes $\{W_{t}^{x}\}_{t \in [0, T]}$ and $\{W_{t}^{y}\}_{t \in [0, T]}$ are two correlated Brownian motions, with $dW^{x}_t dW^{y}_t=\rho dt$, where $-1< \rho< 1$ is the correlation parameter. The process $\{\pi_t\}_{0\le t \le T}$ is a Poisson process with a constant finite intensity rate $\lambda\geq 0$. The random variables $\xi_{\myx}$ and $\xi_{\myy}$, representing the jump multipliers, are two correlated positive random variables with correlation coefficient $\hat{\rho}\in(-1,1)$.
 In  \eqref{eq:dynamics},  $\{\xi_{\myx}^{(\iota)}\}_{\iota = 1}^{\infty}$ and  $\{\xi_{\myy}^{(\iota)}\}_{\iota = 1}^{\infty}$ are independent and identically distributed (i.i.d.) random variables having the same distribution as $\xi_{\myx}$  and $\xi_{\myy}$, respectively; the quantities $\kappa_{\myx}=\Ebb\left[\xi_{\myx}-1\right]$ and
 $\kappa_{\myy}=\Ebb\left[\xi_{\myy}-1\right]$, where $\mathbb{E}[\cdot]$ is the expectation operator taken under
 the risk-neutral measure $\mathfrak{Q}$.

In this paper, we focus our attention on the Merton jump-diffusion model \cite{merton1975}, where the jump multiplier $\xi_{\myx}$ and $\xi_{\myy}$ subject to log-normal distribution, respectively. Specifically, we denote by $f(s_x,s_y)$ the joint density function of the random variable $\ln(\xi_{\myx})  \sim \text{Normal}\l(\tmu_{\myx}, \tsig_{\myx}^2\r)$ and $\ln(\xi_{\myy}) \sim \text{Normal}\l(\tmu_{\myy}, \tsig_{\myy}^2\r)$ with correlation $\rhoh$. Consequently, the joint probability density function (PDF) is given by
\begin{small}
\EQA
f(s_x,s_y)\!=\!\frac{1}{2 \pi  \tsig_{\myx} \tsig_{\myy}   \sqrt{1-\rhoh^2}}\exp
        \bigg(\frac{-1}{2(1 - \rhoh^2)}\bigg[
          \bigg(\frac{s_x-\tmu_{\myx}}{\tsig_{\myx}}\bigg)^2\!\! \!\!-\!
          2\rhoh\bigg(\frac{(s_x - \tmu_{\myx})(s_y - \tmu_{\myy})}{\tsig_{\myx}\tsig_{\myy}}\bigg)\!\!+\!\!
          \bigg(\frac{s_y - \tmu_{\myy}}{\tsig_{\myy}}\bigg)^2
        \bigg]
       \bigg).\label{eq: PDF p(psi) Merton}
\ENA
\end{small}

\subsection{Formulation}
\label{sc:formulation}
For the underlying process $(X_t,Y_t),~t\in [0,T]$, let $(a,b)$ be the state of the system.
We denoted by $v''(a, b,t)$ the time-$t$ no-arbitrage price of a two-asset American option contract
with maturity $T$ and payoff $\vh(a,b)$. It is established that $ v''(\cdot)$ is given by the optimal stopping problem \cite{jacka1991optimal,karatzas1988pricing,martini2000american,pham1998optimal, jaillet1990variational, zhang1997numerical}
\EQA
\label{eq:vf_opt_stop}
v''(a, b, t)=
\sup_{t \le \gamma \le T}\Ebb^{a, b}_{t}
\l[e^{-r(\gamma-t)}\vh'(X_{\gamma},Y_{\gamma})\r], \qquad (a,b,t)\in\Rbb_{+}^2\times [0,T].
\ENA
Here, $\gamma$ represents a stopping time; 
$\mathbb{E}^{x, y}_t$ denotes the conditional expectation under the risk-neutral measure $\mathfrak{Q}$, conditioned on $(X_t, Y_t) = (a, b)$. We focus on the put option case, where the payoff function $\vh'(\cdot)$ is bounded and continuous.

The methods of variational inequalities, originally developed in \cite{bensoussan1982applications}, are widely used for pricing American options, as evidenced by \cite{zhang1997numerical, jaillet1990variational, pham1997optimal, pham1998optimal}, among many other publications. The value function $v''(\cdot)$, defined in \eqref{eq:vf_opt_stop}, is known to be non-smooth, which prompts the use of the notion of viscosity solutions. This approach provides a powerful means for characterizing the value functions in stochastic control problems \cite{crandall_ishii_lions1992, crandall1983viscosity, crandall1984some, fleming2006controlled, pham1998optimal}.

It is well-established that the value function $v''(\cdot)$, defined in \eqref{eq:vf_opt_stop}, is the unique viscosity solution of a variational inequality as noted in \cite{oksendal1997viscosity,pham1997optimal, pham1998optimal}.
While the original references describe the variational inequality using the spatial variables $(a, b)$, our approach employs a logarithmic transformation for theoretical analysis and numerical method development.
Specifically, with $\tau = T - t$, and given positive values for $a$ and $b$, we apply the transformation $x = \ln(a)\in(-\infty,\infty)$ and $y = \ln(b)\in(-\infty,\infty)$.
With $\x = (x,y,\tau)$, we define $v'(\x) \equiv v'(x,y,\tau) = v''(e^x, e^y, T -t)$
and $\vh(x,y) = \vh'(e^x, e^y)$.
Consequently, $v'(\cdot)$ is the unique viscosity solution of the variational inequality given by
\begin{linenomath}
 \begin{subequations}
 \label{eq:VIs_log_full}
\begin{empheq} [left={\empheqlbrace}]{alignat=3}
\min\l\{ \partial v'/ \partial \tau-\Lcal v'-\Jcal v', v'-\hat{v}\r\}=0,&\qquad \x \in \Rbb^2  \times (0, T],
\label{eq:VI_log_full}
\\
\label{eq:boundary_log_full}
v'-\hat{v}=0,&\qquad \x \in \Rbb^{2}\times\{0\}.
\end{empheq}
\end{subequations}
\end{linenomath}
Here, the differential and jump operators $\Lcal(\cdot)$ and $\Jcal(\cdot)$ are defined as follows
\EQA
\Lcal \psi
&=&\frac{\sigma_{\myx}^2}{2} \frac{\partial^2 \psi}{\partial x^2}  + \bigg(r-\lambda\kappa_{\myx}-\frac{\sigma_{\myx}^2}{2}\bigg)\frac{\partial \psi}{\partial x} +\frac{\sigma_{\myy}^2}{2} \frac{\partial^2 \psi}{\partial y^2} + \bigg(r-\lambda\kappa_{\myy}-\frac{\sigma_{\myy}^2}{2}\bigg) \frac{\partial \psi}{\partial y} +\rho \sigma_{\myx}\sigma_{\myy} \frac{\partial^2 \psi}{\partial x\partial y}  - (r+\lambda)\psi,
\nonumber
\\
\Jcal \psi
&=&\lambda \iint_{\Rbb^2} \psi(x+s_x,y+s_y, \tau)~ f(s_x,s_y)~\md s_x\md s_y,
\label{eq:Operator_LJ_log}
\ENA
where $f(s_x,s_y)$ is the joint probability density function of $\l(\ln(\xi_{\myx}),\ln(\xi_{\myy})\r)$.

\subsection{Localization}
Under the log transformation, the formulation \eqref{eq:VIs_log_full} is  posed on an infinite spatial domain $\Rbb^2$.
For problem statement and convergence analysis of numerical schemes, we define a localized pricing problem with a finitem, open, spatial interior sub-domain, denoted by $\myD_{\myin} \subset \Rbb^2$. More specifically, with $x_{\mymin}< 0 < x_{\mymax}$ and $y_{\mymin}< 0 < y_{\mymax}$, where
$x_{\mymin}$,  $x_{\mymax}$, $|y_{\mymin}|$, and $y_{\mymax}$ are sufficiently large,
$\myD_{\myin}$ and its complement $\myDinc$ are respectively defined as follows
\EQ
\label{eq:regG}
\myDin \equiv (x_{\mymin},~ x_{\mymax}) \times (y_{\mymin}, y_{\mymax}), \quad \text{ and }
\quad
\myDinc = \Rbb^2\setminus \myDin.
\EN
Since the jump operator $\Jcal (\cdot)$ is non-local, computing the integral  \eqref{eq:Operator_LJ_log} for $(x, y) \in \myD_{\myin}$  typically requires knowledge of $v(\cdot)$ within the infinite outer boundary sub-domain $\myDinc$. Therefor, appropriate boundary conditions must be established for $\myDinc$. In the following, we define the definition domain and its sub-domains, discuss boundary conditions, and
investigate the impact of artificial boundary conditions on $v(\cdot)$.

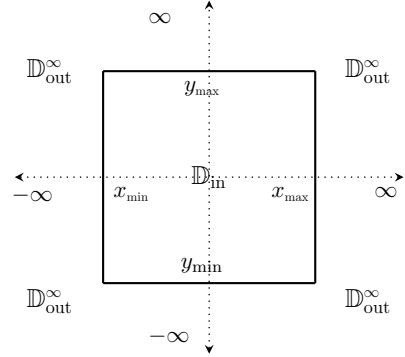
\begin{figure}[H]
\begin{minipage}{0.5\linewidth}
The definition domain comprises a finite sub-domain and an infinite boundary
sub-domain, defined as follows.
\begin{linenomath}
\begin{align}
\label{eq:sub_domain_whole}
\Oinf& = (-\infty, \infty) \times (-\infty, \infty)   \times [0, T],
\nonumber
\\
\Oinf_{\tau_0}
& = (-\infty, \infty) \times(-\infty, \infty)\times \{0\},
\\
\Omega_{\myin}  &=(x_{\mymin}, x_{\mymax}) \times(y_{\mymin}, y_{\mymax}) \times (0,T]
\equiv \myD_{\myin} \times (0,T],
\nonumber
\\
\Oinf_{\myout}
&=  \Oinf\setminus\Oinf_{\tau_0}
\setminus \Omega_{\myin} \equiv \myDinc  \times (0,T].
\nonumber
\end{align}
\end{linenomath}
For subsequent use, we also define the following region: $\Omega_{\tau_0}^{\myin}:= [x_{\mymin}, x_{\mymax}] \times [y_{\mymin}, y_{\mymax}] \times \{0\}$.
An illustration of the sub-domains for the localized problem corresponding to a fixed $\tau \in (0, T]$
is given in Figure~\ref{fig:domain}.
\end{minipage}
\hspace*{-0.5cm}
\begin{minipage}{0.48\linewidth}
\begin{center}
\begin{tikzpicture}[scale=0.47]
   \draw [dotted, line width=0.5pt] [stealth-stealth] (-1.5,2) -- (9.5,2);
    \draw [thick]  (1,5) -- (7,5);
   \draw [dotted, line width=0.5pt][stealth-stealth](4,-3) --(4,7);
   \draw [thick](7,-1) --(7,5);
   \draw [thick]  (1,-1) -- (1,5);
   \draw [thick]  (1,-1) -- (7,-1);
   \node at (4,2) {\scalebox{0.9}{$\myD_{\myin}$}};

%
    \node at (-0.5,5){\scalebox{0.9} {$\myDinc$}};
    \node at (-0.5,-1.5){\scalebox{0.9} {$\myDinc$}};

    \node at (8.5,-1.5) {\scalebox{0.9}{$\myDinc$}};
    \node at (8.5,5) {\scalebox{0.9}{$\myDinc$}};

   \node [below] at (-1,2) {\scalebox{0.8}{$-\infty$}};
   \node [below] at (9,2) {\scalebox{0.8}{$\infty$}};
   \node [right] at (2,-2.5) {\scalebox{0.8}{$-\infty$}};
   \node [right] at (2,6.5) {\scalebox{0.8}{$\infty$}};
   %
    \node [below] at (1.8,2) {\scalebox{0.8}{$x_{\mymin}$}};
    \node [below] at (6.3,2) {\scalebox{0.8} {$x_{\mymax}$}};
    \node [above] at (3.8,-1) {\scalebox{0.8}{$y_{\min}$}};
    \node [below] at (3.8,5) {\scalebox{0.8}{$y_{\mymax}$}};
\end{tikzpicture}
\end{center}
\vspace*{-0.5cm}
\caption{Spatial definition sub-domain at each $\tau\in [0, T]$.}
\label{fig:domain}
\end{minipage}
\end{figure}
\vspace*{-0.25cm}
\noindent For the outer boundary sub-domain $\Oinf_{\myout}$, boundary conditions are generally informed by financial reasonings or derived from the asymptotic behavior of the solution. In this study, we implement a straightforward Dirichlet boundary condition using a known bounded function $\hat{p}(\x)$ for $\x \in \Oinf_{\myout}$. Specifically, $\hat{p}(\x)$ belongs to the space of bounded functions $\mathcal{B}(\Oinf)$, which is defined as follows \cite{Barles2008,Seydel2009}
\EQ
\label{eq:G}
\begin{aligned}
\mathcal{B}(\Oinf) &= \big\{
\psi: \Oinf \to \mathbb{R},
~~\|\psi(\cdot)\|_{\infty} < \infty \big\}.
\end{aligned}
\EN
We denote by $v(\cdot)$ the function that solves the localized problem on $\Oinf$ with the initial and boundary condition given below
\begin{linenomath}
 \begin{subequations}
 \label{eq:VIs_log}
\begin{empheq} [left={\empheqlbrace}]{alignat=3}
\min\l\{ \partial v/ \partial \tau-\Lcal v-\Jcal v, v-\hat{v}\r\}=0,&\qquad \x \in \Omega_{\myin},
\label{eq:VI_log}
\\
v - \hat{p}  = 0, &\qquad \x \in \Oinf_{\myout},
\label{eq:inf_boundary_log}
\\
\label{eq:boundary_log}
v-\hat{v}=0,&\qquad \x \in \Oinf_{\tau_0}.
\end{empheq}
\end{subequations}
\end{linenomath}
The impact of the artificial boundary condition in $\Oinf_{\myout}$, as specified in \eqref{eq:inf_boundary_log},
on the solution within $\Omega_{\myin}$ is established in Lemma~\ref{lemma:bd_error} below.
For simplicity, in the lemma,
we assume that $x_{\mymax} = |x_{\mymin}| =  y_{\mymax} = |y_{\mymin}|$.
The proof can be generalized in a straightforward manner to accommodate
different values  for $x_{\mymin}$, $x_{\mymax}$, $y_{\mymin}$, and  $y_{\mymax}$.
 \begin{lemma}
\label{lemma:bd_error}
Assume that $\vh(\cdot)$ and $\hat{p}(\cdot)$ belong to $\mathcal{B}(\Oinf)$,
and that $A := x_{\mymax} = |x_{\mymin}| = y_{\mymax} = |y_{\mymin}|$.
Then, for $\x = (x, y, \tau) \in \Omega_{\myin}$, the difference between the solution
$v'(\cdot)$ and $v(\cdot)$ to their respective non-localized and localized variational inequalities
\eqref{eq:VIs_log_full} and \eqref{eq:VIs_log} is bounded by
\[
|v'(\x) - v(\x)| \le C(\tau) \l(\|\vh(\cdot)\|_{\infty} + \|\hat{p}(\cdot)\|_{\infty}\r) \l(e^{-(A - |x|)} + e^{-(A - |y|)}\r).
\]
Here, the constant $C(\tau)>0$ is bounded independently of $x_{\mymin}$, $x_{\mymax}$,
$y_{\mymin}$, and $y_{\mymax}$.
\end{lemma}
A proof of Lemma~\ref{lemma:bd_error} is provided in  Appendix~\ref{app:bd_error}.
{\apnum{This lemma establishes that, in log-price coordinates, the error between the localized and full-domain variational inequality solutions decays exponentially as the spatial interior domain size increases, i.e.\ the truncation error is of order $\mathcal{O}(e^{-A})$
as $A \to \infty$.\footnote{{\apnum{Equivalently, if $A'$ represents the spatial domain size in the original price scale for both assets, the error is $\mathcal{O}(1/A')$.}}}}} The result is a local pointwise estimate, indicating that the localization error is more pronounced near the boundary. This rapid decay implies that smaller computational domains can be used, thereby significantly reducing computational costs.

The boundedness conditions $\|\vh(\cdot)\|_{\infty}<\infty$ and $\|\hat{p}(\cdot)\|_{\infty}<\infty$ are satisfied for standard put options and Dirichlet boundary conditions.
{\apnum{However, for payoffs that grow unbounded as $x, y \to +\infty$ (e.g.\ standard calls), a more refined analysis is required. Large-deviation techniques have been employed in pure-diffusion models (see \cite{barles1995convergence}) to estimate the probability of large excursions beyond the truncated domain. These techniques can also be adapted to certain (finite-activity) jump processes (see, for example, \cite{jakubowski2005levy}).
A rigorous extension of these ideas to jump-diffusion models with unbounded payoffs, while theoretically possible, is considerably more involved and lies beyond the scope of this work.
We therefore concentrate on bounded-payoff options—such as puts—where $\|\hat{v}(\cdot)\|_{\infty}$ remains finite.
}}

For the remainder of the analysis,  we choose the Dirichlet condition based on discounted payoff as follows
\EQA
\label{eq:vxymax}
\hat{p}(x, y, \tau) = \vh(x, y) e^{-r\tau}, \quad (x, y, \tau) \in \Omega_{\myout}.
\ENA
While more sophisticated boundary conditions might involve the asymptotic properties of the variational inequality \eqref{eq:VI_log_full} as $x, y \to -\infty$ or $x, y \to \infty$, our observations indicate that these sophisticated boundary conditions do not significantly impact the accuracy of the numerical solution within  $\Omega_{\myin}$.
This will be  illustrated through numerical experiments in Subsection~\ref{sec:constant_pad}.

\subsection{Viscosity solution and a comparison result}
We now write \eqref{eq:sub_domain_whole} in a compact form, which includes the terminal and boundary
conditions in a single equation. We let $Dv({\mathbf{x}})$ and
$D^2 v( {\mathbf{x}} )$ represent the first-order and second-order partial derivatives of $v\left( {\mathbf{x}} \right)$. The variational inequality \eqref{eq:VIs_log} can be expressed compactly as
\EQA
\label{eq:F}
0 = F\left({ \mathbf{x}},
                    v({\mathbf{x}}),
                    Dv({\mathbf{x}}),
                    D^2 v({\mathbf{x}}),
                    \mathcal{J} v({\mathbf{x}})
                \right) \equiv F\left({ \mathbf{x}}, v\right),
\ENA
where
\begin{linenomath}
 \begin{subequations}
 \label{eq:fall}
\begin{empheq}[left={F\left({ \mathbf{x}}, v\right)=\empheqlbrace}]{alignat=3}
&~F_{\myin} \left({ \mathbf{x}}, v\right) &&= \min\l\{ \partial v/ \partial \tau-\Lcal v-\Jcal v, v-\hat{v}\r\},&&\qquad\x \in \Omega_{\myin},
\label{eq:Finn}
\\
&~F_{\myout}\left({ \mathbf{x}}, v\right) &&=v - e^{-r\tau} \hat{v}, && \qquad\x \in \Oinf_{\myout},
\label{eq:Fout}
\\
&~F_{\tau_0}\left({ \mathbf{x}}, v\right) &&= v-\hat{v},&&\qquad \x \in \Oinf_{\tau_0}.
\label{eq:ftau0}
\end{empheq}
\end{subequations}
\end{linenomath}
For a locally bounded function $\psi: \mathbb{D} \rightarrow \mathbb{R}$, where $\mathbb{D}$ is a closed
subset of $\mathbb{R}^n$, we recall its upper semicontinuous (u.s.c.\ in short) and the lower semicontinuous (l.s.c.\ in short) envelopes given by
\EQA
\label{eq:envelop}
\psi^*({\bf{\hat{x}}}) = \limsup_{
    \subalign{{\bf{x}} &\to {\bf{\hat{x}}}
\\
{\bf{x}}, {\bf{\hat{x}}} &\in\mathbb{X}
}}
\psi({\bf{x}})
\quad
(\text{resp.}
\quad
\psi_*({\bf{\hat{x}}}) = \liminf_{
    \subalign{{\bf{x}} &\to {\bf{\hat{x}}}
\\
{\bf{x}}, {\bf{\hat{x}}} &\in\mathbb{X}
}}
\psi({\bf{x}})
).
\ENA
\begin{definition}[Viscosity solution of \eqref{eq:F}]
\label{def:vis}
(i) A locally bounded function  $v\in \mathcal{B}(\Oinf)$ is a viscosity supersolution of \eqref{eq:F} in $\Oinf$ if and only if for all test function $\phi \in  \mathcal{B}(\Oinf)\cap\mathcal{C}^{\infty}(\Oinf)$
and for all points ${\bf{\hat{x}}} \in \Oinf$ such that $(v_*-\phi)$ has a \emph{global} minimum on $\Oinf$ at ${\bf{\hat{x}}}$
and $v_*({\bf{\hat{x}}}) = \phi({\bf{\hat{x}}})$, we have
\EQA
\label{eq:Def1}
F^* \left({\bf{\hat{x}}}, \phi({\bf{\hat{x}}}), D\phi({\bf{\hat{x}}}), D^2 \phi({\bf{\hat{x}}}),
            \mathcal{J} \phi({\bf{\hat{x}}})
             \right) \geq   0. 
\ENA
Viscosity subsolutions are defined symmetrically.

(ii) A locally bounded function $v\in \mathcal{B}(\Oinf)$ is a viscosity solution of \eqref{eq:F} in  $\Omega_{\myin} \cup \Omega_{\tau_0}^{\myin}$ if $v$ is a viscosity subsolution and a viscosity supersolution in $\Omega_{\myin} \cup \Omega_{\tau_0}^{\myin}$.
\label{Def:viscosity_VIs}
\end{definition}
In the context of numerical solutions to degenerate parabolic equations in finance, the convergence to viscosity solutions is ensured when the scheme is stable, consistent, and monotone, provided that a comparison result holds \cite{jakobsen2006maximum, barles-rouy:1998,  Barles2008, barles-souganidis:1991, barles-burdeau:1995, barles95a, barles:1997}. Specifically,  stability, consistency and monotonicity facilitate the identification of u.s.c.\ subsolutions and l.s.c.\ supersolutions  through the respective use of $\limsup$  and $\liminf$ of the numerical solutions as the discretization parameter approaches zero.

Suppose $\underline{v}(\cdot)$ and $\overline{v}(\cdot)$ respectively denote such subsolution  and supersolution within a region, referred to as $\mathcal{S}$, where $\mathcal{S} = \mathbb{S} \times [0, T]$ for an open set $\mathbb{S} \subseteq \Rbb^2$. By construction using
 $\limsup$ for  $\underline{v}(\cdot)$ and $\liminf$ for $\overline{v}(\cdot)$,
and the nature of $\limsup \ge  \liminf $, we have $\underline{v}({\mathbf{x}}) \ge \overline{v}({\mathbf{x}})$ for all ${\mathbf{x}} \in \mathcal{S}$.
If a comparison result holds in $\mathcal{S}$, it means that $\underline{v}({\mathbf{x}}) \le \overline{v}({\mathbf{x}})$ for all ${\mathbf{x}} \in \mathcal{S}$. Therefore,
$v({\mathbf{x}}) = \underline{v}({\mathbf{x}}) =  \overline{v}({\mathbf{x}})$
is the unique, continuous viscosity solution within the region  $\mathcal{S}$.

It is established  that the full-domain variational inequality defined in \eqref{eq:VIs_log_full} satisfies a comparison result in 
\cite{Seydel2009,pham1998optimal,awatif1991equqtions,ishii1989uniqueness}.
Similarly, the localized variational inequality~\eqref{eq:F} also satisfies a comparison result,
as detailed in the lemma below.
We recall $\Oinf_{\myout}$ defined in \eqref{eq:sub_domain_whole}.
\begin{lemma}
\label{theorem:comparison}
Suppose that  a locally bounded and u.s.c.\ function $\underline{v}:\Oinf \to \mathbb{R}$ and a locally bounded
l.s.c.\ function $\overline{v}:\Oinf \to \mathbb{R}$ are, respectively, a viscosity subsolution and supersolution of \eqref{eq:F} in the sense of Definition~\ref{def:vis}.
If $\underline{v}({\bf{x}}) \leq \overline{v}({\bf{x}})$ for all ${\bf{x}} \in \Oinf_{\tau_0}$,
and similarly for all ${\bf{x}} \in \Oinf_{\myout}$,
then it follows that $\underline{v}(\x) \leq \overline{v}(\x)$ for all $\x \in \Omega_{\myin}$.
\end{lemma}
The proof of the comparison result follows a similar approach to that in \cite{lu2024semi}, and is therefore omitted here for brevity.

%% file: section_3_Green_function.tex
\section{An associated Green's function}
\label{sc:green}
Central to our numerical scheme for the variational inequality~\eqref{eq:F} is the Green's function of
an associated PIDE in the variables  $(x, y)$, analyzed independently of the constraints dictated by the variational inequality.
To facilitate this analysis, for a fixed $\Delta \tau>0$,
let $\tau'\ge 0$ be such that $\tau' +\Delta \tau < T$, and proceed to consider the 2-D PIDE:
\EQ
\label{eq:2dPIDEs*}
\partial u/\partial \tau -\Lcal u-\Jcal u = 0, \qquad (x,y,\tau) \in  \Rbb^2 \times (\tau', \tau'+\Delta \tau],
\EN
subject to the time-$\tau'$ initial condition  specified by a generic function $\tilde{u}(\cdot, \tau')$.
We denote by the function $g(\cdot, \Delta \tau) \equiv g(x, x',y,y', \Delta \tau)$ the Green's function associated with the PIDE \eqref{eq:2dPIDEs*}. The stochastic system described in \eqref{eq:dynamics} exhibits spatial homogeneity, which  leads to the spatial translation-invariance of both the differential operator $\Lcal(\cdot)$ and the jump operator $\Jcal(\cdot)$. As a result, the Green's function $g(x, x',y,y', \Delta \tau)$ depends only on the relative displacement between starting and ending spatial points,
thereby simplifying to $g(x- x',y- y',   \Delta \tau)$.

\subsection{An infinite series representation of $\boldsymbol{g\l(\cdot\r)}$ }
\label{sec:inf_s_g}
We let $G(\eta_x,\eta_y,\Delta \tau)$ be the Fourier transform of $g(x, y, \Delta \tau)$ with respect to
the spatial variables, i.e.\
\EQ
\label{eq:ft_pair}
 \left\{
\begin{array}{lll}
\mathfrak{F}|g(x,y,\cdot)| (\eta_x,\eta_y)&=G(\eta_x,\eta_y,\cdot)&=\displaystyle\iint_{\Rbb^2}e^{-i(\eta_x x+\eta_y y)}g(x,y,\cdot)\md x\md y,
\\
\mathfrak{F}^{-1}|G(\eta_x,\eta_y,\cdot)|(x,y)&=g(x,y,\cdot)&=\frac{1}{(2\pi)^2}\displaystyle\iint_{\Rbb^2}e^{i(\eta_x x+\eta_y y)}G(\eta_x,\eta_y,\cdot)\md\eta_x \md\eta_y.
\end{array}
\right.
\EN
 A closed-form expression for $G(\eta_x,\eta_y,\cdot)$ is given as follows \cite{ruijter2012two}
\begin{align}
    \label{eq:G_closed}
    &G(\eta_x,\eta_y,\cdot)=\exp\big(\Psi\big(\eta_x,\eta_y\big) \Delta \tau\big), \text{ where }
\\
&\Psi(\eta_x,\eta_y)= -\frac{\sigma_{\myx}^{2}\eta_x^2}{2}-\frac{\sigma_{\myy}^{2}\eta_y^2}{2} \!+\!
\big(r-\lambda\kappa_x-\frac{\sigma_{\myx}^2}{2}\big)i\eta_x\!+\!\big(r-\lambda\kappa_y\!-\!\frac{\sigma_{\myy}^2}{2}\big)i\eta_y-\rho\sigma_{\myx}\sigma_{\myy}\eta_x\eta_y-(r+\lambda)+\lambda\Gamma(\eta_x,\eta_y)
\nonumber.
\end{align}
Here, $\Gamma(\eta_x,\eta_y)=\iint_{\Rbb^2}f(s_x,s_y)~e^{i(\eta_x s_x+\eta_y s_y)}~\md s_x\md s_y$, where $f(s_x,s_y)$ is the joint probability density function of
random variables $\xi_x$ and $\xi_y$ given in \eqref{eq: PDF p(psi) Merton}.

For convenience, we define $\bfz=[x,y]$, $\bmeta=[\eta_x,\eta_y]$ and $\bfs=[s_x,s_y]$ are the column vectors, $\bfz\cdot\bmeta$ is the dot product of vectors $\bfz$ and $\bmeta$, $\bfz^{\top}$ is the transpose of a vector $\bfz$, and $\tilde{\bfC}$ is the covariance matrix of $x$ and $y$. The covariance matrix $\tilde{\bfC}$ and its inverse $\tilde{\bfC}^{-1}$ are respectively given as follows
\EQ
\label{eq:cov_m}
\tilde{\bfC} = \begin{bmatrix}
    \sigma_{\myx}^2 & \rho\sigma_{\myx}\sigma_{\myy} \\
    \rho\sigma_{\myx}\sigma_{\myy} & \sigma_{\myy}^2
\end{bmatrix},\quad
\tilde{\bfC}^{-1} = \frac{1}{\det(\tilde{\bfC})}\begin{bmatrix}
    \sigma_{\myy}^2 & -\rho\sigma_{\myx}\sigma_{\myy} \\
    -\rho\sigma_{\myx}\sigma_{\myy} & \sigma_{\myx}^2
    \end{bmatrix},\quad \text{where } \det(\tilde{\bfC})=\sigma_{\myx}^2\sigma_{\myy}^2(1-\rho^2).
\EN
For subsequent use, we express the function $G(\eta_x,\eta_y,\cdot)$ given in \eqref{eq:G_closed} in a compact matrix-vector form as follows
\begin{align}
    G(\bmeta,\cdot)=\exp(\Psi(\bmeta) \Delta \tau), \quad \text{ with } {\Psi(\bmeta)= \bigg(-\frac{1}{2}\bmeta^{\top}\tilde{\bfC}\bmeta+i\tilde{\bmbeta}\cdot\bmeta-(r+\lambda)+\lambda\Gamma(\bmeta)\bigg)},
     \label{eq:vmf_G_closed}
\end{align}
where $\Gamma(\bmeta)=\int_{\Rbb^2}f(\bfs)~e^{i\bfs\cdot\bmeta}~\md \bfs$, and  $\tilde{\bmbeta}=\l[(r-\lambda\kappa_x-\frac{\sigma_{\myx}^2}{2}),~(r-\lambda\kappa_y-\frac{\sigma_{\myy}^2}{2})\r]$ is the column vector. For brevity, we use $\int_{\bmeta \in \Rbb^2}(\cdot)~ d\bmeta$ to represent the 2-D integral $\iint_{\Rbb^2}(\cdot)~d\eta_xd\eta_y$.
\begin{lemma}
\label{lemma:series_g}
Let $g(\bfz,\cdot)$ and $G(\bmeta,\cdot)$ be a Fourier transform pair defined in \eqref{eq:ft_pair} and $G(\bmeta,\cdot)$ is given in \eqref{eq:vmf_G_closed}. Then, the Green's function $g(\bfz, \Delta \tau)$ can be expressed as
\begin{align}
g(\bfz, \Delta \tau) &=  \f{1}{ 2\pi\s{\det(\bfC)} } \sum_{k=0}^{\infty} g_k(\bfz,\Delta \tau), \text{ where }
\label{eq:g_proof_sum}
\\
 g_k(\bfz,\Delta \tau) &=  \f{\l(\lambda \Delta \tau\r)^{k}}{k!}
\int_{\Rbb^2} \ldots \int_{\Rbb^2}
            	\exp\l(\theta  -
	\f{\l(\bmbeta + \bfz +  \bmS_k\r)^{\top}\bfC^{-1}(\bmbeta + \bfz +  \bmS_k)}{2}\r)
    \l(\prod_{\ell = 1}^{k} f(\bfs_\ell)\r) \md \bfs_1 \ldots \md \bfs_k.
\nonumber
\end{align}
{\apnum{Here, $\bmS_k = \sum_{\ell=1}^{k} \bfs_{\ell} =\sum_{\ell=1}^{k}[s_x, s_y]_{\ell}$,
with $\bmS_{0} = [0, 0]$.}}
\end{lemma}
A proof of Lemma~\ref{lemma:series_g} is provided in  Appendix~\ref{app:series_g}.
We emphasize that the infinite series representation in Lemma~\ref{lemma:series_g} does not rely on the specific form of the joint probability density function $f(\cdot)$, and thus it applies broadly to general two-asset jump-diffusion model. In the specific case of the two-asset Merton jump-diffusion model, where the joint probability density function $f(\cdot)$ is given by \eqref{eq: PDF p(psi) Merton}, the terms of the series can be explicitly evaluated, as detailed in the corollary below.
\begin{corollary}
\label{cor:twodis}
Let $\bmxi = [\xi_1,\xi_2]$ and $\bmtmu = [\tmu_1,\tmu_2]$.
 For the case $\ln(\bmxi) \sim \text{Normal}\l(\bmtmu, \bfC_{\mathcal{M}}\r)$
{\myblue{whose}} joint
PDF is given by \eqref{eq: PDF p(psi) Merton},
the infinite series representation of the conditional density $g(\bfz,\Delta \tau)$ given in Lemma~\ref{lemma:series_g}
is evaluated to $\ds g(\bfz, \Delta \tau) = g_0(\bfz, \Delta \tau) + \sum_{k=1}^{\infty} g_k(\bfz, \Delta \tau)$, where
\begin{align*}
 g_0(\cdot) = \f{\exp\big(\theta -
	\f{\l(\bmbeta + \bfz\r)^{\top}\bfC^{-1}(\bmbeta + \bfz)}{2}\big)}
{2\pi\s{ \det(\bfC) }}, \text{ and }~
g_k(\cdot) = \f{\l(\lambda \Delta \tau\r)^{k}}{k!} \, \f{\exp\big(\theta - \f{\l(\bmbeta + \bfz + k \bmtmu \r)^{\top}(\bfC+k\bfC_{\mathcal{M}})^{-1}\big(\bmbeta + \bfz + k \bmtmu \big)}{2}\big)}
{2\pi\s{\det(\bfC+k\bfC_{\mathcal{M}})}},
\end{align*}
with $\bfC = \Delta \tau \,\tilde{\bfC}$, $\bmbeta =  \Delta \tau\,\tilde{\bmbeta}$,
and $\theta = -(r+\lambda) \Delta \tau$.
\end{corollary}
A proof of Corollary \ref{cor:twodis} is given in Appendix~\ref{app:merton}.


{\apnum{
\subsection{Extension to the 2-D Kou model}
\label{ssc:Koutwo}
In the 2-D Merton model of Corollary~\ref{cor:twodis}, $[\log(\xi_x), \log(\xi_y)]$ are assumed to follow a bivariate normal distribution. Consequently, the sum of $k$ i.i.d.\ random vectors, each following this bivariate normal distribution, is also bivariate normal. Convolution with the 2-D diffusion kernel $\exp\l(\theta  -	\f{(\cdots)^{\top}\bfC^{-1}(\cdots)}{2}\r)$, which is itself Gaussian, then reduces to a straightforward Gaussian–Gaussian double integral. Summing over all possible numbers of jumps according to the Poisson distribution gives an infinite series of tractable Gaussian double integrals, thus leading to the closed‐form solution presented in Corollary~\ref{cor:twodis}.

However, extending this approach to the 2-D Kou model is much more challenging—even in the simplest case where $\log(\xi_x)$ and $\log(\xi_y)$ are independent. In this setting, the joint PDF of $\log(\xi_x)$ and $\log(\xi_y)$ is simply the product of
two 1-D Kou PDFs, resulting in a four‐region, piecewise‐exponential density on $\Rbb^2$.
Summing $k$ i.i.d.\ random vectors, each following this independent 2-D Kou jump distribution,
 leads to multi‐region double integrals, since each dimension can exhibit positive or negative jumps. As $k$ grows, the number of piecewise‐exponential components increases combinatorially. Moreover, these sums produce factors that often combine polynomial or gamma‐type terms with exponentials in two variables.

When convolved with the 2-D Gaussian diffusion kernel, the resulting integrals become exponential–Gaussian double integrals, which generally do not reduce to standard special functions (unlike the 1-D case, which sometimes admits the $\mathrm{Hh}(\cdot)$ family from \cite{AbramowitzStegun1972}, as shown in \cite{kou01}). Although one can still construct a series expansion over the Poisson‐distributed jumps, the individual terms would require intricate generalized special functions or numerous piecewise integrals, making the final expression far more cumbersome and less recognizable as a ``closed‐form'' solution.

A possible way forward is inspired by \cite{du2024fourier}, where a single-hidden-layer neural network (NN) with Gaussian activation functions is used to approximate an unknown transition density or Green's function via a closed-form expression of its Fourier transform. However, this approach can lead to potential loss of nonnegativity and requires retraining the neural network for different $\Delta \tau$.

For the general 2-D Kou model, an alternative approach is to instead approximate the joint PDF of $\log(\xi_x)$ and $\log(\xi_y)$ using a NN akin to that in \cite{du2024fourier}, where Gaussian activation functions are employed. Once trained, the joint PDF is represented as a finite sum of 2-D Gaussians, effectively forming a Gaussian mixture model. Notably, the nonnegativity of this approximation can also be enforced, aligning with the monotonicity requirements of numerical schemes.\footnote{{\apnum{While we focus on the 2-D Kou model, this approach can, in principle, be applied to other 2-D jump-diffusion dynamics where the joint jump density can be effectively approximated by a Gaussian mixture.}}}

This formulation transforms the convolution with the 2-D diffusion kernel into a tractable Gaussian–Gaussian double integral, allowing the same efficient techniques developed in this paper for the 2-D Merton model to be applied. Notably, because the NN does not depend on $\Delta \tau$, it can be trained only once and reused for all $\Delta \tau$, ensuring computational efficiency by eliminating the need for retraining at different time step sizes. We plan to explore this approach in future work in the context of fully 2-D asset allocation.

}}

\subsection{Truncated series and error}
For subsequent analysis, we study the truncation error in the infinite series
representation of the Green's function $g(\cdot)$ as given in \eqref{eq:g_proof_sum}.
Notably, this truncation error bound is derived independently of the specific form of the joint probability density function $f(\cdot)$, ensuring its applicability to a broad range of two-asset jump-diffusion models.

Specifically, for a fixed $\bfz = [x, y] \in \Rbb^2$, we denote by $g(\bfz, \Delta \tau, K)$ an approximation
of the Green's function $g(\bfz, \Delta \tau)$ using the first $(K+1)$ terms from the series \eqref{eq:g_proof_sum}.
As $K$ approaches $\infty$, the approximation $g(\bfz, \Delta \tau, K)$ becomes exact with
no loss of information. However, with a finite $K$,  the approximation incurs an error due to the truncation of the series.
This truncation error can be bounded as follows:
    \begin{align}
    | g(\bfz, \Delta\tau) - g(\bfz, \Delta \tau, K) | &=
    \l|\frac{1}{(2\pi)^2} \sum_{k=K+1}^{\infty} \f{\l(\lambda\Delta \tau\r)^k}{k!}
    \int_{\Rbb^2}
    e^{- \frac{1}{2} \bmeta^{\top} \mathbf{C}\bmeta +
    i\l(\bmbeta + \bfz \r) \cdot\bmeta  + \theta} \,
    \l(\Gamma\l( \bmeta \r) \r)^k ~ \md \bmeta \r|
    \nonumber\\
    &\leq
    \frac{1}{(2\pi)^2} \sum_{k=K+1}^{\infty} \f{\l(\lambda\Delta \tau\r)^k}{k!}
    \int_{\Rbb^2}
    \l|e^{- \frac{1}{2} \bmeta^{\top} \mathbf{C}\bmeta +
    i\l(\bmbeta + \bfz \r) \cdot\bmeta  + \theta}\r| \,
   \l| \l(\Gamma\l( \bmeta \r) \r)^k \r| ~ \md \bmeta
    \nonumber\\
     & ~{\buildrel (\text{i}) \over \le}~
    \frac{1}{(2\pi)^2} \sum_{k=K+1}^{\infty} \f{\l(\lambda\Delta \tau\r)^k}{k!}
    \int_{\Rbb^2}
  e^{- \frac{1}{2} \bmeta^{\top} \mathbf{C}\bmeta  + \theta} ~ \md \bmeta
    \nonumber\\
    &=
\sum_{k=K+1}^{\infty}
    \f{\exp(\theta)(\lambda\Delta\tau)^{k}}{{(k)}!2\pi\sqrt{\det(\bfC)}}  ~{\buildrel (\text{ii}) \over \le}~
    \f{e^{-(r
+\lambda)\Delta\tau}}{2\pi\sqrt{\det(\bfC)}}\,
    \f{(e\lambda \Delta\tau)^{K+1}}{(K+1)^{K+1}}.
    \label{eq:Kbound}
    \end{align}
    Here, in (i), we apply the following fact: if $\omega$ denotes a complex number, then the modulus of the complex exponential is equivalent to the exponential of the real part of $\omega$, i.e $\l|e^{\omega}\r|=\exp(\Re(\omega))$ and  $\l| \l(\Gamma\l( \bmeta \r) \r)^{K+1} \r|
    \le \l(\int_{\Rbb^{2}} f(\bfs)~\l|e^{i \bfs\cdot \bmeta}\r|~\md \bfs\r)^{K+1}
    \le 1$,
        (ii) is due to the Chernoff-Hoeffding bound for the tails of a Poisson distribution $\mathrm{Poi}(\lambda \Delta \tau)$, which reads as $\mathbb{P}\l(\mathrm{Poi}(\lambda \Delta \tau)\geq k\r)\leq\f{e^{-\lambda \Delta\tau}(e\lambda \Delta\tau)^{k}}{k^k}$, for $k>\lambda \Delta \tau$ \cite{mitzenmacher2017probability}.

\noindent Therefore, from \eqref{eq:Kbound}, as $K \rightarrow \infty$, we have $\f{\l(e\lambda\Delta \tau\r)^{K+1}}{(K+1)^{K+1}} \rightarrow 0$,
    resulting in no loss of information. For a given $\epsilon > 0$, we can choose $K$ such that
    the error $\l| g(\bfz, \Delta\tau) - g(\bfz, \Delta \tau, K) \r| < \epsilon$.
    This can be achieved by enforcing
    \EQA
    \label{eq:K_Oh}
    \f{\l(e\lambda\Delta \tau\r)^{K+1}}{(K+1)^{K+1}} \le \frac{\epsilon ~{2\pi \sigma_{\myx}\sigma_{\myy}\Delta \tau\s{1-\rho^2}}}{e^{-(r+\lambda)\Delta \tau}}.
    \ENA
  It is straightforward to see that, if $\epsilon = \Ocal((\Delta \tau)^{2})$, then $K = \Ocal(\ln(1/\Delta \tau))$, as $\Delta \tau \rightarrow 0$. In summary,  we can attain
    \EQ
    \label{eq:gerr}
    0\le g(\bfz, \Delta\tau) - g(\bfz, \Delta \tau, K) =   \Ocal((\Delta \tau)^{2}), \quad \text{ if   $K = \Ocal(\ln(1/\Delta \tau))$}.
    \EN

%% file: section_4_Numerical_methods.tex
\section{Numerical methods}
\label{section:num}
A common approach to handling the constraint posed by variational inequalities is to explicitly determine the optimal decision between immediate exercise and holding the contract for potential future exercise \cite{tavella2000, forsyth_2002a, kim1990analytic}. We define $\{\tau_m\}$, $m = 0, \ldots, M$, as an equally spaced partition of $[0, T]$, where $\tau_{m} = m\Delta \tau$ and $\Delta \tau = T/M$. We denote by $u(\cdot) \equiv u(x, y, \tau)$ the continuation value of the option. For a fixed $\tau_{m+1} < T$, the solution to the variational inequality \eqref{eq:F} at $(x, y, \tau_{m+1}) \in \Omega_{\myin}$, can be approximated by explicitly handling the
constraint as follows
\EQA
\label{eq:Bermudan}
v(x,y,\tau_{m+1})\simeq \max\{u(x, y, \tau_{m+1}),\hat{v}(x, y)\}, \qquad (x,y) \in  \myDin.
\ENA
Here, the continuation value $u(\cdot)$ is given by the solution of the 2-D PIDE
of the form \eqref{eq:2dPIDEs*}, i.e.\
\EQ
\label{eq:2dPIDEs}
\partial u/\partial \tau -\Lcal u-\Jcal u = 0, \qquad (x,y,\tau) \in  \Rbb^2 \times (\tau_m, \tau_{m+1}],
\EN
subject to  the initial condition at time $\tau_m$ given by a function $\tilde{v}(x, y, \tau_m)$, where
\EQ
\tilde{v}(x, y, \tau_m) =
\left\{
\begin{array}{lllll}
v(x, y, \tau_m) &\text{ satisfies \eqref{eq:VI_log} }  & \quad (x,y) \in \myDin,
\\
v_{\myout}(x, y, \tau_m) &\text{ satisfies \eqref{eq:inf_boundary_log} }  & \quad (x,y) \in \myDinc.
\end{array}
\right.
\EN
The solution $u(x,y,\tau_{m+1})$ for $(x, y) \in \myDin$ can be represented as the convolution integral of the Green's function $g(\cdot,\Delta \tau)$ and the initial condition $\tilde{v}(\cdot, \tau_m)$ as follows \cite{garronigreenfunctionssecond92, Duffy2015}
\EQ
u(x,y,\tau_{m+1})= \iint_{\Rbb^2}g\left(x- x',y- y',   \Delta \tau\right)\tilde{v}(x',y', \taus) dx'dy',\qquad
(x,y) \in  \myDin.
\label{eq:bkinteg}
\EN
The solution $u(x,y,\tau_{m+1})$ for $(x, y) \in  \myDinc$ is given by the boundary condition
\eqref{eq:inf_boundary_log}. In the analysis below, we focus on the convolution integral
\eqref{eq:bkinteg}.

\subsection{Computational domain}
For computational purposes, we truncate the infinite region of integration of \eqref{eq:bkinteg}
to the finite region $\myD^{\dagger}$ defined as follows
\EQ
\label{eq:truncate_region}
\myD^{\dagger}\equiv [x_{\min}^{\dagger}, x_{\max}^{\dagger}] \times [y_{\min}^{\dagger}, y_{\max}^{\dagger}].
\EN
Here,
where, for $z\in\l\{x,y\r\}$, $z^{\dagger}_{\min}< z_{\min}<0<z_{\max}< z^{\dagger}_{\max}$ and $|z^{\dagger}_{\min}|$ and $z^{\dagger}_{\max}$ are sufficiently large.
This results in the approximation for the continuation value
\EQ
\label{eq:integral_truncated}
u(x,y, \tau_{m+1}) \simeq \iint_{\myD^{\dagger}}g\left(x- x',y- y',   \Delta \tau\right)\tilde{v}( x',y', \taus) dx'dy',\qquad (x,y) \in\myDin.
\EN

We note that,  in approximating the above truncated 2-D convolution integral \eqref{eq:integral_truncated}
over the  finite integration domain $\myD^{\dagger}$, it is also necessary to obtain values of the Green's function $g(\cdot, \cdot,\Delta \tau)$ at spatial points $(x- x', y- y')$
which fall outside $\myD^{\dagger}$. 
More specifically, it is straightforward to see that, with $(x, y) \in \myD$ and $(x', y') \in \myD^{\dagger}$,
the point $(x- x', y- y') \in \myD^{\ddagger} \supset \myD^{\dagger}$ defined as follows
\EQ
\label{eq:extra_dom}
\myD^{\ddagger} =  [x_{\mymin}^{\ddagger}, x_{\mymax}^{\ddagger}] \times [y_{\mymin}^{\ddagger}, y_{\mymax}^{\ddagger}],
\qquad  z_{\mymin}^{\ddagger} = z_{\mymin} - z_{\mymax}^{\dagger},
~z_{\mymax}^{\ddagger} = z_{\mymax} - z_{\mymin}^{\dagger},~
 \text{ for } z \in \{x, y\}.
\EN
We emphasize that computing the solutions for $(x, y) \in \myD_{\myout}^{\dagger} =   \myD^{\dagger} \setminus \myD_{\myin}$ is not necessary, nor are they required for our convergence analysis. The primary purpose of $\myD_{\myout}^{\dagger}$ is to ensure the well-definedness of the Green's function $g(\cdot)$ used in the convolution integral \eqref{eq:integral_truncated}.

\begin{figure}[H]
\begin{minipage}{0.61\linewidth}
We now have a finite computational domain, denoted by $\Omega$, and its sub-domains defined  as follows
\begin{linenomath}
\begin{align}
\label{eq:sub_domain_whole*}
\Omega& = [x_{\min}^{\dagger}, x_{\max}^{\dagger}]  \times [y_{\min}^{\dagger}, y_{\max}^{\dagger}]   \times [0, T] \equiv \myD^{\dagger} \times [0, T],
\nonumber
\\
\Omega_{\tau_0}
& = [x_{\min}^{\dagger}, x_{\max}^{\dagger}] \times [y_{\min}^{\dagger}, y_{\max}^{\dagger}] \times  \{0\}
\equiv \myD^{\dagger} \times  \{0\},
\nonumber
\\
\Omega_{\myin}  &=(x_{\mymin}, x_{\mymax}) \times(y_{\mymin}, y_{\mymax}) \times (0,T]
\equiv \myD_{\myin} \times (0,T],
\\
\Omega_{\myout}
&=  \Omega \setminus \Omega_{\myin} \setminus  \Omega_{\tau_0}
\equiv \myD_{\myout} \times [0, T], \text{ where }
\myD_{\myout} = \myD^{\dagger} \setminus \myD_{\myin}.
\nonumber
\end{align}
\end{linenomath}
Here, $\myD_{\myin}$ and $\myD^{\dagger}$ are respectively defined in \eqref{eq:regG} and \eqref{eq:truncate_region}.
An illustration of the spatial computational sub-domains corresponding each $\tau \in (0, T]$
is given in Figure~\ref{fig:comp_domain}.
We note that $\myD_{\myout} = \myD^{\dagger} \setminus \myD_{\myin}$ and
$\myD_{\myout}^{\dagger} = \myD^{\ddagger} \setminus \myD^{\dagger}$,
where region $\myD^{\ddagger}$ is defined in \eqref{eq:extra_dom}.
\end{minipage}
\hspace*{-0.5cm}
\begin{minipage}{0.45\linewidth}
\begin{center}
\begin{tikzpicture}[scale=0.4]
   \draw [dotted, line width=0.5pt] [stealth-stealth] (-4,2) -- (12,2);
   \draw [dotted, line width=0.5pt][stealth-stealth](4,-5.5) --(4,9.5);
   \draw [thick]  (1,5) -- (7,5);
   \draw [thick]  (-1,6.5) -- (9,6.5);
   \draw [dashed, line width=0.5pt]   (-3,8.5) -- (11,8.5);

   \draw [thick]  (1,-1) -- (7,-1);
  \draw [thick]  (-1,-2.5) -- (9,-2.5);
  \draw [dashed, line width=0.5pt]   (-3,-4.5) -- (11,-4.5);

   \draw [thick]  (1,-1) -- (1,5);
   \draw [thick]  (-1,-2.5) -- (-1,6.5);
   \draw [dashed, line width=0.5pt]   (-3,8.5) --  (-3,-4.5);

   \draw [thick](7,-1) --(7,5);
   \draw [thick](9,-2.5) --(9,6.5);
   \draw [dashed, line width=0.5pt]   (11,8.5) -- (11,-4.5);

   \node at (4,2) {\scalebox{0.9}{$\myD_{\myin}$}};

   \node at (7.5,7.5) {\scalebox{0.9}{$\myD_{\myout}^{\dagger}$}};
   \node at (7.5,5.8) {\scalebox{0.9}{$\myD_{\myout}$}};

   \node [below] at (1.8,2) {\scalebox{0.8}{$x_{\mymin}$}};
   \node [below] at (-0.2,3.5) {\scalebox{0.8} {$x_{\mymin}^{\dagger}$}};
   \node [below] at (-2.2,3.5) {\scalebox{0.8} {$x_{\mymin}^{\ddagger}$}};
      \node [below] at (6.3,2) {\scalebox{0.8} {$x_{\mymax}$}};
   \node [below] at (8.3,3.5) {\scalebox{0.8} {$x_{\mymax}^{\dagger}$}};
   \node [below] at (10.3,3.5) {\scalebox{0.8} {$x_{\mymax}^{\ddagger}$}};
   \node [above] at (3.8,-1) {\scalebox{0.8}{$y_{\min}$}};
   \node [above] at (5,-2.5) {\scalebox{0.8}{$y_{\min}^{\dagger}$}};
   \node [above] at (5,-4.5) {\scalebox{0.8}{$y_{\min}^{\ddagger}$}};

   \node [below] at (3.8,5) {\scalebox{0.8}{$y_{\mymax}$}};
   \node [below] at (5,6.8) {\scalebox{0.8}{$y_{\mymax}^{\dagger}$}};
   \node [below] at (5,8.8) {\scalebox{0.8}{$y_{\mymax}^{\ddagger}$}};
\end{tikzpicture}
\end{center}
\vspace*{-0.5cm}
\caption{Spatial computational sub-domain at each $\tau \in [0, T]$,
$\myD_{\myout}^{\dagger} = \myD^{\ddagger} \setminus \myD^{\dagger}$.
}
\label{fig:comp_domain}
\end{minipage}
\end{figure}

Without loss of generality, for convenience, we assume that  $|z_{\min}|$ and $z_{\max}$, for $z\in\l\{x,y\r\}$, are chosen sufficiently large so that
\EQA
\label{eq:w_choice_green_jump_form}
z^{\dagger}_{\min} = z_{\min} - \frac{z_{\max} - z_{\min}}{2},
~~~\text{and}~~~
z^{\dagger}_{\max} =  z_{\max} + \frac{z_{\max} - z_{\min}}{2}.
\ENA
With \eqref{eq:w_choice_green_jump_form} in mind, recalling $z_{\min}^{\ddagger}$ and $z_{\max}^{\ddagger}$, for $z\in\{x,y\}$ as defined
in \eqref{eq:extra_dom} gives
\EQ
\label{eq:w_choice_green_jump_form_dd}
z_{\min}^{\ddagger} = z^{\dagger}_{\min} - z_{\max} =  -\frac{3}{2}\l(z_{\max} - z_{\min}\r),
~~~\text{and}~~~
z_{\max}^{\ddagger} = z^{\dagger}_{\max} - z_{\min}  = \frac{3}{2}\l(z_{\max} - z_{\min}\r).
\EN

\subsection{Discretization}
We denote by $N$ (resp.\ $N^{\dagger}$ and $N^{\ddagger}$ ) the number of intervals of a uniform partition of $[x_{\mymin}, x_{\mymax}]$
(resp.\ $[x_{\mymin}^{\dagger}, x_{\mymax}^{\dagger}]$ and $[x_{\mymin}^{\ddagger}, x_{\mymax}^{\ddagger}]$).
For convenience, we typically choose $N^{\dagger} = 2N$ and $N^{\ddagger} = 3N$ so that only one set of $z$-coordinates is needed. Also, let $P = x_{\mymax} - x_{\mymin}$, $P_x^{\dagger} = x^{\dagger}_{\mymax} - x^{\dagger}_{\mymin}$, and $P_x^{\ddagger} = x^{\ddagger}_{\mymax} - x^{\ddagger}_{\mymin}$. We define $\Delta x = \frac{P_x}{N} = \frac{P_x^{\dagger}}{N^{\dagger}} =\frac{P_x^{\ddagger}}{N^{\ddagger}}$. We use an equally spaced partition in the $x$-direction, denoted by $\{x_{n}\}$, and is defined as follows
\EQA
\label{eq:grid_x}
    x_{n}& = &{\hat{x}}_{0} + n\Delta x,
    ~~
    n = -N^{\ddagger}/2, \ldots, N^{\ddagger}/2, ~~\text{where}~~
    \nonumber
    \\
    \Delta x& =& P_x/N ~=~ P_x^{\dagger}/N^{\dagger}=P_x^{\ddagger}/N^{\ddagger}, ~~\text{and}~~
    \\
    \nonumber
    \hat{x}_{0} &~=~& ( x_{\mymin} + x_{\mymax})/2 ~=~ (x^{\dagger}_{\mymin} + x^{\dagger}_{\mymax})/2~=~ (x^{\ddagger}_{\mymin} + x^{\ddagger}_{\mymax})/2.
    \nonumber
\ENA
Similarly, for the $y$-dimension, with  $J^{\dagger} = 2J$, $J^{\ddagger} = 3J$, $P_{y} = y_{\mymax} - y_{\mymin}$, $P_y^{\dagger} = y^{\dagger}_{\mymax} - y^{\dagger}_{\mymin}$, and $P_y^{\ddagger} = y^{\ddagger}_{\mymax} - y^{\ddagger}_{\mymin}$, we denote by $\{y_{j}\}$, an equally spaced partition in the $y$-direction defined as follows
\EQA
\label{eq:grid_y}
    y_{j}& = &{\hat{y}}_{0} + j\Delta y,
    ~~
    j = -J^{\ddagger}/2, \ldots, J^{\ddagger}/2, ~~\text{where}~~
    \nonumber
    \\
    \Delta y& =& P_y/J ~=~ P_y^{\dagger}/J^{\dagger}=P_y^{\ddagger}/J^{\ddagger}, ~~\text{and}~~
    \\
    \nonumber
    \hat{y}_{0} &~=~& ( y_{\mymin} + y_{\mymax})/2 ~=~ (y^{\dagger}_{\mymin} + y^{\dagger}_{\mymax})/2~=~ (y^{\ddagger}_{\mymin} + y^{\ddagger}_{\mymax})/2.
    \nonumber
\ENA
We use the previously defined uniform partition
$\{\tau_m\}$, $m = 0, \ldots, M$, with $\tau_{m} = m\Delta \tau = mT/M$.\footnote{While it is straightforward to generalize the numerical method to non-uniform partitioning of the $\tau$-dimension, to prove convergence, uniform partitioning suffices.}


For convenience, we let $\mathbb{M}=\l\{0, \ldots M-1\r\}$ and we also define the following index sets:
\begin{linenomath}
\postdisplaypenalty=0
\begin{alignat}{8}
\label{index_sets}
&\mathbb{N} &&= \l\{-N/2+1, \ldots N/2-1\r\}, \quad && \mathbb{N}^{\dagger} &&= \l\{-N, \ldots N\r\}, \quad &&&\mathbb{N}^{\ddagger} &&= \l\{-3N/2+1, \ldots 3N/2-1\r\},
\nonumber
\\
&\mathbb{J} &&= \l\{-J/2+1, \ldots J/2-1\r\}, \quad &&\mathbb{J}^{\dagger} &&= \l\{-J, \ldots J\r\}, \quad &&&\mathbb{J}^{\ddagger} &&= \l\{-3J/2+1, \ldots, 3J/2-1 \r\}.
\end{alignat}
\end{linenomath}
With $n \in \mathbb{N}^{\dagger}$, $j \in \mathbb{J}^{\dagger}$, and $m \in \{0, \ldots, M\}$,
we denote by $v_{n, j}^m$ (resp.\ $u_{n, j}^m$)
a numerical approximation to the exact solution $v(x_n, y_j, \tau_m)$ (resp.\ $u(x_n, y_j, \tau_m)$) at the reference node $(x_n, y_j, \tau_m) = {\bf{x}}_{n, j}^{m}$. 
For $m \in \mathbb{M}$, nodes ${\bf{x}}_{n, j}^{m+1}$ having
$n \in \mathbb{N} \text{ and } j \in \mathbb{J}$ are in $\Omega_{\myin}$.
Those with either $n \in \mathbb{N}^{\dagger}\setminus \mathbb{N}$  or
$j \in \mathbb{J}^{\dagger}\setminus \mathbb{J}$ are in $\Omega_{\myout}$.
For double summations, we adopt the short-hand notation:  $\mysum_{d\in\mathbb{D}}^{q\in \mathbb{Q}}(\cdot):=\sum_{q\in\mathbb{Q}}\sum_{d\in\mathbb{D}}(\cdot)$,
unless otherwise noted. Lastly, it's important to note that references to indices
$n \in \mathbb{N}^{\ddagger}\setminus \mathbb{N}^{\dagger}$ or $j \in \mathbb{J}^{\ddagger}\setminus \mathbb{J}^{\dagger}$
pertain to points within $\myD^{\dagger}_{\myout} = \myD^{\ddagger} - \myD^{\dagger}$.
As noted earlier, no numerical solutions are required for these points.

\subsection{Numerical schemes}
\label{sec:NS}
For $(x_{n},y_{j},\tau_0)\in \Omega_{\tau_0}$, we impose the initial condition 
\EQA
\label{eq:tau0i}
v_{n,j}^0 &=& \hat{v}_{n,j},
\quad n \in \mathbb{N}^{\dagger}  ,~  j \in \mathbb{J}^{\dagger}.
\ENA\
For $(x_{n},y_{j},\tau_{m+1})\in \Omega_{{\myout}}$, we impose the boundary condition
\eqref{eq:vxymax} as follow
\EQA
\label{eq:outi}
v_{n,j}^{m+1} = \hat{v}_{n,j} e^{-r\tau_{m+1}},
\quad
n \in \mathbb{N}^{\dagger}\setminus \mathbb{N} \text{ or } j \in \mathbb{J}^{\dagger}\setminus \mathbb{J}.
\ENA
For subsequent use, we adopt the following notational convention:  (i) $x_{n-l} \equiv x_n-x_l = (n -l) \Delta x$, for $n\in\mathbb{N}$ and $l\in \mathbb{N}^{\dagger}$,  and (ii) $y_{j-d} \equiv y_j-y_d = (j-d)\Delta y$, for $j\in \mathbb{J}$ and $d\in\mathbb{J}^{\dagger}$. In addition, we denote by $g_{n-l,j-d} \equiv g_{n-l,j-d}(\Delta \tau, K)$ an approximation to $g\l(x_{n-l},y_{j-d}, \Delta\tau\r)$ using the first {\apnum{$(K+1)$}} terms of the infinite series representation in Corollary~\ref{cor:twodis}. {\apnum{When the role of $\Delta \tau$
is important, we explicitly write $g_{n-l,j-d}(\Delta \tau)$, omitting $K$ for brevity.}}

The continuation value at node $(x_{n},y_{j},\tau_{m+1})\in \Omega_{{\myin}}$
is approximated through the 2-D convolution integral \eqref{eq:integral_truncated} using a 2-D composite trapezoidal rule. This approximation, denoted by $u_{n,j}^{m+1}$, is computed as follows:
\EQ
\label{eq:dissum}
u_{n,j}^{m+1} = \Delta x \Delta y \mysum_{l\in\mathbb{N}^{\dagger}}^{d\in\mathbb{J}^{\dagger}}\varphi_{l,d}~g_{n-l,j-d} ~v^{m}_{l,d},
\quad n\in\mathbb{N},~j\in \mathbb{J}.
\EN
Here, the coefficients $\varphi_{l,d}$ in \eqref{eq:dissum}, where $l\in \mathbb{N}^{\dagger}$ and $d\in\mathbb{J}^{\dagger}$, are the weights of the 2-D composite trapezoidal rule.
Finally, the discrete solution $v_{n,j}^{m+1}$ is computed as follow
\EQ
\label{eq:scheme}
v_{n,j}^{m+1} =\max\{u_{n,j}^{m+1},~\hat{v}_{n,j}\}=\max \bigg\{\Delta x \Delta y \mysum_{l\in\mathbb{N}^{\dagger}}^{d\in\mathbb{J}^{\dagger}}\varphi_{l,d}~g_{n-l,j-d}~v^{m}_{l,d},~\hat{v}_{n,j} \bigg\},
\quad
n\in\mathbb{N},~j\in \mathbb{J}.
\EN
\begin{remark}[Monotonicity]
\label{rm:mono}
We note that the Green's function $g\l(x_{n-l},y_{j-d}, \Delta\tau\r)$, as given by the infinite series in Corollary~\ref{cor:twodis},
is defined and non-negative for all $n \in \mathbb{N}$, $l \in \mathbb{N}^{\dagger}$,
$j \in \mathbb{J}$, $d \in \mathbb{J}^{\dagger}$. Therefore, scheme \eqref{eq:dissum}-\eqref{eq:scheme} is monotone. We highlight that for computational purposes, $g\l(x_{n-l},y_{j-d}, \Delta\tau\r)$ is truncated to $g_{n-l, j - d}(\Delta \tau, K)$. However, since each term of the series is non-negative, this truncation ensures no loss of monotonicity, which is a key advantage of the proposed approach.
\end{remark}

{\apnum{
\begin{remark}[Rescaled weights and convention]
\label{rem:rescaled_kernel}
In the scheme \eqref{eq:dissum}, the weights $g_{n-l,j-d}(\Delta \tau)$ are multiplied by the grid area $\Delta x\,\Delta y$. As $\Delta \tau \to 0$, the Green's function
$g(\cdot, \Delta \tau)$ approaches a Dirac delta function, becoming increasingly peaked and unbounded. However, as proved subsequently, with $\Delta x \,\Delta y$
absorbed into the numerators of the terms in the (truncated) series representation
of $g_{n-l,j-d}(\Delta \tau)$, the resulting rescaled weights remain bounded.

To formalize this, we define the rescaled weights of our scheme as follows:
\begin{equation}
\label{eq:widetilde_g}
  \widetilde{g}_{n-l,j-d}(\Delta \tau)
  \;:=\;
  \Delta x\,\Delta y\,  \odot g_{n-l,j-d}(\Delta \tau), \quad
  n \in \mathbb{N},~ l \in \mathbb{N}^{\dagger},~
j \in \mathbb{J}, \text{and } d \in \mathbb{J}^{\dagger}.
\end{equation}
Here, $\odot$ indicates that $\Delta x\,\Delta y$ is effectively absorbed into the numerators of the terms in the (truncated) series representation of $g_{n-l,j-d}(\Delta \tau)$, ensuring that $\widetilde{g}_{n-l,j-d}(\Delta\tau)$ remains bounded as $\Delta \tau \to 0$.

For subsequent use, we also define the infinite series counterpart
$\widetilde{g}(x_{n-l},y_{j-d},\Delta \tau)$ as follows:
\begin{equation}
\label{eq:widetilde_g_full}
  \widetilde{g}(x_{n-l},y_{j-d}, \Delta \tau)
  \;:=\;
  \Delta x\,\Delta y\,  \odot g(x_{n-l},y_{j-d}, \Delta \tau), \quad
  n \in \mathbb{N},~ l \in \mathbb{N}^{\dagger},~
j \in \mathbb{J}, \text{and } d \in \mathbb{J}^{\dagger}.
\end{equation}
A formal proof of the boundedness of both $\widetilde{g}(x_{n-l},y_{j-d},\Delta \tau)$
and its truncated counterpart $\widetilde{g}_{n-l,j-d}(\Delta \tau)$ is provided in Lemma~\ref{lemma:bounded_g_scaled}.

\medskip
\noindent
\underline{Convention:}
For the rest of the paper, we adopt the convention of continuing to write
$\Delta x\,\Delta y\, g_{n-l,j-d}(\Delta \tau)$
and
$\Delta x\,\Delta y\,
  \mysum_{l\in\mathbb{N}^{\dagger}}^{d\in\mathbb{J}^{\dagger}}
  (\cdot)~g_{n-l,j-d}(\Delta\tau)~(\cdot)$
in our scheme, implementation descriptions, and subsequent analysis.
These expressions should respectively be understood as shorthand for
$\widetilde{g}_{n-l,j-d}(\Delta \tau)$
and
$\mysum_{l\in\mathbb{N}^{\dagger}}^{d\in\mathbb{J}^{\dagger}}
  (\cdot)~\widetilde{g}_{n-l,j-d}(\Delta \tau)~(\cdot)$,
where $\widetilde{g}_{n-l,j-d}(\Delta \tau)$ is the rescaled weight defined in \eqref{eq:widetilde_g}.
The same convention applies to expressions involving $g(x_{n-l},y_{j-d},\Delta \tau)$.
\end{remark}
}}

{\apnum{We note that the scheme in \eqref{eq:scheme} explicitly enforces the American option's early exercise constraint by taking $\max\{u_{n,j}^{m+1},\hat{v}_{n,j}\}$ at each timestep. While this approach is standard in computational finance and preserves the monotonicity of the scheme, it limits the global accuracy in time to first order due to the non-smooth $\max\{\cdot, \cdot\}$ operation.

A possible approach to achieving higher-order time accuracy is to adopt a penalty formulation (such as that in \cite{forsyth_2002a}), which leads to a system of nonlinear equations that must be solved iteratively at each timestep, typically using fixed-point or Newton-type methods \cite{huang:2010a}. Exploring these iterative methods in conjunction with the integral operator, while preserving the advantages of our convolution-based Green's function framework
in a fully 2-D jump-diffusion setting, is an interesting, but challenging, direction for future work. However, this lies beyond the scope of the present paper. We focus here on an explicit enforcement approach, which is simpler to implement and preserves the monotonicity of the scheme, thereby ensuring convergence to the viscosity solution.
}}

\subsection{Efficient implementation and complexity}
In this section, we discuss an efficient implementation of the 2-D discrete convolution \eqref{eq:dissum} using FFT.
Initially developed for 1-D problems \cite{zhang2023monotone} and extended to 2-D cases \cite{dang2024monotone}, this technique enables efficient computation of the convolution as a circular product.
Specifically, the goal of this technique is to represent \eqref{eq:dissum} for $ n \in \mathbb{N}$, $j \in \mathbb{J}$, $l \in \mathbb{N}^{\dagger}$ and $d \in \mathbb{J}^{\dagger}$ as a 2-D circular convolution product of the form
\EQA
\label{eq:cir_p}
   \U^{m+1} =  \Delta x \Delta y~ {\gb} \ast \V^{m}.
\ENA
Here, $\gb:= \gb(\Delta \tau, K_{\epsilon})$  is the first column of an associated circulant block matrix constructed from
$g_{n-l, j - d}(\Delta \tau, K_{\epsilon})$, where $K_{\epsilon}$ is sufficiently large,  reshaped into a $(3N- 1) \times (3J -1) $ matrix, while $\V^{m}$ is reshaped from an associated augmented block vector into a $(3N- 1) \times (3J -1) $ matrix. The notation $\ast$ denotes the circular convolution product. Full details on constructing $\gb$ and $\V^{m}$ can be found in~\cite{dang2024monotone}.

The resulting circular convolution product \eqref{eq:cir_p} can then be computed efficiently using FFT and inverse FFT (iFFT) as:
\EQA
\label{eq:fft_ifft}
  {\apnum{\U^{m+1} = {\text{FFT}}^{-1}\l\{\text{FFT}\l\{\V^{m}\r\} \circ \text{FFT}\l\{\Delta x \Delta y~ \gb\r\}\r\}}}.
\ENA
After computation, we discard components in $\U^{m+1}$ for indices $n \in \mathbb{N}^{\ddagger} \setminus \mathbb{N}$ or $j \in \mathbb{J}^{\ddagger} \setminus \mathbb{J}$, obtaining discrete solutions $u_{n, j}^{m+1}$ for $\Omega_{\myin}$.

{\apnum{
As explained in Remark~\ref{rem:rescaled_kernel}, the factor $\Delta x,\Delta y$ is incorporated into $g_{n-l,j-d}$, yielding $\widetilde{g}_{n-l,j-d}$.
Following our convention, we continue to write $\Delta x\,\Delta y\, \gb$ in~\eqref{eq:cir_p}-\eqref{eq:fft_ifft} for simplicity.
}}

The implementation outlined in \eqref{eq:fft_ifft} indicates that the {\apnum{rescaled weight array $\Delta x\,\Delta y\,\bf{g}$}} needs to be computed only once using the infinite series expression from Corollary~\ref{cor:twodis}, after which they can be reused across all time intervals. Specifically, for a given user-defined tolerance $\epsilon$, we use \eqref{eq:K_Oh} to determine a sufficiently large number of terms $K = K_{\epsilon}$ in the series representation \eqref{eq:g_proof_sum} for these weights. The resulting {\apnum{rescaled weight array $\Delta x\,\Delta y\,\gb := \Delta x\,\Delta y\, \gb(\Delta \tau, K_{\epsilon})$}} is then calculated. For the two-asset Merton jump-diffusion model, the  {\apnum{rescaled weight array  $\Delta x\,\Delta y\,\gb$}} needs only be computed once as per \eqref{eq:fft_ifft}, and can subsequently be reused across all time intervals. This step is detailed in Algorithm~\ref{alg:Gtilde}.

\begin{algorithm}[htb!]
\caption{
\label{alg:Gtilde}
Computation of the {\apnum{rescaled weight array  $\Delta x\,\Delta y\, \gb := \Delta x\,\Delta y\, \gb(\Delta \tau, K_{\epsilon})$}}; $\epsilon>0$ is a user-defined tolerance.}
\begin{algorithmic}[1]

\STATE set $k = K_{\epsilon} = 0$;

\STATE compute
$\texttt{test} = \frac{e^{-(r + \lambda)\Delta\tau}}{2\pi\sqrt{\det(\bfC)}}\,
\frac{(e\lambda \Delta\tau)^{k+1}}{(k+1)^{k+1}}$;

\STATE set
{\apnum{$\Delta x\, \Delta y\,g_{n-l, j - d}(\Delta \tau, K_{\epsilon}) = \Delta x\, \Delta y\,g_{0}\left(x_{n-l}, y_{j-d}, \Delta\tau\right)$}},  $n \in \mathbb{N}, \, j \in \mathbb{J}, \, l \in \mathbb{N}^{\dagger}, \, d \in \mathbb{J}^{\dagger}$;

\WHILE{$\texttt{test} \ge \epsilon$}
    \STATE set $k = k + 1$, and $K_{\epsilon} = k$;

    \STATE compute
     {\apnum{$\Delta x\, \Delta y \, g_k(x_{n-l}, y_{j-d}, \Delta\tau)$}}, $n \in \mathbb{N}, \, j \in \mathbb{J}, \, l \in \mathbb{N}^{\dagger}, \, d \in \mathbb{J}^{\dagger}$, given in Corollary~\ref{cor:twodis};

    \STATE set
    {\apnum{$\Delta x\, \Delta y \, g_{n-l, j - d}(\Delta \tau, K_{\epsilon}) = \Delta x \Delta y \, g_{n-l, j - d}(\Delta \tau, K_{\epsilon}) + \Delta x \Delta y \, g_k(x_{n-l}, y_{j-d}, \Delta\tau)$}}, $n \in \mathbb{N}, \, j \in \mathbb{J}, \, l \in \mathbb{N}^{\dagger}, \, d \in \mathbb{J}^{\dagger}$;

    \STATE {\myblue{
        compute
        $\texttt{test} = \frac{e^{-(r+\lambda)\Delta\tau}}{2\pi\sqrt{\det(\bfC)}}\,
        \frac{(e\lambda \Delta\tau)^{k+1}}{(k+1)^{k+1}}$;
    }}

\ENDWHILE

\STATE construct the weight array {\apnum{$\Delta x\, \Delta y\, \gb = \Delta x\, \Delta y \, \gb(\Delta \tau, K_{\epsilon})$ using $\Delta x\, \Delta y \, g_{n-l, j - d}(\Delta \tau, K_{\epsilon})$}}, $n \in \mathbb{N}$, $j \in \mathbb{J}$, $l \in \mathbb{N}^{\dagger}$, $d \in \mathbb{J}^{\dagger}$;

\STATE output $\gb$;
\end{algorithmic}
\end{algorithm}

Putting everything together, the proposed numerical scheme for the American options under two-asset Merton jump-diffusion model is presented in Algorithm \ref{alg:monotone} below.
\begin{algorithm}[htb!]
\caption{
A monotone numerical integration algorithm for pricing American options
under the two-asset Merton jump-diffusion model.
}

\begin{algorithmic}[1]
\label{alg:monotone}

\STATE
compute the {\apnum{rescaled weight array $\Delta x\, \Delta y\, \gb$}} using Algorithm~\ref{alg:Gtilde};

\STATE
\label{alg:initial}
initialize $v_{n, j}^{0} =\vh(x_n,y_j)$, $n\in\Nd$, $j\in \Jd$;

 \FOR{$m = 0, \ldots, M-1$}

    \STATE
    \label{alg:step1}
    compute intermediate values $\U^{m+1}$ using FFT as per \eqref{eq:fft_ifft};

    \STATE
    \label{alg:dis}
    obtain discrete solutions $\l[u_{n, j}^{m+1}\r]_{n\in \Nbb, j \in \Jbb}$ by discarding the components in $\U^{m+1}$ corresponding to indices $n \in \mathbb{N}^{\ddagger} \setminus \mathbb{N}$ or $j \in \mathbb{J}^{\ddagger} \setminus \mathbb{J}$;

    \STATE
    \label{alg:step4}
    set $v_{n,j}^{m+1}= \max\{u_{n,j}^{m+1},v_{n,j}^{0}\}$, $n\in\Nbb$ and  $j\in\Jbb$,
    where $u_{n,j}^{m+1}$ are from Line~\ref{alg:dis};
     \hfill $\Omega_{\myin}$

    \STATE
    \label{alg:step5}
    compute $v_{n,j}^{m+1}$, $n\in\Nbb^{\dagger}\setminus\Nbb$ or $j\in\Jbb^{\dagger}\setminus\Jbb$, using \eqref{eq:outi}; \hfill $\Omega_{\myout}$

\ENDFOR
\end{algorithmic}
\end{algorithm}

\begin{remark}[Complexity]
\label{rm:complexity}
Our algorithm involves, for $m=0,\ldots,M-1$, the following key steps:
\begin{itemize}

    \item Compute $u_{n,j}^{m+1}$, $n\in\Nbb^{\dagger}$, $j\in\Jbb^{\dagger} $ via the proposed 2-D FFT algorithm. The complexity of this step is {\zblue{$\Ocal\l(N^{\dagger}J^{\dagger}\log(N^{\dagger}J^{\dagger})\r)=\Ocal\l(NJ\log(NJ)\r)$, considering that $N^{\dagger}=2N$ and $J^{\dagger}=2J$}}.

    \item Finding the optimal control for each node $\x_{n,j}^{m+1}$ by directly comparing $u_{n,j}^{m+1}$ with the payoff requires $\Ocal(1)$ complexity. Thus, with a total of {\zblue{$NJ$}}  interior nodes, this gives a complexity {\zblue{$\Ocal(NJ)$}}.

    \item Therefore, the major cost of our algorithm is determined by the step of FFT algorithm. With {\zblue{$M$}} timesteps, the total complexity is {\zblue{$\Ocal(MNJ\log(NJ))$}}.
    \end{itemize}

\end{remark}

%% file: section_5_Convergence.tex
\section{Convergence to viscosity solution}
\label{sc:conv}
In this section, we demonstrate that, as a discretization paremeter approaches zero, our numerical scheme in the interior sub-domain $\Omega_{\myin}$ converges to the viscosity solution of the variational inequality \eqref{eq:F} in the sense of Definition~\ref{Def:viscosity_VIs}. To achieve this, we examine three critical properties: $\ell_\infty$-stability, consistency, and
monotonicity \cite{crandall_ishii_lions1992}.

\subsection{Error analysis}
\label{ssc:error}
To commence,  we shall identify errors arising in our numerical scheme and make assumptions needed for
subsequent proofs. In the discussion below, $\phi(\cdot)$ is a test function in $(\mathcal{B}\cap\mathcal{C}^{\infty})(\Rbb^{2}\times[0,T])$.
\begin{itemize}

\item  
Truncating the infinite region of integration $\Rbb^2$ of the convolution integral
\eqref{eq:bkinteg} between the Green's function $g(\cdot)$ and $\phi(\cdot)$
to $\myD^{\dagger}$ results in a boundary truncation error $\errorb$, where
\EQA
\label{eq:green_integral_truncated_e}
\errorb  = \iint_{\mathbb{R}^2 \setminus \myD^{\dagger}}
g(x - x', y - y', \dtau)~\phi(x', y', \cdot,\taus)~dx'~dy',
\quad (x, y) \in \myD_{\myin}.
\ENA
It has been established that for general jump diffusion models, such as those considered in this paper, the error bound $\errorb$ is bounded by \cite{Cont2005,cont2003financial}
\EQ
\label{eq:truncation}
        \left| \errorb \right| \le  C_1 \Delta \tau e^{-C_2 \l(P_{x}^{\dagger} \wedge P_{y}^{\dagger}\r)},
    \quad
    \text{$\forall (x,y) \in \myD_{\myin}$},
     \quad
     \text{$C_1, C_2 > 0$ independent of $\Delta \tau$, $P_{x}^{\dagger}$ and $P_{y}^{\dagger}$},
\EN
where $P_{\myx}^{\dagger}=x^{\dagger}_{\mymax}-x^{\dagger}_{\mymin}$, $P_{\myy}^{\dagger}=y^{\dagger}_{\mymax}-y^{\dagger}_{\mymin}$ and $a\wedge b = \min(a, b)$. For fixed $P_{\myx}^{\dagger}$ and $P_{\myy}^{\dagger}$, \eqref{eq:truncation} shows $\errorb\to 0$, as $\Delta \tau\to 0$. However, as typical required for showing consistency, one would need to ensure $\frac{\errorb}{\Delta \tau}\to 0$ as $\Delta\tau\to 0$. Therefore, from \eqref{eq:truncation}, we need $P_{\myx}^{\dagger}\to \infty$ and $P_{\myy}^{\dagger}\to \infty$ as $\Delta\tau\to 0$, which can be achieved by letting $P_{\myx}^{\dagger}=C_1'/\Delta\tau$ and $P_{\myy}^{\dagger}=C_2'/\Delta\tau$, for finite $C_1', C_2'>0$
independent of $\Delta \tau$.

\item The next source of error is identified in approximating the truncated 2-D convolution integral
\[
\iint_{\myD^{\dagger}}
g(x - x', y - y', \dtau)~\phi(x', y', \cdot,\taus)~dx'~dy',
\quad (x, y) \in \myD_{\myin}.
\]
by the composite trapezoidal rule:
$
\Delta x \Delta y \mysum_{l\in\mathbb{N}^{\dagger}}^{d\in\mathbb{J}^{\dagger}} \varphi_{l,d}~ g(x_{n-l}, y_{j-d}, \Delta \tau) ~\phi(x_l, y_d, \tau_m),
$
where $g(x_{n-l}, y_{j-d}, \Delta \tau)$ represents the exact Green's function.
We denote by $\errorc$ the numerical integration error associated with this approximation.
For a fixed integration domain  $\myD^{\dagger}$, due to the smoothness of the test function $\phi$,  we have that $\errorc = \mathcal{O}\l(\max(\Delta x, \Delta y)^2\r)$ as $\Delta x, \Delta y \to 0$.

\item Truncating the Green's function $g(x_{n-l}, y_{j-d}, \Delta \tau)$, which is expressed as an infinite series in \eqref{eq:g_proof_sum}, to $g_{n-l, j - d}(\Delta \tau, K)$  using only the first
    $(K+1)$ terms introduces a series truncation error, denoted by $\errorf$. As discussed previously in \eqref{eq:gerr}, with $K = \Ocal(\ln((\Delta \tau)^{-1}))$, then  $\errorf = \Ocal((\Delta \tau)^{2})$ as $\Delta \tau \to 0$.
\end{itemize}
Motivated by the above discussions, for convergence analysis, we make the assumption below about the discretization parameter.
\begin{assumption}
\label{as:dis_parameter}
We assume that there is a discretization parameter $h$
such that
\EQA
\label{eq:dis_parameter}
\Delta x=  C_1 h, \quad
\Delta y = C_2 h,\quad
\Delta \tau = C_3 h,
\quad
{P_x^{\dagger} = C'_1/h,}\quad
{P_y^{\dagger} = C'_2/h,}
\ENA
where the positive constants $C_1$, $C_2$, $C_3$, $C_1'$ and $C_2'$ are independent of $h$.
\end{assumption}
Under Assumption~\ref{as:dis_parameter}, and for a test function
$\phi(\cdot) \in \mathcal{B}(\Oinf)\cap\mathcal{C}^{\infty}(\Oinf)$,
we have
\EQ
\label{eq:err_h}
\errorb = \mathcal{O}(he^{-\frac{1}{h}}),
\quad
\errorc = \mathcal{O}(h^2),
\quad
\errorf  = \mathcal{O}(h^2).
\EN
It is also straightforward to ensure the theoretical requirement $P_x^{\dagger}, P_y^{\dagger} \to \infty$ as $h \to 0$.
For example, with $C'_2 = C'_2 = 1$ in \eqref{eq:dis_parameter},
we can quadruple $N^{\dagger}$ and $J^{\dagger}$ as we halve $h$.
We emphasize that, for practical purposes, if $P_x^{\dagger}$ and $ P_y^{\dagger}$ are chosen sufficiently large,
both can be kept constant for all $\Delta \tau$ refinement levels (as we let $\Delta \tau  \to 0$).
The effectiveness of this practical approach is demonstrated through numerical experiments
in Section~\ref{sec:num_test}. 
Finally, we note that, the total complexity of the proposed algorithm, as outlined in Remark~\ref{rm:complexity}, is $\Ocal(1/h^3\cdot\log(1/h))$.

{\apnum{
We now present a lemma on the boundedness of the rescaled weights $\widetilde{g}_{n-l,j-d}(\Delta \tau)$ defined in \eqref{eq:widetilde_g}.
\begin{lemma}[Boundedness of the rescaled weights]
\label{lemma:bounded_g_scaled}
Let $g(x,y,\Delta\tau)$ be the two-asset Merton jump-diffusion Green's function
from Corollary~\ref{cor:twodis}, and recall the rescaled weights
$\widetilde{g}_{n-l,j-d}(\Delta \tau)$ and its infinite series counterpart
$\widetilde{g}(x_{n-l},y_{j-d}, \Delta \tau)$, respectively defined
in \eqref{eq:widetilde_g} and \eqref{eq:widetilde_g_full}.

Under Assumption~\ref{as:dis_parameter}, there exists a finite constant $C'>0$, independent of the discretization parameter  $h$, such that for all sufficiently small $h$,
\[
  0 \le \widetilde{g}(x_{n-l},y_{j-d}, \Delta \tau) \le C',
  \quad
  \text{for all }
  n \in \mathbb{N},~ l \in \mathbb{N}^{\dagger},~
j \in \mathbb{J}, \text{and } d \in \mathbb{J}^{\dagger}.
\]
Consequently, the rescaled weights $\widetilde{g}_{n-l,j-d}(\Delta \tau)$ satisfy
\[
0 \le \widetilde{g}_{n-l,j-d}(\Delta \tau) \le C'.
\]
\end{lemma}
A proof of Lemma~\ref{lemma:bounded_g_scaled} is provided in  Appendix~\ref{app:bounded_g_scaled}.
We reiterate that, as explained in Remark~\ref{rem:rescaled_kernel}, we adopt the convention of writing $\Delta x\,\Delta y\,g_{n-l,j-d}(\Delta \tau)$ and $\Delta x\,\Delta y\,
  \mysum_{l\in\mathbb{N}^{\dagger}}^{d\in\mathbb{J}^{\dagger}}
  (\cdot)~g_{n-l,j-d}(\Delta\tau)~(\cdot)$ even though the factor $\Delta x\,\Delta y$ can be interpreted as folded into $g_{n-l,j-d}(\Delta \tau)$. Under this convention, the discrete rescaled weights, as well as its full series counter parts, remain bounded as $h \to 0$ by Lemma~\ref{lemma:bounded_g_scaled}.

%
}}

For subsequent use,  we present a result about $g(x_{n-l}, y_{j-d}, \Delta \tau)$
in the form of a lemma.
\begin{lemma}
\label{lemma:rootuni}
Suppose the discretization parameter $h$ satisfies \eqref{eq:dis_parameter}.
For sufficiently small $h$, we have
\EQ
\label{eq:geps}
\Delta x \Delta y \mysum_{l\in\mathbb{N}^{\dagger}}^{d\in\mathbb{J}^{\dagger}} ~ \varphi_{l, d}~
g_{n-l, j-d}\le
e^{-r \Delta \tau} \!+\! \mathcal{O}(h^2)
\le
1 + \epsilon_g \frac{\Delta \tau}{T} \le e^{\epsilon_g\frac{\Delta \tau}{T}}.
\quad n \in \mathbb{N}, ~j \in \mathbb{J},
\EN
where $\epsilon_g = C h$ with $C>0$ being a bounded constant independently of $h$.
\end{lemma}
\begin{proof}[Proof of Lemma~\ref{lemma:rootuni}]
In this proof, we let $C$ be a generic positive constant independent of $h$, which may take different values
from line to line. We note that  $g_{n-l,j-d} \equiv g_{n-l,j-d}(\Delta \tau, K)$ is an approximation to $g\l(x_{n-l},y_{j-d}, \Delta\tau\r)$ using the first $(K+1)$ terms of the infinite series.
Recall that  $G(\eta_x,\eta_y, \Delta \tau)= \iint_{\Rbb^2}e^{-i(\eta_x x+\eta_y y)}g(x,y, \Delta \tau)~\md x\md y$,
and also, by  \eqref{eq:G_closed},  $G(\eta_x,\eta_y, \Delta \tau)= \exp(\Psi(\eta_x,\eta_y) \Delta \tau)$.
Hence, setting $\eta_x=\eta_y=0$ in the above gives
\EQ
\label{eq:gdb}
\iint_{\Rbb^2} g(x,y, \Delta \tau)~\md x\md y = \exp(\Psi(0,0) \Delta \tau) = e^{-r \Delta \tau},
\quad \forall (x, y) \in \Rbb^2
\EN
As $h \to 0$,
we have:
$0\le \Delta x \Delta y \mysum_{l\in\mathbb{N}^{\dagger}}^{d\in\mathbb{J}^{\dagger}} ~ \varphi_{l, d}~
g_{n-l, j-d}  \overset{\text{(i)}}{\le} \Delta x \Delta y \mysum_{l\in\mathbb{N}^{\dagger}}^{d\in\mathbb{J}^{\dagger}} ~ \varphi_{l, d}~ g(x_{n-l}, y_{j-d}, \Delta \tau) = \ldots$
\begin{linenomath}
\postdisplaypenalty=0
\begin{align*}
\ldots\overset{\text{(ii)}}{=}\!\!
\iint_{\mathbb{R}^2}\!\!
g(x_n - x,y_j - y, \Delta \tau)\md x\md y
+ \mathcal{O}(he^{-\frac{1}{h}}) + \mathcal{O}(h^2)
\overset{\text{(iii)}}{=}
e^{-r \Delta \tau} \!+\! \mathcal{O}(h^2)\!
\overset{\text{(iv)}}{\le}
1 \!+\! C h \frac{\Delta \tau}{T}
\overset{\text{(v)}}{\le}
e^{\epsilon_g\frac{\Delta \tau}{T}}.
\label{eq:sumgG}
\end{align*}
\end{linenomath}
Here, (i) is due to the fact that all terms of the infinite series are non-negative, so are the weights
$\varphi_{l, d}$ of the composite trapezoidal rule; in (ii), as  previously discussed,
 the error  $\mathcal{O}(he^{-\frac{1}{h}})$ is due to the boundary truncation error \eqref{eq:truncation}, together with $P_x^{\dagger}, P_y^{\dagger} \sim \mathcal{O}(1/h)$ as $h \to 0$, as in \eqref{eq:dis_parameter}; the $ \mathcal{O}(h^2)$ error arises from the trapezoidal rule approximation of the double integral, as noted in \eqref{eq:err_h}; 
 (iii)~follows from \eqref{eq:gdb}; (iv)~is due to  $e^{-r \Delta \tau} \le 1$ and \eqref{eq:dis_parameter}.
Letting $\epsilon_g = C h$ gives (iv). This concludes the proof.
\end{proof}
\subsection{Stability}
Our scheme consists of the following equations: \eqref{eq:tau0i} for $\Omega_{\tau_0}$, \eqref{eq:outi}  for $\Omega_{\myout}$,
and finally \eqref{eq:scheme} for $\Omega_{\myin}$. We start by verifying $\ell_\infty$-stability of our scheme.
\begin{lemma}[$\ell_\infty$-stability]
\label{lemma:stability}
Suppose the discretization parameter $h$ satisfies \eqref{eq:dis_parameter}.
The scheme  \eqref{eq:tau0i}, \eqref{eq:outi}, and \eqref{eq:scheme} satisfies
the bound $\ds \sup_{h > 0} \left\| v^{m} \right\|_{\infty} < \infty$
for all $m = 0, \ldots, M$, as the discretization parameter $h \to 0$.
Here, we have $\left\| v^{m} \right\|_{\infty} = \max_{n, j} |v_{n, j}^{m}|$,
$n \in \mathbb{N}^{\dagger}$ and $j \in \mathbb{J}^{\dagger}$.
\end{lemma}

\begin{proof}[Proof of Lemma~\ref{lemma:stability}]
Since the function $\hat{v}(\cdot)$ is a bounded, for any fixed $h >0$, we have
\EQ
\left\| v^{0} \right\|_{\infty} = \max_{n, j} |v_{n, j}^{0}| =  \max_{n, j} |\hat{v}_{n, j}| <\infty,
\quad n \in \mathbb{N}^{\dagger},~j \in \mathbb{J}^{\dagger}
\label{eq:bdbounded}
\EN
Hence,  $\sup_{h > 0} \left\| v^{0} \right\|_{\infty} < \infty$.
Motivated by this observation and Lemma~\ref{lemma:rootuni},
to demonstrate $\ell_\infty$-stability of our scheme,
we will show that, for a fixed $h > 0$, at any $(x_n, y_j, \tau_m)$, we have
\EQ
\label{eq:key}
|v_{n, j}^{m}| < e^{m \frac{\Delta \tau}{T} \epsilon_g}  \left\| v^{0} \right\|_{\infty},
\quad
\text{ where $\epsilon_g = C h$ from \eqref{eq:geps} },
\EN
from which, we obtain $\sup_{h > 0} \left\| v^{m} \right\|_{\infty} < \infty$  for all $m = 0, \ldots, M$, as wanted, noting $m \frac{\Delta\tau}{T} \le 1$.

For the rest of the proof, we will show the key inequality \eqref{eq:key} when $h > 0$ is fixed.
We will address $\ell_\infty$-stability for the boundary and interior sub-domains
(together with their respective initial conditions) separately, starting with
the boundary sub-domains $\Omega_{\tau_0}$ and $\Omega_{\myout}$.
It is straightforward to see that both \eqref{eq:tau0i} (for $\Omega_{\tau_0}$) and \eqref{eq:outi} (for $\Omega_{\myout}$)  satisfy \eqref{eq:key}, respectively due to  \eqref{eq:bdbounded} and the following
\EQA
\label{eq:stab_max}
\max_{n, j} |v_{n, j}^{m}|  = \max_{n, j} |\hat{v}_{n, j} e^{-r\tau_{m}} | \le
\max_{n, j} |\hat{v}_{n, j}| = \left\| v^{0} \right\|_{\infty}  < \infty,
\quad
n \in \mathbb{N}^{\dagger}\setminus \mathbb{N} \text{ or } j \in \mathbb{J}^{\dagger}\setminus \mathbb{J}.
\ENA
We now demonstrate the bound \eqref{eq:key} for $\Omega_{\myin}$ using induction on $m$, $m = 0, \ldots, M$.
For the base case, $m = 0$, the bound \eqref{eq:key} holds for all $n \in \mathbb{N}$
and $j \in \mathbb{J}$, which follows from \eqref{eq:bdbounded}.
Assume that \eqref{eq:key} holds for $n \in \mathbb{N}$, $j \in \mathbb{J}$,
and $m = m' \le M-1$, i.e.\
$|v_{n, j}^{m'}| < e^{m' \frac{\Delta \tau}{T} \epsilon_g}  \left\| v^{0} \right\|_{\infty}$,
for $n \in \mathbb{N}$ and $j \in \mathbb{J}$.
We now show that \eqref{eq:key} also holds for $m = m' + 1$.
The continuation value $u_{n,j}^{m'+1}$, for $n \in \mathbb{N}$ and $j \in \mathbb{J}$,
satisfies $|u_{n,j}^{m'+1}| \le \Delta x \Delta y \mysum_{l\in\mathbb{N}^{\dagger}}^{d\in\mathbb{J}^{\dagger}}\varphi_{l,d}~g_{n-l,j-d}~|v^{m'}_{l,d}|
\le \ldots$

\begin{linenomath}
\begin{align}
\ldots \overset{(i)}{\le}
e^{m'  \frac{\Delta \tau}{T} \epsilon_g} \left\| v^{0}\right\|_{\infty}
\Delta x \Delta y
\mysum_{l\in\mathbb{N}^{\dagger}}^{d\in\mathbb{J}^{\dagger}}\varphi_{l,d}~g_{n-l,j-d}  \,
\overset{\text{(ii)}}{\le}
e^{(m' +1) \frac{\Delta \tau}{T}\epsilon_g}
\left\| v^{0}\right\|_{\infty}.
\label{eq:vl}
\end{align}
\end{linenomath}
Here, (i) is due to induction hypothesis, and (ii) is due to Lemma~\ref{lemma:rootuni}.
Using \eqref{eq:vl} and \eqref{eq:bdbounded}, we have
\[
|v_{n,j}^{m'+1}| \le \max\{|u_{n,j}^{m'+1}|, |\hat{v}_{n,j}| \}
\le e^{(m' +1) \frac{\Delta \tau}{T}\epsilon_g}
\left\| v^{0}\right\|_{\infty}.
\]
This concludes the proof.
\end{proof}

\subsection{Consistency}
While equations \eqref{eq:tau0i}, \eqref{eq:outi}, and \eqref{eq:scheme} are convenient for computation, they are not well-suited for analysis. To verify consistency in the viscosity sense,
it is more convenient to rewrite them in a single equation that encompasses the interior and boundary
sub-domains. To this end, for grid point $(x_n, y_j, \tau_{m+1}) \in \Omega_{\myin}$, we define operator $\mathcal{C}_{n, j}^{m+1}(\cdot) \equiv
\mathcal{C}_{n, j}^{m+1} \bigg(h, v_{n, j}^{m+1},
\left\{v_{l,d}^{m}\right\}_{\subalign{l\in \Nd\\d\in \Jd}}
 \bigg)= \ldots$
\EQ
\label{eq:scheme_CD}
\begin{aligned}
\ldots=\min\bigg\{ \frac{1}{\Delta \tau}\bigg(v_{n,j}^{m+1}
-
 \Delta x \Delta y \mysum_{l \in \Nd}^{d \in \Jd} \varphi_{l, d}~
g_{n-l, j-d}~v^{m}_{l,d} \bigg),~ v_{n,j}^{m+1}-\hat{v}_{n,j}\bigg\}.
\end{aligned}
\EN
Using $\mathcal{C}_{n,j}^{m+1}(\cdot)$ defined
in \eqref{eq:scheme_CD}, our numerical scheme at  the reference node $\bfx = (x_n, y_j, \tau_{m+1}) \in \Omega$
can be rewritten in an equivalent form as follows
\EQA
\label{eq:scheme_GF}
0=
\mathcal{H}_{n,j}^{m+1}
\bigg(h, v_{n, j}^{m+1},
\left\{v_{l,d}^{m}\right\}_{\subalign{l\in \Nd\\d\in \Jd}}
 \bigg)
\equiv
\left\{
\begin{array}{lllllllllllll}
\mathcal{C}_{n,j}^{m+1}
\left(\cdot\right)
&
\quad {\bf{x}} \in \Omega{\myin},
\\
v_{n,j}^{m+1} - \hat{v}_{n,j} e^{-r\tau_{m+1}}
&
\quad {\bf{x}} \in \Omega_{\myout},
\\
v_{n,j}^{m+1} -\hat{v}_{n,j}
&
\quad {\bf{x}} \in \Omega_{\tau_0},
\end{array}
\right.
\ENA
where the sub-domains $ \Omega{\myin}$, $\Omega_{\myout}$ and  $\Omega_{\tau_0}$ are defined in \eqref{eq:sub_domain_whole*}.

To establish convergence of the numerical scheme to the viscosity solution
in $\Omega_{\myin}$, we first consider an intermediate result presented in
Lemma~\ref{lemma:ar} below.
\begin{lemma}
\label{lemma:ar}
Suppose the discretization parameter $h$ satisfies \eqref{eq:dis_parameter}.
Let $\phi$  be a test function in $(\mathcal{B}\cap\mathcal{C}^{\infty})(\Rbb^{2}\times[0,T])$.
For $\bfx_{n, j}^m = (x_n, y_j, \tau_m) \in \Omega_{\myin}$, where $n \in \mathbb{N}$, $j \in \mathbb{J}$, and $m \in \{0, \ldots, M\}$,
with $\phi_{n,j}^{m} = \phi(\bfx_{n, j}^m)$, and for sufficiently small $h$, we have
\begin{align}
\Delta x \Delta y
\mysum_{l\in\mathbb{N}^{\dagger}}^{d\in\mathbb{J}^{\dagger}}
    g_{n-l, j-d}~
    \phi_{l,d}^{m}
&=
\phi_{n,j}^{m}
+ \Delta \tau \left[ \mathcal{L} \phi + \mathcal{J} \phi\right]_{n, j}^{m} + \mathcal{O}( h^2).
\label{eq:error_analysis_smooth}
\end{align}
Here, $\left[\mathcal{L} \phi\right]_{n, j}^{m} = [\mathcal{L} \phi](\bfx_{n, j}^m)$ and
 $\left[\mathcal{J} \phi\right]_{n, j}^{m} = [\mathcal{J} \phi](\bfx_{n, j}^m)$.
\end{lemma}
\begin{proof}[Proof of Lemma~\ref{lemma:ar}]
Lemma~\ref{lemma:ar} can be proved using similar techniques in
\cite{lu2024semi}[Lemmas~5.2].
Starting from the discrete convolution on the left-hand-side (lhs) of \eqref{eq:error_analysis_smooth},
we need to recover an associated convolution integral of the form \eqref{eq:bkinteg}
which is posed on an infinite integration region.
Since for an arbitrary fixed $\tau_m$, $\phi(x, y, \tau_m)$ is not necessarily in $L_1(\Rbb^2)$, standard mollification techniques can be used to obtain a mollifier $\chi(x, y, \tau_m) \in L_1(\Rbb^2)$ which agrees with  $\phi(x, y, \tau_m)$ on $\mathbf{D}^{\dagger}$ \cite{Johnlee}, and has bounded derivatives up to second order across $\Rbb^2$. For brevity, instead of $\chi(x, y, \tau_{m})$, we will write $\chi(x, y)$.
Recalling different errors outlined in \eqref{eq:err_h}, 
we have
\begin{align}
\label{eq:error_smooth_b}
\Delta x \Delta y
\mysum_{l\in\mathbb{N}^{\dagger}}^{d\in\mathbb{N}^{\dagger}} \varphi_{l, d}~
    g_{n-l, j - d}~  \phi_{l,d}^{m}~
&\overset{\text{(i)}}{=}
\Delta x \Delta y
\mysum_{l\in\mathbb{N}^{\dagger}}^{d\in\mathbb{N}^{\dagger}} \varphi_{l, d}~
     g(x_{n-l}, y_{j -d}, \Delta \tau)~  \phi_{l,d}^{m}~  +   \errorf
\nonumber
\\
&\overset{\text{(ii)}}{=}
\iint_{\Rbb^2} g(x_n - x, y_j - y, \Delta \tau)~ \chi(x, y) ~dx~dy +
\errorf + \errorb + \errorc
\nonumber
\\
&\overset{\text{(iii)}}{=}  [\chi*g](x_n, y_j) + \mathcal{O}(h^2) +  \mathcal{O}\big(he^{-1/h}\big) + \mathcal{O}(h^2)
\nonumber
\\
&=
\mathcal{F}^{-1}\l[\mathcal{F}\left[\chi\right](\eta_x, \eta_y)~ G\left(\eta_x, \eta_y, \Delta \tau\right)\r](x_n, y_j) + \mathcal{O}(h^2).
 \end{align}
 Here, in (i), the error $\errorf$ is the series truncation error;
 in (ii), two additional errors $\errorb$ and $\errorc$ are due to the boundary truncation error  and
 the numerical integration error, respectively; in (iii) $[\chi * g]$ denotes the convolution of $\chi(x, y)$ and $g(x, y, \Delta \tau)$; in addition, $\errorf = \mathcal{O}(h^2)$, $\errorb = \mathcal{O}\big(he^{-1/h}\big)$
  and $\errorc = \mathcal{O}(h^2)$ as previously discussed in \eqref{eq:err_h}.
In \eqref{eq:error_smooth_b},  with $\Psi(\eta_{\myx}, \eta_{\myy})$ given in \eqref{eq:G_closed}, expanding $G(\eta_{\myx}, \eta_{\myy}; \Delta \tau) = e^{\Psi(\eta_{\myx}, \eta_{\myy})\Delta \tau}$
using a Taylor series gives
\EQ
\label{eq:taylor}
G(\eta_{\myx}, \eta_{\myy}; \Delta \tau) \approx 1 + \Psi(\eta_{\myx}, \eta_{\myy}) \Delta \tau + \mathcal{R}(\eta_{\myx}, \eta_{\myy}) \Delta \tau^2,
\quad
\mathcal{R}(\eta_{\myx}, \eta_{\myy}) = \frac{\Psi(\eta_{\myx}, \eta_{\myy})^2 e^{\xi \Psi(\eta_{\myx}, \eta_{\myy}) }}{2},
\quad \xi \in (0, \Delta \tau).
\EN
Therefore,
\EQA
\label{eq:error_smooth_1}
\l[\chi*g\r](x_n, y_j)
&=&
\mathcal{F}^{-1}\l[\mathcal{F}\left[\chi\right]\!(\eta_{\myx},\eta_{\myy})~\l(1 + \Psi(\eta_{\myx},\eta_{\myy})\Delta \tau + \mathcal{R}(\eta_{\myx}, \eta_{\myy}) \Delta \tau^2)\r) \r]\l(x_n, y_j\r)
\nonumber
\\
&=& \chi(x_n, y_j) + \Delta \tau \mathcal{F}^{-1}\l[\mathcal{F}\left[\chi\right]\!(\eta_{\myx}, \eta_{\myy})~\Psi\left(\eta_{\myx}, \eta_{\myy}\right)\r](x_n, y_j)
\nonumber
\\
&& \qquad + \Delta \tau^2  \mathcal{F}^{-1}\l[\mathcal{F}\left[\chi\right]\!(\eta_{\myx}, \eta_{\myy}) ~\mathcal{R}(\eta_{\myx}, \eta_{\myy}) \r](x_n, y_j).
\ENA
Here, the first term in \eqref{eq:error_smooth_1}, namely $\chi(x_n, y_j) \equiv \chi(x_n, y_j, \tau_m) $ is simply $\phi_{n, j}^{m}$
by construction of $\chi(\cdot)$.
For the second term in \eqref{eq:error_smooth_1}, we focus on  $\mathcal{F}\left[\chi\right]\!(\eta_{\myx}, \eta_{\myy})~\Psi\left(\eta_{\myx}, \eta_{\myy}\right)$.
Recalling the closed-form expression for $\Psi(\eta_{\myx}, \eta_{\myy})$ in \eqref{eq:G_closed},
we obtain $\mathcal{F}[\chi](\eta_{\myx}, \eta_{\myy}) \Psi(\eta_{\myx}, \eta_{\myy}) $

\begin{align}
 \mathcal{F}[\chi](\eta_{\myx}, \eta_{\myy}) \Psi(\eta_{\myx}, \eta_{\myy}) &= \mathcal{F}\big[ \frac{\sigma_{x}^{2}}{2} \chi_{xx} + \frac{\sigma_{y}^{2}}{2} \chi_{yy} + (r-\lambda\kappa_{x}-\frac{\sigma_x^2}{2})\chi_{x} + (r-\lambda\kappa_{y}-\frac{\sigma_y^2}{2})\chi_{y} + \rho\sigma_x\sigma_y\chi_{xy}-\ldots  \nonumber\\
 &\qquad\ldots- (r+\lambda)\chi +\lambda\int_{\Rbb^2}\chi(x+s_x,y+s_y)f(s_x,s_y)\md s_x\md s_y \big](\eta_{\myx}, \eta_{\myy})\nonumber\\
&= \mathcal{F}\l[ \mathcal{L} \chi+\mathcal{J} \chi\r](\eta_{\myx}, \eta_{\myy}).
\end{align}
Therefore, the second term in \eqref{eq:error_smooth_1} becomes
\EQA
\label{eq:im_2}
\Delta \tau \mathcal{F}^{-1}\l[\mathcal{F}\left[\chi\right]\!(\eta_{\myx}, \eta_{\myy})~\Psi\left(\eta_{\myx}, \eta_{\myy}\right)\r](x_n, y_j)
= \Delta \tau \l[ \mathcal{L}\chi +\mathcal{J} \chi\r] (\x_{n, j}^{m})
= \Delta \tau \l[ \mathcal{L} \chi+\mathcal{J} \chi \r]_{n, j}^{m}.
\ENA
For the third term $\Delta \tau^2  \mathcal{F}^{-1}\l[\mathcal{F}\left[\chi\right]\!(\eta_{\myx}, \eta_{\myy}) ~\mathcal{R}(\eta_{\myx}, \eta_{\myy}) \r](x_n, y_j)$ in \eqref{eq:error_smooth_1}, we have
\begin{align}
\label{eq:err_g}
&\Delta \tau^2 \l| \mathcal{F}^{-1}\l[ \mathcal{F} [\chi](\eta_{\myx},\eta_{\myy})~\mathcal{R}(\eta_{\myx}, \eta_{\myy})\r]\!(x_n, y_j)\r|
\nonumber
\\
& \qquad\qquad =
\frac{\Delta \tau^2}{(2\pi)^2} \bigg| \iint_{\Rbb^2} e^{i(\eta_{\myx} x_n+\eta_{\myy} y_j)} \mathcal{R}(\eta_{\myx},\eta_{\myy}) \bigg[ \iint_{\Rbb^2}  e^{-i(\eta_{\myx} x+\eta_{\myy} y)} \chi(x,y)~dx~dy  \bigg] d \eta_{\myx} d \eta_{\myy} \bigg|
\nonumber
\\
&\qquad \qquad\leq
\Delta \tau^2  \iint_{\Rbb^2} \l| \chi(x,y)  \r|~dxdy~ \iint_{\Rbb^2}
 \l| \mathcal{R}(\eta_{\myx},\eta_{\myy}) \r|~d\eta_{\myx} d\eta_{\myy}.
\end{align}
Noting $\ds \mathcal{R}(\eta_{\myx},\eta_{\myy}) = \frac{\Psi(\eta_{\myx},\eta_{\myy})^2 e^{\xi\Psi(\eta_{\myx},\eta_{\myy})}}{2}$,
as shown in \eqref{eq:taylor}, where a closed-form expression for $\Psi(\eta_{\myx},\eta_{\myy})$ is given in \eqref{eq:G_closed},
we obtain
\[
|\mathcal{R}(\eta_{\myx}, \eta_{\myy})| = \frac{|(\Psi(\eta_{\myx}, \eta_{\myy}))^2|}{2} \exp\big(\xi\big(-\frac{\sigma_{x}^{2}\eta_{\myx}^2}{2} - \frac{\sigma_{y}^{2}\eta_{\myy}^2}{2} - \rho\sigma_x\sigma_y\eta_{\myx}\eta_{\myy} - (r+\lambda)\big)\big).
\]
The term $|(\Psi(\eta_{\myx}, \eta_{\myy}))^2|$ can be written in the form
$|\Psi|^2 = \sum_{\substack{k+q=4 \\ k, q \geq 0}} C_{kq} \eta_{\myx}^k \eta_{\myy}^q$, where $C_{kq}$ are bounded coefficients.
This is a quartic polynomial in $\eta_{\myx}$ and $\eta_{\myy}$. Furthermore, the exponent of exponential term is bounded by
\[
-\frac{1}{2} \sigma_{\myx}^2 \eta_{\myx}^2 -\frac{1}{2} \sigma_{\myy}^2 \eta_{\myy}^2- \rho \sigma_{\myx}\sigma_{\myy}\eta_{\myx}\eta_{\myy}-(r+\lambda) \le 
-\frac{1}{2} \sigma_{\myx}^2 \eta_{\myx}^2-\frac{1}{2} \sigma_{\myy}^2 \eta_{\myy}^2 + |\rho| \sigma_{\myx}\sigma_{\myy}|\eta_{\myx}\eta_{\myy}|
\]
For $|\rho|< 1$, we have $|\rho| \sigma_{\myx}\sigma_{\myy}|\eta_{\myx}\eta_{\myy}| < \frac{1}{2}(\sigma_{\myx}^2 \eta_{\myx}^2 + \sigma_{\myy}^2 \eta_{\myy}^2)$. Therefore, we conclude that for $|\rho|< 1$,  the term  $\iint_{\Rbb^2} \l| \mathcal{R}(\eta_{\myx},\eta_{\myy}) \r|~d\eta_{\myx} d\eta_{\myy}$ is bounded since
\[
\iint_{\Rbb^2} |\eta_{\myx}|^k |\eta_{\myy}|^q~e^{-\frac{1}{2} \sigma_{\myx}^2 \eta_{\myx}^2-\frac{1}{2} \sigma_{\myy}^2 \eta_{\myy}^2 - \rho \sigma_{\myx}\sigma_{\myy}\eta_{\myx}\eta_{\myy}}~d\eta_{\myx}~d\eta_{\myy}, \quad k+q=4,~ k, q \geq 0,
\]
is also bounded.
Together with  $\chi(x,y) \in L_1(\mathbb{R}^2)$,  the rhs of \eqref{eq:err_g} is  $\mathcal{O}(\Delta \tau^2)$, i.e.\
\begin{align}
\label{eq:err_g2}
\Delta \tau^2 \l| \mathcal{F}^{-1}\l[ \mathcal{F} [\chi](\eta_{\myx},\eta_{\myy})~\mathcal{R}(\eta_{\myx}, \eta_{\myy})\r]\!(x_n, y_j)\r|
= \mathcal{O}(\Delta \tau^2)
\end{align}
Substituting \eqref{eq:im_2} and \eqref{eq:err_g2} back into \eqref{eq:error_smooth_1},
noting \eqref{eq:error_smooth_b} and $\chi(x, y, \tau_m) = \phi(x, y, \tau_m)$ for all  $(x, y) \in \mathbf{D}^{\dagger}$,
we have
\EQAS
\Delta x \Delta y
\mysum_{l\in\mathbb{N}^{\dagger}}^{d\in\mathbb{N}^{\dagger}} \varphi_{l, d}~
    g_{n-l, j - d}~  \phi_{l,d}^{m}~
=~
\phi_{n,j}^{m} + \Delta \tau \l[ \mathcal{L} \phi+ \mathcal{J} \phi\r]_{n,j}^{m}
+ \mathcal{O}(h^2).
\ENAS
This concludes the proof.
\end{proof}

\noindent Below, we state the key supporting lemma related to local consistency of our numerical scheme \eqref{eq:scheme_GF}.
\begin{lemma} [Local consistency]
\label{lemma:consistency}
Suppose that (i) the discretization parameter $h$ satisfies \eqref{eq:dis_parameter}.
Then, for any test function $\phi \in \mathcal{B}(\Oinf)\cap\mathcal{C}^{\infty}(\Oinf)$,
with  $\phi_{n, j}^{m} = \phi\l({\bf{x}}_{n, j}^{m}\r)$ and ${\bf{x}} \coloneqq (x_n, y_j, \tau_{m+1}) \in \Omega$, and for a sufficiently small $h$,  we have
\begin{linenomath}
\postdisplaypenalty=0
\EQA
\label{eq:lemma_1}
\mathcal{H}_{n, j}^{m+1}
\bigg(h, \phi_{n, j}^{m+1} + \xi,
\l\{\phi_{l, d}^{m}+\xi \r\}_{\subalign{l\in \Nd\\d\in \Jd}}
\bigg)
=
\left\{
\begin{array}{llllllllllr}
F_{\myin}\l(\cdot, \cdot\r)
&\!\!\!\!\!+~
c(\x)\xi
+ \mathcal{O}(h)
&&
{\bf{x}} \in \Omega_{\myin},
\\
F_{\myout}\l(\cdot, \cdot\r)
&
&&
{\bf{x}} \in \Omega_{\myout};
\\
F_{\tau_0}\l(\cdot, \cdot\r)
&
&&
{\bf{x}} \in \Omega_{\tau_0}.
\end{array}\right.
\ENA
\end{linenomath}
Here, $\xi$ is a constant, and $c(\cdot)$ is a bounded function
satisfying $|c({\bf{x}})| \le \max(r, 1)$ for all ${\bf{x}}~\in~\Omega$.
The  operators $F_{\myin}(\cdot, \cdot)$, $F_{\myout}(\cdot, \cdot)$, and $F_{\tau_0}(\cdot, \cdot)$, defined in \eqref{eq:fall}, are functions of $\l({\bf{x}}, \phi\l({\bf{x}}\r)\r)$.
\end{lemma}
\begin{proof}[Proof of Lemma~\ref{lemma:consistency}]
We now show that the first equation of \eqref{eq:lemma_1} is true, that is,
\EQAS
\mathcal{H}_{n, j}^{m+1}(\cdot) &\equiv& \mathcal{C}_{n, j}^{m+1} \l(\cdot\r)
= F_{\myin}\l({\bf{x}}, \phi\l({\bf{x}}\r)\r)
+ c(\x)\xi + \mathcal{O}(h)
\\
&&\qquad \qquad
\text{if}~
x_{\min} < x_n < x_{\max},~
y_{\min} < y_j < y_{\max},~
0 < \tau_{m+1} \le T.
\ENAS
where operators $\mathcal{C}_{n, j}^{m+1}(\cdot)$ is defined in \eqref{eq:scheme_CD}.
In this case,  the first argument of the $\min(\cdot, \cdot)$ operator in $\mathcal{C}_{n, j}^{m+1}(\cdot)$ is written as follows
\begin{linenomath}
\postdisplaypenalty=0
\label{eq:all_c}
\begin{align}
\text{$1^{st}$ arg} &= \frac{1}{\Delta \tau}\bigg[ \phi_{n, j}^{m+1} + \xi
-
\Delta x \Delta y \mysum_{l \in \Nd}^{d \in \Jd} \varphi_{l, d}~
g_{n-l, j-d}~(\phi^{m}_{l,d}+ \xi) \bigg]
\nonumber
\\
&= \frac{1}{\Delta \tau}\bigg[ \phi_{n, j}^{m+1}
-
 \bigg(\Delta x \Delta y \mysum_{l \in \Nd}^{d \in \Jd} \varphi_{l, d}~
g_{n-l, j-d}~\phi^{m}_{l,d}\bigg)
+ \xi \bigg( 1 - \Delta x \Delta y \mysum_{l \in \Nd}^{d \in \Jd} \varphi_{l, d}~
g_{n-l, j-d}\bigg)
\bigg]
\nonumber
\\
&\overset{\text{(i)}}{=}
\frac{\phi_{n, j}^{m+1} - \phi_{n, j}^{m}}{\Delta \tau}
- \l[ \Lcal \phi + \Jcal \phi \r]_{n,j}^{m}
+ \mathcal{O}(h)
+ \frac{\xi}{\Delta \tau} \bigg( 1 - \bigg\{\Delta x \Delta y \mysum_{l \in \Nd}^{d \in \Jd} \varphi_{l, d}~
g_{n-l, j-d}\bigg\}\bigg).
\end{align}
\end{linenomath}
Here,  (i) follows from Lemma~\ref{lemma:ar}, where the $\mathcal{O}(h^2)$ error term is divided by $\Delta \tau$, yielding $\mathcal{O}(h)$.
Regarding the second term of \eqref{eq:all_c}, 
we have
$1 - \Delta x \Delta y \mysum_{l \in \Nd}^{d \in \Jd} \varphi_{l, d}~
g_{n-l, j-d} = \ldots$
\EQ
\label{eq:termq}
\ldots
=
\bigg(1 - \iint_{\Rbb^2}g(x_n-x,y_j - y, \Delta \tau)dxdy\bigg)
+
\bigg(\iint_{\Rbb^2}g(\cdot,\cdot, \Delta \tau)dxdy -
\Delta x \Delta y \mysum_{l \in \Nd}^{d \in \Jd} \varphi_{l, d}~
g_{n-l, j-d}
\bigg).
\EN
The first term of \eqref{eq:termq} is simply $1 - e^{-r\Delta \tau} = r \Delta \tau + \mathcal{O}(h^2)$,
due to \eqref{eq:gdb}. The second term  of \eqref{eq:termq} is simply $\mathcal{O}(h^2) + \mathcal{O}(he^{-1/h})=\mathcal{O}(h^2)$
due to infinite series truncation error, numerical integration error, and boundary truncation error, as noted earlier. Thus,  the second term of \eqref{eq:all_c} becomes
\[
\frac{\xi}{\Delta \tau} \bigg( 1 -  \Delta x \Delta y \mysum_{l \in \Nd}^{d \in \Jd} \varphi_{l, d}~
g_{n-l, j-d}\bigg) =   r\xi + \mathcal{O}(h).
\]
Substituting this result into \eqref{eq:all_c} gives
 \begin{linenomath}
\postdisplaypenalty=0
\begin{align*}
\text{$1^{st}$ arg} &=
\frac{\phi_{n, j}^{m+1} - \phi_{n, j}^{m}}{\Delta \tau}
-  \l[ \mathcal{L}  \phi + \mathcal{J} \phi \r]_{n,j}^{m}
+ r\xi + \mathcal{O}(h)
\overset{\text{(i)}}{=}
\big[ \partial\phi/\partial \tau -   \mathcal{L} \phi  - \mathcal{J} \phi \big]_{n,j}^{m+1} + r\xi
+ \mathcal{O} (h).
\end{align*}
\end{linenomath}
Here, in (i), we use $(\partial\phi/\partial \tau )_{n,j}^{m} = (\partial\phi/\partial \tau )_{n,j}^{m+1} + \mathcal{O}\l(h\r)$; for $z \in \{x, y\}$,
$(\partial\phi/\partial z )_{n,j}^{m} = (\partial\phi/\partial z )_{n,j}^{m+1} + \mathcal{O}\l(h\r)$;
and for the cross derivative term $(\partial^2\phi/\partial x \partial y)_{n,j}^{m} = (\partial^2\phi/\partial x \partial y)_{n,j}^{m+1} + \mathcal{O}\l(h\r)$.

The second argument of the $\min(\cdot, \cdot)$ operator in $\mathcal{C}_{n, j}^{m+1}(\cdot)$ is simply
$\phi_{n, j}^{m+1} + \xi  - \vh_{n, j}$. Thus,
\begin{align*}
\mathcal{C}_{n, j}^{m+1}(\cdot) &= \min \left(\big[ \partial\phi/\partial \tau -   \mathcal{L} \phi  - \mathcal{J} \phi \big]_{n,j}^{m+1}
+ r\xi
+ \mathcal{O} (h),~
\phi_{n, j}^{m+1} + \xi  - \vh_{n, j}\right)
\\
&=\min \left(\big[ \partial\phi/\partial \tau -   \mathcal{L} \phi  - \mathcal{J} \phi \big]_{n,j}^{m+1},~
\phi_{n, j}^{m+1} -  \vh_{n, j}\right)
+ c({\bf{x}})\xi + \mathcal{O} (h),
\\
&= F_{\myin}\l({\bf{x}}, \phi\l({\bf{x}}\r)\r)
+ c(\x)\xi + \mathcal{O}(h).
\end{align*}
Here, ${\bf{x}} =  (x_n, y_j, \tau_{m+1}) \in \Omega_{\myin}$, $|c({\bf{x}})| \le \max(r, 1)$. This proves the first equation in (\ref{eq:lemma_1}).
The remaining equations in (\ref{eq:lemma_1}) can be proved using similar arguments with the first equation,
and hence omitted for brevity. This concludes the proof.
\end{proof}

\noindent We now formally state a lemma regarding the consistency of scheme \eqref{eq:scheme_GF} in the viscosity sense.
\begin{lemma}
\label{lemma:consistency_viscosity}
Suppose that the discretization parameter $h$ satisfies \eqref{eq:dis_parameter}.
For all ${\bf{\hat{x}}} = (\hat{x}, \hat{y}, \hat{\tau}) \in \Oinf$,
and for any $\phi \in \mathcal{B}(\Oinf)\cap\mathcal{C}^{\infty}(\Oinf)$,
with  $\phi_{n, j}^{m} = \phi\big({\bf{x}}_{n, j}^{m}\big)$ and {\bf{x}}~=~$(x_n, y_j, \tau_{m+1})$,
the scheme \eqref{eq:scheme_GF} satisfies
\EQA
\limsup_{\subalign{h \to 0, & ~  {\bf{x}} \to {\bf{\hat{x}}} \\ \xi &\to 0}}
\mathcal{H}_{n, j}^{m+1}
\bigg(h, \phi_{n, j}^{m+1}+ \xi,
\l\{\phi_{l, d}^{m}+\xi \r\}_{\subalign{l\in \Nd\\d\in \Jd}} \bigg)
\leq
F^* \l(
              {\bf{\hat{x}}}, \phi({\bf{\hat{x}}}), D\phi({\bf{\hat{x}}}), D^2 \phi({\bf{\hat{x}}}),  \mathcal{J} \phi({\bf{\hat{x}}})\r),
\label{eq:consistency_viscosity_1}
\\
\liminf_{\subalign{h \to 0, & ~ {\bf{x}} \to {\bf{\hat{x}}} \\ \xi &\to 0}}
\mathcal{H}_{n, j}^{m+1}
\bigg(h, \phi_{n, j}^{m+1}+ \xi,
\l\{\phi_{l, k}^{m}+\xi \r\}_{\subalign{l\in \Nd\\d\in \Jd}} \bigg)
\geq
F_*\l(
              {\bf{\hat{x}}}, \phi({\bf{\hat{x}}}), D\phi({\bf{\hat{x}}}), D^2 \phi({\bf{\hat{x}}}),  \mathcal{J} \phi({\bf{\hat{x}}})
             \r).
\label{eq:consistency_viscosity_2}
\ENA
Here, $F^*(\cdot)$ and  $F_*(\cdot)$ respectively are
the u.s.c. and the l.s.c. envelop of the operator $F(\cdot)$ defined in \eqref{eq:F}.
\end{lemma}
\begin{proof}[Proof of Lemma~\ref{lemma:consistency_viscosity}]
The proof is straightforward, deriving from Lemma~\ref{lemma:consistency}
and the definitions of u.s.c.\ and l.s.c.\ envelopes given in \eqref{eq:envelop}.
\end{proof}

\subsection{Monotonicity}
Below, we present a result concerning the monotonicity of our scheme \eqref{eq:scheme_GF}.
\begin{lemma}{(Monotonicity)}
\label{lemma:mon}
Scheme \eqref{eq:scheme_GF} satisfies
 \EQA
\label{eq:mon}
\mathcal{H}^{m+1}_{n,j}\l(h,  v^{m+1}_{n,j}, \l\{w^{m}_{l,d}\r\}\r)\leq\mathcal{H}^{m+1}_{n,j}\l(h, v^{m+1}_{n,j}, \l\{z^{m}_{l,d}\r\}\r)
\ENA
for bounded $\l\{w^{m}_{l,d}\r\}$ and $\l\{z^{m}_{l,d}\r\}$ having $\l\{w^{m}_{l,d}\r\}\geq \l\{z^{m}_{l,d}\r\}$,
where the inequality is understood in the component-wise sense.
\end{lemma}
\begin{proof}[Proof of Lemma~\ref{lemma:mon}]
Since scheme \eqref{eq:scheme_GF} is defined case-by-case, to establish \eqref{eq:mon}, we show that each
case satisfies \eqref{eq:mon}. It is straightforward that the scheme satisfies \eqref{eq:mon} in
${\Omega}_{\tau_0}$) and $ {\Omega}_{{\myout}}$. We now establish
that  $\mathcal{C}_{n, j}^{m+1}\l(\cdot\r)$, as defined in \eqref{eq:scheme_CD}
for $\Omega_{\myin}$, also satisfies \eqref{eq:mon}. We have
\begin{align}
\label{eq:mon_p}
&\mathcal{C}^{m+1}_{n,j}\l(h,  v^{m+1}_{n,j}, \l\{w^{m}_{l,d}\r\}\r)-\mathcal{C}^{m+1}_{n,j}\l(h, v^{m+1}_{n,j}, \l\{z^{m}_{l,d}\r\}\r)\nonumber
\\
&\quad = \min\bigg\{ \frac{1}{\Delta \tau}\bigg(v_{n,j}^{m+1}
-
 \Delta x \Delta y \mysum_{l \in \Nd}^{d \in \Jd} \varphi_{l, d}~
g_{n-l, j-d}~w^{m}_{l,d} \bigg),~ v_{n,j}^{m+1}-\hat{v}_{n,j}\bigg\}
\nonumber
\\
&\qquad \qquad -
\min\bigg\{ \frac{1}{\Delta \tau}\bigg(v_{n,j}^{m+1}
-
 \Delta x \Delta y \mysum_{l \in \Nd}^{d \in \Jd} \varphi_{l, d}~
g_{n-l, j-d}~z^{m}_{l,d} \bigg),~ v_{n,j}^{m+1}-\hat{v}_{n,j}\bigg\}
\nonumber
\\
&\quad \overset{\text{(i)}}{\le}
\max
\bigg\{ \frac{1}{\Delta \tau}  \Delta x \Delta y \mysum_{l \in \Nd}^{d \in \Jd} \varphi_{l, d}~
g_{n-l, j-d}~(z_{l,d}^{m+1} - w_{l,d}^{m+1}),0\bigg\}  \overset{\text{(ii)}}{=} 0.
\end{align}
Here, (i) is due to the fact that $\min (c_1, c_2) - \min(c_3, c_4) \le
\max(c_1-c_3, c_2-c_4)$ for real numbers $c_1, c_2, c_3, c_4$;
(ii) follows from $\max(\cdot, 0) = 0$, since $z_{l,d}^{m+1} - w_{l,d}^{m+1} \le 0$ and $g_{n-l, j-d}\ge 0$ for
all $n\in \Nbb$, $j \in \Jbb$, $l\in \Nd$, and $d \in \Jd$.
This concludes the proof.
\end{proof}

\subsection{Main convergence result}
We have demonstrated that the scheme \ref{eq:scheme_GF} satisfies three key properties in
$\Omega$: (i) $\ell_{\infty}$-stability (Lemma~\ref{lemma:stability}), (ii) consistency in the viscosity
sense (Lemma~\ref{lemma:consistency_viscosity}) and (iii) monotonicity (Lemma~\ref{lemma:mon}).
With a provable strong comparison principle result for
$\Omega_{\myin}$ in Theorem~\ref{theorem:comparison}, we now present the main convergence result of the paper.
\begin{theorem} [Convergence to viscosity solution in $\Omega_{\myin}$]
\label{thm:convergence}
Suppose that all the conditions for Lemmas~\ref{lemma:stability}), \ref{lemma:consistency_viscosity}
and \ref{lemma:mon} are satisfied. As the parameter discretization $h \to 0$, the scheme \eqref{eq:scheme_GF} converges uniformly on $\Omega_{\myin}$ to the unique continuous viscosity solution of the variational inequality \eqref{eq:F}
in the sense of Definition~\ref{def:vis}.
\end{theorem}
\begin{proof}[Proof of Theorem~\ref{thm:convergence}]
Our scheme is $\ell_\infty$-stable (Lemma \ref{lemma:stability}),
and consistent in the viscosity sense (Lemma \ref{lemma:consistency_viscosity})
and monotone (Lemma~\ref{lemma:mon}). Since a comparison result holds in $\Omega_{\myin}$ (Theorem~\ref{theorem:comparison}),
by \cite{barles-souganidis:1991, barles:1997, barles95a}, our scheme converges uniformly on $\Omega_{\myin}$ to the unique continuous viscosity solution of the variational inequality~\eqref{eq:F}.
\end{proof}

%% file: section_6_Numerical_experiments.tex
\section{Numerical experiments}
\label{sec:num_test}
In this section, we present the selected numerical results of our monotone integration method (MI)
applied to the American options under a two-asset Merton jump-diffusion model pricing problem.

\subsection{Preliminary}
\label{sc:num}
For our numerical experiments, we evaluate three parameter sets for the two-asset Merton jump-diffusion model, labeled as Case I, Case II, and Case III. The modeling parameters for these tests, reproduced from \cite{ghosh2022high}[Table~1], are provided in Table~\ref{tab:parameter02}. Notably, the parameters in Cases I, II, and III feature progressively larger jump intensity rates \(\lambda\). As previously mentioned, we can choose \(P^{\dagger} = P^{\dagger}_{\myx} \wedge P^{\dagger}_{\myy}\) sufficiently large to remain constant across all refinement levels (as \(h \to 0\)). Due to the varying jump intensity rates, we select computational domains of appropriate size for each case, listed in Table~\ref{tab:step03}, and confirm that these domains are sufficiently large through numerical validation in Subsection \ref{ssc:spatial_size}. Furthermore, values for \(x_{\mymin}^{\dagger}\), \(x_{\mymax}^{\dagger}\), \(y_{\mymin}^{\dagger}\), and \(y_{\mymax}^{\dagger}\) were chosen according to \eqref{eq:w_choice_green_jump_form}.
Unless specified otherwise, the details on mesh size and timestep refinement levels (``Refine.\ level'') used in all experiments are summarized in Table~\ref{tab:step01}.

 \begin{minipage}{0.5\textwidth}
\strut\vspace*{-\baselineskip}\newline
\center
\begin{tabular}{cccc}
\hline
 & \multicolumn{1}{c}{Case I} & \multicolumn{1}{c}{Case II} & \multicolumn{1}{c}{Case III}                                     \\ \cline{1-4}
 Diffusion parameters &  &      &                            \\ \cline{1-4}
  $\sigma_{\myx}$              & 0.12        &0.30 &0.20\\
    $\sigma_{\myy}$              & 0.15        &0.30&0.30\\
     $\rho$              & 0.30        &0.50&0.70\\ \hline\hline
Jump parameters & &&\
                            \\ \hline
$\lambda$                      & 0.60    &  2 & 8 \\
$\Tilde{\mu}_{\myx}$                      & -0.10    &  -0.50 &-0.05  \\
$\Tilde{\mu}_{\myy}$                      & 0.10    &  0.30  &-0.20\\
$\rhoh$                      & -0.20    &-0.60 &0.50   \\
$\Tilde{\sigma}_{\myx}$                      & 0.17    &0.40 &0.45  \\
$\Tilde{\sigma}_{\myy}$                      & 0.13    & 0.10 &0.06 \\\hline\hline
$K$ & 100 & 40 &40\\
$T$ (years) & 1 &0.5 & 1
\\
$r$  & 0.05 &0.05 & 0.05
\\
\hline
\end{tabular}

\captionof{table}{Model parameters used in numerical experiments for two-assets Merton jump-diffusion model
reproduced from \cite{ghosh2022high}~Table~1.
}
\label{tab:parameter02}
\end{minipage}
\begin{minipage}{0.5\textwidth}
\strut\vspace*{-\baselineskip}\newline
\center
\centering
\begin{tabular}{crrrr}
\hline
Refine.\  level & $N$ &  $J$                & $M$      \\
      & ($x$) & ($y$)               & ($\tau$) &  \\ \hline
0     & $2^{8}$ & $2^{8}$          & 50    \\
1     & $2^{9}$ & $2^{9}$          & 100     \\
2     & $2^{10}$ & $2^{10}$        & 200       \\
3     & $2^{11}$& $2^{11}$        & 400     \\
4     & $2^{12}$&$2^{12}$        & 800       \\
\hline
\end{tabular}
\captionof{table}{Grid and timestep refinement levels for numerical tests.}
\label{tab:step01}
\vspace{1em}
\begin{tabular}{crrr}
\hline
 & Case I & Case II & Case III               \\\hline
$x_{\mymin}$     &$\ln(X_{0})-1.5 $ &$\ln(X_{0})-3 $ &$\ln(X_{0})-6 $   \\
$x_{\mymax}$     &$\ln(X_{0})+1.5 $ &$\ln(X_{0})+3 $ &$\ln(X_{0})+6 $   \\
$y_{\mymin}$     &$\ln(Y_{0})-1.5 $ &$\ln(Y_{0})-3 $ &$\ln(Y_{0})-6 $   \\
$y_{\mymax}$     &$\ln(Y_{0})+1.5 $ &$\ln(Y_{0})+3 $ &$\ln(Y_{0})+6 $
\\\hline
\end{tabular}
\captionof{table}{Computational domains of numerical experiments for Cases I, II, and III.}
\label{tab:step03}
\end{minipage}
\subsection{Validation examples}
For the numerical experiments, we analyze two types of options: an American put-on-the-min option and an American put-on-the-average option, each with a strike price \( K \), as described in \cite{ghosh2022high, boen2020operator}.
\subsubsection{Put-on-the-min option}
%
Our first test case examines an American put option on the minimum of two assets, as described in \cite{ghosh2022high, boen2020operator}. The payoff function $\hat{v}(x, y)$ is defined as
\EQA
\label{eq:BSO_P1}
\hat{v}(x, y) = \max(K-\min(e^x,e^y),0), \quad K>0.
\ENA
As a representative example, we utilize the parameters specified in Case I, with initial asset values \(X_0 = 90\) and \(Y_0 = 90\), for the put-on-the-min option. Computed option prices for this test case are presented in Table~\ref{tab:result01}. To estimate the convergence rate of the proposed method, we calculate the ``Change'' as the difference between computed option prices at successive refinement levels and the ``Ratio'' as the quotient of these changes between consecutive levels. As shown, these computed option prices exhibit first-order convergence and align closely with results obtained using the operator splitting method in \cite{boen2020operator}. In addition, Figure~\ref{fig:sub4} displays the early exercise regions at \(T/2\) for this test case.

Tests conducted under Cases II and III demonstrate similar convergence behavior. Numerical results for American put-on-the-min options with various initial asset values and parameter sets are summarized in Section \ref{ssc:ct} [Table ~\ref{tab:result06}].

\begin{minipage}{0.5\textwidth}
\strut\vspace*{-\baselineskip}
\flushleft
\begin{linenomath}
\begin{table}[H]
\center
\begin{tabular}{llll}
\hline
 Refine.\ level     & Price        & Change      & Ratio           \\
\hline
     \quad   0         &16.374702 & &            \\
    \quad    1         &16.383298 &8.60e-03 &         \\
    \quad 2            &16.387210 &3.91e-03  &2.20         \\
      \quad  3         &16.389079 &1.87e-03  &2.09         \\
   \quad  4            &16.389991 &9.11e-04  &2.05          \\
\hline
Ref.\ \cite{boen2020operator}&  16.390 & &
\\     \hline

\end{tabular}
\captionof{table}{Convergence study for a put-on-the-min American option under two-assets Merton jump-diffusion model (modeling parameters in Table \ref{tab:parameter02}, Case I) with initial asset values $X_0 =90$ and $Y_0 =90$ -  payoff function in \eqref{eq:BSO_P1}. Reference prices: by FD  method (operator splitting) is 16.390 \cite{boen2020operator}.
}
\label{tab:result01}
\end{table}
\end{linenomath}
\end{minipage}
\begin{minipage}{0.5\textwidth}
\strut\vspace*{-\baselineskip}

\begin{figure}[H]
\center
\includegraphics[width=0.8\linewidth]{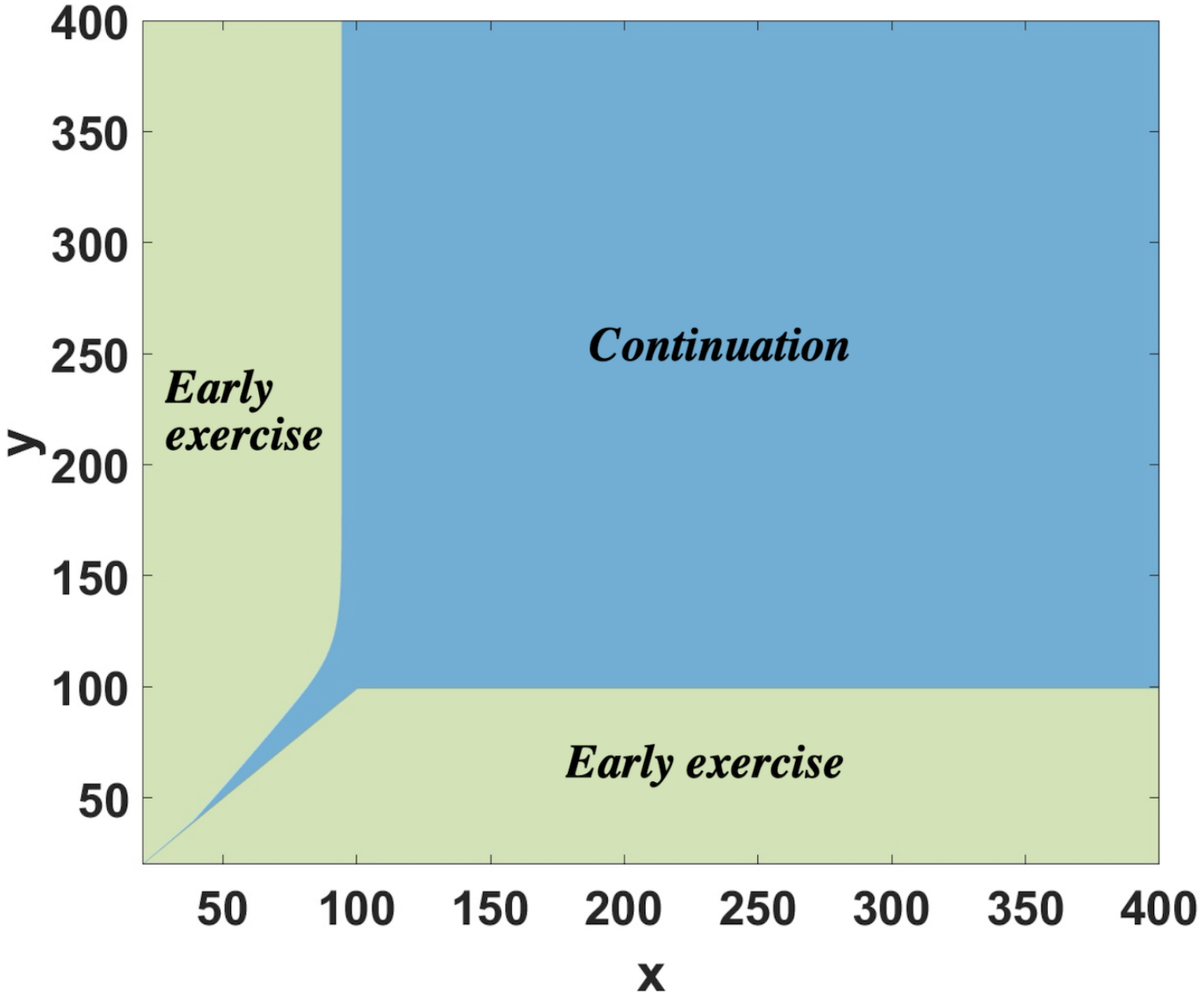}
    \caption{Early exercise regions for the American put-on-the-min at $t=T/2$, corresponding to
    Refine.\ level 4 from Table~\ref{tab:result01}. }
    \label{fig:sub4}
\end{figure}
\end{minipage}

\subsubsection{Put-on-the-average option}
%
For the second test case, we examine an American option based on the arithmetic average of two assets. The payoff function, $\hat{v}(x, y)$, is defined as:
\EQA
\label{eq:BSO_P2}
\hat{v}(x,y) = \max(K-(e^x+e^y)/2,0), \quad K>0.
\ENA
As a representative example, we use the modeling parameters from Case I, with initial asset values set at \(X_0 = 100\) and \(Y_0 = 100\) to illustrate the put-on-the-average option. The computed option prices, presented in Table~\ref{tab:result02}, demonstrate a first order of convergence and show strong agreement with the results reported in \cite{boen2020operator}. In addition, the early exercise regions at \(T/2\) for this case are depicted in Figure~\ref{fig:eer_aver}. Similar experiments conducted for Cases II and III yield comparable results. Further numerical results for American put-on-the-average options, encompassing various initial asset values and parameter sets, are presented in Section~\ref{ssc:ct} [Table~\ref{tab:result07}].

 \begin{minipage}{0.5\textwidth}
\strut\vspace*{-\baselineskip}
\flushleft
\begin{linenomath}
\begin{table}[H]
\center
\begin{tabular}{llll}
\hline
 Refine.\ level     & Price        & Change      & Ratio           \\
\hline
     \quad   0         &3.431959 & &            \\
    \quad    1         &3.436727 &4.77e-03 &         \\
    \quad 2            & 3.439096 &2.37e-03  &2.01         \\
      \quad  3         &3.440278&1.18e-03  &2.00         \\
   \quad  4            & 3.440868&5.90e-04  &2.00          \\
\hline
Ref.\ \cite{boen2020operator}&  3.442 & &
\\     \hline

\end{tabular}
\captionof{table}{Convergence study for a put-on-the-average American option under two-assets Merton jump-diffusion model (modeling parameters in Table \ref{tab:parameter02}, Case I) with initial asset values $X_0 =100$ and $Y_0 =100$ -  payoff function in \eqref{eq:BSO_P2}. Reference prices: by FD  method (operator splitting) is 3.442 \cite{boen2020operator}.
}
\label{tab:result02}
\end{table}
\end{linenomath}
\end{minipage}
\begin{minipage}{0.5\textwidth}
\strut\vspace*{-\baselineskip}
\begin{figure}[H]
\center
\includegraphics[width=0.8\linewidth]{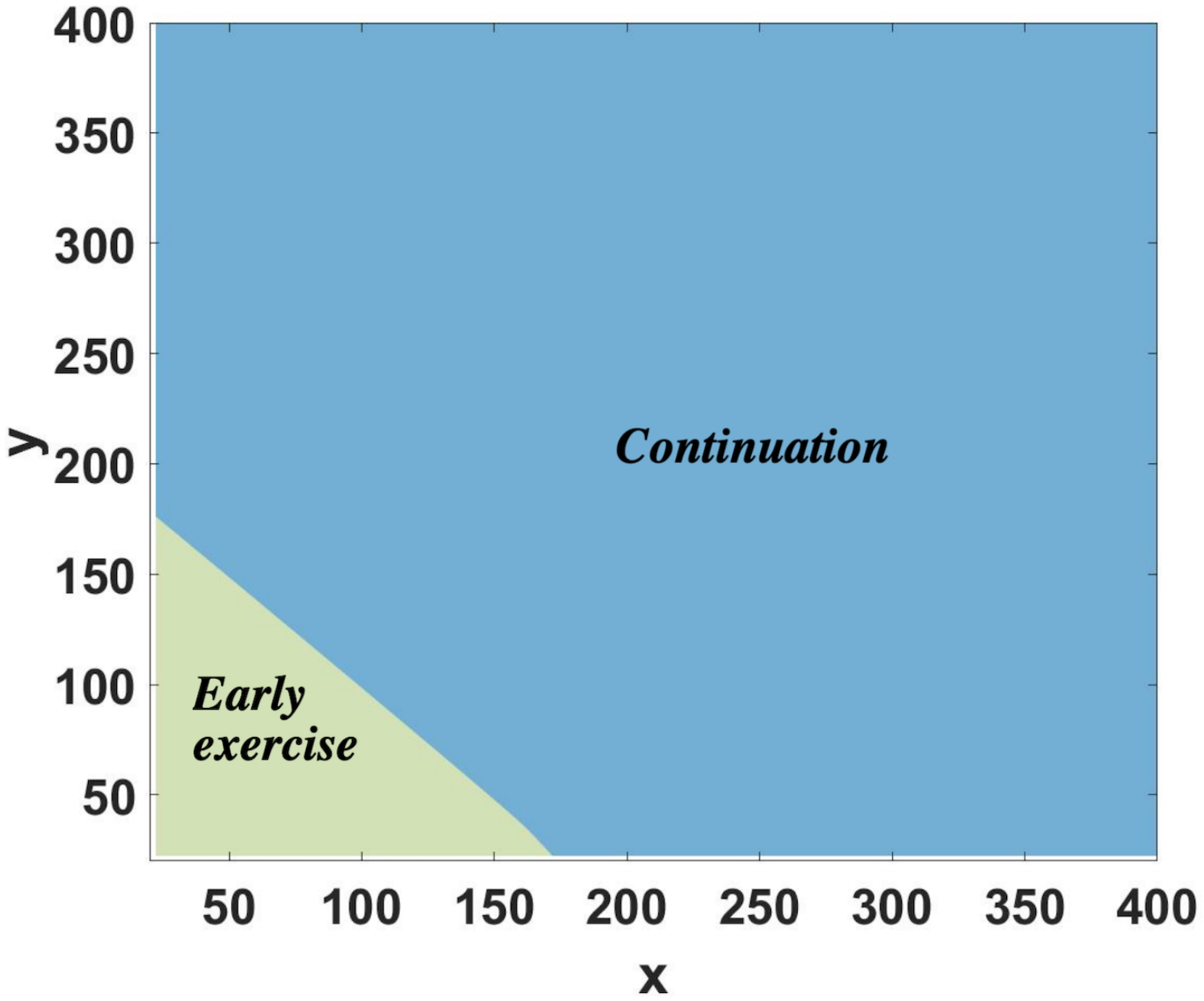}
    \caption{Early exercise regions for the American put-on-the-average at $t=T/2$,
    corresponding to Refine.\ level 4 from Table~\ref{tab:result02}.}
    \label{fig:eer_aver}
\end{figure}
\end{minipage}

\medskip
{\apnum{
It is important to note that the convergence rates in our numerical experiments are not directly comparable to those reported in \cite{ruijter2012two}. As highlighted earlier, the method presented in this work is more general and can be extended in a straightforward manner to stochastic control problems, such as asset allocation, where no continuation boundary exists. Furthermore, our method directly converges to the viscosity solution of the true American option pricing problem, unlike \cite{ruijter2012two}, which focuses on Bermudan options and does not examine the limiting case as  $\Delta \tau \to 0$.}}

\subsubsection{Impact of spatial domain sizes}
\label{ssc:spatial_size}
In this subsection, we validate the adequacy of the chosen spatial domain for our experiments, focusing on Case I for brevity. Similar tests for Cases II and III yield consistent results and are omitted here.

To assess domain sufficiency, we revisit the setup from Table~\ref{tab:result01} and double the sizes
of the interior sub-domain \(\myDin\), extending \(x_{\min} = \ln(X_0) - 1.5\), \(x_{\max} = \ln(X_0) + 1.5\), \(y_{\min} = \ln(Y_0) - 1.5\), and \(y_{\max} = \ln(Y_0) + 1.5\) to \(x_{\min} = \ln(X_0) - 3\), \(x_{\max} = \ln(X_0) + 3\), \(y_{\min} = \ln(Y_0) - 3\), and \(y_{\max} = \ln(Y_0) + 3\). The boundary sub-domains are adjusted accordingly as in \eqref{eq:w_choice_green_jump_form}. We also double the intervals \(N\) and \(J\) to preserve \(\Delta x\) and \(\Delta y\) as in the setup from Table~\ref{tab:result01}.

The computed option prices for this larger domain, presented in Table~\ref{tab:result04} under ``Larger \(\myDin\)'' show minimal differences from the original results (shown under ``Table~\ref{tab:result01}), with discrepancies only appearing at the 8th decimal place. These differences are recorded in the ``Diff.'' column, which represents the absolute difference between the computed option prices from Table~\ref{tab:result01} and those obtained with either an extended or contracted interior sub-domain  \(\myDin\).
This indicates that further enlarging the spatial computational domain has a negligible effect on accuracy.

\begin{table}[htb!]
\center
\begin{tabular}{ll|ll|ll|}
\hline
 Refine.\               & Table~\ref{tab:result01}        & \multicolumn{2}{c|}{(Larger $\myDin$)}          &    \multicolumn{2}{c|}{(Smaller $\myDin$)}
                       \\
level                        & Price        &Price    &   Diff.      & Price   &Diff.\
                            \\
\hline
0                      & 16.374702    &  16.374702        & 1.64e-08    &   16.374210      & 4.92e-04                                                               \\
1                      & 16.383298     &  16.383298      & 1.60e-08    &  16.382820   & 4.78e-04                     \\
2                      & 16.387210     &  16.387210    & 1.60e-08     & 16.386736 & 4.74e-04        \\
3                      & 16.389079     &  16.389079     &1.62e-08    & 16.388605   & 4.74e-04      \\
4                      & 16.389991     &  16.389991    & 1.62e-08      & 16.389515   & 4.76e-04         \\ \hline
\end{tabular}
\caption{Prices (put-on-min) obtained using different spatial computational domain: (i) a \underline{Larger} $\myDin$ with
$x_{\min} = \ln(X_0)-3$, $x_{\max} = \ln(X_0)+3$, $y_{\min} = \ln(Y_0)-3$, $y_{\max} = \ln(Y_0)+3$,  and (ii) a \underline{Smaller}
$\myDin$ with  $x_{\min} = \ln(X_0)-0.75$, $x_{\max} = \ln(X_0)+0.75$, $y_{\min} = \ln(Y_0)-0.75$, $y_{\max} = \ln(Y_0)+0.75$.
These are to compare with prices in Table~\ref{tab:result01}
obtained using the original $\myDin$ with  $z_{\min} = \ln(Z_0)-1.5$, $z_{\max} = \ln(Z_0)+1.5$, for $z\in \{x,y\}$
as in Table~\ref{tab:step03}[Case 1].
}
\label{tab:result04}
\end{table}
In addition, we test a smaller interior domain \(\myDin\) with boundaries \(x_{\min} = \ln(X_0) - 0.75\), \(x_{\max} = \ln(X_0) + 0.75\), \(y_{\min} = \ln(Y_0) - 0.75\), and \(y_{\max} = \ln(Y_0) + 0.75\), while keeping \(\Delta x\) and \(\Delta y\) constant. The results, shown in Table~\ref{tab:result04} under ``Smaller \(\myDin\)'', reveal differences starting at the third decimal place compared to the original setup. This indicates that the selected domain size is essential for achieving accurate results; further expansion of the domain size offers negligible benefit, whereas any reduction may introduce noticeable errors.

In Table~\ref{tab:result05}, we present the test results for extending and contracting \(\myDin\) for the American put-on-average option, which yield similar conclusions to those observed previously.
\begin{table}[htb!]
\center
\begin{tabular}{cl|ll|ll|}
\hline
Refine.\               & Table~\ref{tab:result02}        & \multicolumn{2}{c|}{(Larger $\myDin$)}          &    \multicolumn{2}{c|}{(Smaller $\myDin$)}
                       \\
level                        & Price        &Price    &   Diff.      & Price   &Diff.\
                            \\
                            \hline
0                      & 3.431959   &  3.431959        & 2.91e-07    &   3.431348      & 6.12e-04                                                               \\
1                      & 3.436727     &  3.436727      & 3.09e-07    &  3.436080   & 6.37e-04                     \\
2                      & 3.439096     &  3.439096    & 3.21e-07     & 3.436080 & 6.69e-04        \\
3                      & 3.440278     &  3.440278     &3.26e-07    & 3.438426   & 6.82e-04      \\
4                      & 3.440868     &  3.440868    & 3.29e-07      & 3.440178   & 6.91e-04         \\ \hline
\end{tabular}
\caption{Prices (put-on-average)  obtained using different spatial computational domain: (i) a \underline{Larger} $\myDin$ with
$x_{\min} = \ln(X_0)-3$, $x_{\max} = \ln(X_0)+3$, $y_{\min} = \ln(Y_0)-3$, $y_{\max} = \ln(Y_0)+3$, and (ii) a \underline{Smaller} $\myDin$ with
$x_{\min} = \ln(X_0)-0.75$, $x_{\max} = \ln(X_0)+0.75$, $y_{\min} = \ln(Y_0)-0.75$, $y_{\max} = \ln(Y_0)+0.75$.
These are to compare with prices in Table~\ref{tab:result02}
obtained using the original $\myDin$ with  $z_{\min} = \ln(Z_0)-1.5$, $z_{\max} = \ln(Z_0)+1.5$, for $z\in \{x,y\}$,
as in Table~\ref{tab:step03}[Case 1].}
\label{tab:result05}
\end{table}

\subsubsection{Impact of boundary conditions}
\label{sec:constant_pad}

In this subsection, we numerically demonstrate that our straightforward approach of employing discounted payoffs for boundary sub-domains is adequate. For brevity, we show the tests of impact of boundary conditions for Case I. Similar experiments for Cases II and III yield the same results.

We revisited previous experiments reported in Tables~\ref{tab:result01}, introducing more sophisticated boundary conditions based on the asymptotic behavior of the PIDEs \eqref{eq:2dPIDEs*} as $z \to -\infty$ and $z \to \infty$ for $z\in\{x,y\}$ as proposed in \cite{clift2008numerical}.
Specifically,  the PIDEs \eqref{eq:2dPIDEs*} simplifies to the 1D PDEs shown in \eqref{eq:hjbxymin} when $x$ or $y$ tends to $-\infty$:
\EQ
\label{eq:hjbxymin}
\begin{aligned}
v_{\tau} -  \big(r-(\sigma_{y})^2/2\big)v_{y}  +   (\sigma_{y})^2/2 v_{yy}  + rv &= 0, \quad x \to -\infty,
\\
v_{\tau} - \big(r-(\sigma_{x})^2/2\big)v_{x} +   (\sigma_{x})^2/2v_{xx}  +rv &= 0, \quad y \to -\infty.
\end{aligned}
\EN
That can be justified based on the properties of the Green’s function of the PIDE \cite{garroni1992green}.
As  $x, y \to -\infty$, the PIDEs \eqref{eq:2dPIDEs*} simplifies to the ordinary differential equation $v_{\tau} + rv = 0$.

To adhere to these asymptotic boundary conditions, we choose a much large spatial domain: $x_{\min} = \ln(X_0)-12$, $x_{\max} = \ln(X_0)+12$, $y_{\min} = \ln(Y_0)-12$, $y_{\max} = \ln(Y_0)+12$,  and adjust the number of intervals $N$ and $J$ accordingly to maintain the same grid resolution ($\Delta x$ and $\Delta y$).
Employing the monotone integration technique, tailored for the 1D case, we effectively solved the 1D PDEs in \eqref{eq:hjbxymin}. The ordinary differential equation $v_{\tau} + rv = 0$ is solved directly and efficiently. The scheme's convergence to the viscosity solution can be rigorously established in the same fashion as the propose scheme.
\begin{table}[htb!]
\center
\begin{tabular}{c|ll|ll|}
\hline
 & \multicolumn{2}{c|}{Put-on-the-min} & \multicolumn{2}{c|}{Put-on-the-average}                                     \\
\cline{2-5}
 Refine.    & Price            & Price       & Price            & Price
                       \\
 level                       &        & (Table~\ref{tab:result01})   &        & (Table~\ref{tab:result02})
                            \\ \hline
0                      & 16.374702    &  16.374702        & 3.431959   &  3.431959                                                            \\
1                      & 16.383298     &  16.383298        & 3.436727     &  3.436727                \\
2                      & 16.387210     &  16.387210        & 3.439096     &  3.439096    \\
3                      & 16.389079     &  16.389079        & 3.440278     &  3.440278 \\
4                      & 16.389991     &  16.389991         & 3.440868     &  3.440868  \\ \hline
\end{tabular}
\caption{
Results using sophisticated boundary conditions. Compare with computed prices in Table~\ref{tab:result01} and Table~\ref{tab:result02} where discounted payoff boundary conditions are used.
}
\label{tab:result08}
\end{table}
The computed option prices for the put-on-the-min option, as shown in Table~\ref{tab:result08}, are virtually identical with those from the original setup (see column marked ``Table.~\ref{tab:result01}''). In addition, Table~\ref{tab:result08} presents results for the put-on-the-average option using the similar sophisticated boundary condition, with findings consistent with the put-on-the-min option. These results confirm that our simple boundary conditions are not only easy to implement but also sufficient to meet the theoretical and practical requirements of our numerical experiments.
\subsubsection{Comprehensive tests}
\label{ssc:ct}
In the following, we present a detailed study of two types of options: an American put-on-the-min option and an American put-on-the-average option, tested with various strike prices \(K\) and initial asset values. Across all three parameter cases, our computed prices closely align with the reference prices given in \cite{boen2020operator}[Tables 5 and 6].

 \begin{minipage}{0.49\textwidth}
\strut\vspace*{-\baselineskip}
\flushleft
\begin{linenomath}
\begin{table}[H]
\center
\scalebox{0.9}
{
\begin{tabular}{lrrrr}
\hline
&\multicolumn{4}{c}{Put-on-min} \\ \hline
&\multicolumn{4}{c}{Case I} \\ \hline
&$Y_0$ & \multicolumn{3}{c}{$X_0$}\\ \hline
 & & \multicolumn{1}{c}{90} & \multicolumn{1}{c}{100} &  \multicolumn{1}{c}{110} \\ \hline
 MI & 90 &16.389991 &13.998405 &12.756851 \\
  & 100 &13.020204&9.619252&7.876121\\
   & 110 &11.441389 &7.226153 &5.131663\\\hline
Ref. \cite{boen2020operator}&$ 90$ & 16.391 & 13.999 & 12.758 \\
&$ 100$ & 13.021 & 9.620 & 7.877 \\
&$ 110$ & 11.443 & 7.227 & 5.132 \\ \hline
&\multicolumn{4}{c}{Case II} \\ \hline
& &\multicolumn{1}{c}{36} & \multicolumn{1}{c}{40} &  \multicolumn{1}{c}{44} \\ \hline
 MI & 36 &15.469776 & 14.566197  &13.796032 \\
  & 40 &14.094647&13.109244 &12.265787 \\
   & 44 &12.924092 &11.879584  &10.984126 \\\hline
Ref. \cite{boen2020operator}&$ 36$ & 15.467 & 14.564 & 13.794 \\
&$ 40$ & 14.092 & 13.107 & 12.263 \\
&$ 44$ & 12.921 & 11.877 & 10.982 \\ \hline
&\multicolumn{4}{c}{Case III} \\ \hline
& & \multicolumn{1}{c}{36} & \multicolumn{1}{c}{40} &  \multicolumn{1}{c}{44} \\ \hline
 MI & 36 &21.750926  & 20.917727 &20.176104  \\
  & 40 &21.281139 &20.403611 &19.620525 \\
   & 44 &20.906119  &19.992702  &19.176009 \\\hline
Ref. \cite{boen2020operator} &$ 36$ & 21.742 & 20.908 & 20.167 \\
&$ 40$ & 21.272 & 20.394 & 19.611 \\
&$ 44$ & 20.892 & 19.983 & 19.166 \\ \hline
\end{tabular}
}
\caption{Results for an American put-on-min option under Cases I, II, III. Reference price by FD  method (operator splitting) from \cite{boen2020operator}[Table~5]. }
\label{tab:result06}
\end{table}
\end{linenomath}
\end{minipage}
\begin{minipage}{0.49\textwidth}
\strut\vspace*{-\baselineskip}
\center
\begin{table}[H]
\center
\scalebox{0.9}
{
\begin{tabular}{lrrrr}
\hline
&\multicolumn{4}{c}{Put-on-average} \\ \hline
&\multicolumn{4}{c}{Case I} \\ \hline
&$Y_0$ & \multicolumn{3}{c}{$X_0$}\\ \hline
 & & \multicolumn{1}{c}{90} & \multicolumn{1}{c}{100} &  \multicolumn{1}{c}{110} \\ \hline
MI  & 90 &10.000000 &5.987037 &3.440343\\
  & 100 &6.028929 &3.440868&1.886527\\
   & 110 &3.490665 &1.890874 &0.992933\\\hline
Ref. \cite{boen2020operator}&$ 90$ & 10.003 & 5.989 & 3.441 \\
&$ 100$ & 6.030 & 3.442 & 1.877 \\
&$ 110$ & 3.491 & 1.891 & 0.993 \\ \hline
&\multicolumn{4}{c}{Case II} \\ \hline
& & \multicolumn{1}{c}{36} & \multicolumn{1}{c}{40} &  \multicolumn{1}{c}{44} \\ \hline
 MI & 36 &5.405825 & 4.363340  &3.547399 \\
  & 40 &4.213899&3.338840 &2.669076 \\
   & 44 &3.224979 &2.506688  &1.969401\\\hline
Ref. \cite{boen2020operator}&$ 36$ & 5.406 & 4.363 & 3.547 \\
&$ 40$ & 4.214 & 3.339 & 2.669 \\
&$ 44$ & 3.225 & 2.507 & 1.969 \\ \hline
&\multicolumn{4}{c}{Case III} \\ \hline
& & \multicolumn{1}{c}{36} & \multicolumn{1}{c}{40} &  \multicolumn{1}{c}{44} \\ \hline
 MI & 36 &12.472058  & 11.935904 &11.446078  \\
  & 40 &11.439979 &10.948971 &10.500581 \\
   & 44 &10.499147  &10.049777  &9.639534 \\\hline
Ref. \cite{boen2020operator} &$ 36$ & 12.466 & 11.930 & 11.440 \\
&$ 40$ & 11.434 & 10.943 & 10.495 \\
&$ 44$ & 10.493 & 10.043 & 9.633 \\ \hline
\end{tabular}
}
\caption{Results for an American put-on-average option under  Cases I, II, III. Reference prices by FD method (operator splitting)  from \cite{boen2020operator}[Table~6]. }
\label{tab:result07}
\end{table}
\end{minipage}



%% file: section_7_Conclusion.tex
\section{Conclusion and future work}
\label{sc:conclude}
 In this paper, we address an important gap in the numerical analysis of two-asset American options under the Merton jump-diffusion model by introducing an efficient and straightforward-to-implement monotone scheme based on numerical integration. The pricing of these options involves solving complex 2-D variational inequalities that include cross derivative and nonlocal integro-differential terms due to jumps. Traditional finite difference methods often struggle to maintain monotonicity in cross derivative approximations—crucial for ensuring convergence to the viscosity solution-and accurately discretize 2-D jump integrals. Our approach overcomes these challenges by leveraging an infinite series representation of the Green’s function, where each term is non-negative and computable, enabling efficient approximation of 2-D convolution integrals through a monotone integration method. In addition, we rigorously establish the stability and consistency of the proposed scheme in the viscosity sense and prove its convergence to the viscosity solution of the variational inequality. This overcomes  several significant limitations associated with previous numerical techniques.

Extensive numerical results demonstrate strong agreement with benchmark solutions from published test cases, including those obtained via operator splitting methods,  highlighting the utility of our approach as a valuable reference for verifying other numerical techniques.

{\apnum{Although our focus has been on American option pricing under a correlated two-asset Merton jump-diffusion model, the core methodology—particularly the infinite series representation of the Green's function, where each term is non-negative—has broader applicability. One such application is asset allocation with a stock index and a bond index in both discrete and continuous rebalancing settings.
Extending this approach to the 2-D Kou model introduces significant additional challenges, which could be addressed using a neural network-based approximation of the joint PDF of the jumps, along with a Gaussian activation function, as previously discussed. While we utilize the closed-form Fourier transform of the Green's function for the Merton model, iterative techniques for differential-integral operators, such as those discussed in \cite{garronigreenfunctionssecond92}, could be employed to extend this framework to more general jump-diffusion models.
Exploring such extensions and applying our framework to a wider range of financial models remains an exciting direction for future research.}}

%% file: Appendix.tex
\section{Proof of Lemma~\ref{lemma:bd_error}}
\label{app:bd_error}
We extend the methods from \cite{Cont2005}, originally developed for 1-D European options, to address 2-D variational inequalities \eqref{eq:VIs_log_full} and \eqref{eq:VIs_log}.
For simplicity, we denote by $d(\tau)$ the discounting factor.
The solution $v'(\x)$ to the full-domain variational inequality \eqref{eq:VIs_log_full} is simply
\[
v'(\x)=\sup_{\gamma\in[0,\tau]}\Ebb_{\tau}^{\myx,\myy}\l[d(\tau)\hat{v}({\zblue{X}^{'}_{\gamma}},{\zblue{Y}^{'}_{\gamma}})\r],
\]
{\zblue{which comes from \eqref{eq:vf_opt_stop} with a change of variables from $\l(X_{t},Y_{t}\r)$ to $(X'_t,Y'_t)=\l(\ln(X_{t}),\ln(Y_{t})\r)$ and from $t$ to $\tau$}.}
To obtain a probabilistic representation of the solution $v(\x)$ to the localized variational inequality \eqref{eq:VIs_log},
for fixed $\x = (x, y, \tau)$, we define the random variables
$ M_{\tau}^{x}=\sup_{\varsigma\in [0,\tau]}|{\zblue{X}^{'}_{\varsigma}}+x|$ and $M_{\tau}^{y}=\sup_{\varsigma\in [0,\tau]}|{\zblue{Y}^{'}_{\varsigma}}+y|$ to respectively represent the maximum deviation
of processes $\{{\zblue{X}^{'}_{\varsigma}}\}$ and $\{{\zblue{Y}^{'}_{\varsigma}}\}$ from $x$ and $y$ over the interval $[0, \tau]$.
We also define the random variable  $\theta(x)=\inf\{\varsigma\geq0,|{\zblue{X}^{'}_{\varsigma}}+x|\geq A\}$ as the first exit time of the process $ \{{\zblue{X}^{'}_{\varsigma}}+x\}$  from $[-A,A]$. Similarly, the random variable $\theta(y)$ is defined
for the process $\{{\zblue{Y}^{'}_{\varsigma}}+y\}$. Using these random variables, $v(\x)$ can be expressed as
\EQA
&&v(\x)=\sup_{\gamma\in[0,\tau]}\Ebb_{\tau}^{\myx,\myy}\l[d(\tau)\l(\hat{v}({\zblue{X}^{'}_{\gamma}},{\zblue{Y}^{'}_{\gamma}}) \mathbb{I}_{\l\{\{M_{\tau}^{x}<A\}\cap\{M_{\tau}^{y}<A\}\r\}}+ \hat{p}({\zblue{X}^{'}_{\gamma}},{\zblue{Y}^{'}_{\theta(y)}}) \mathbb{I}_{\l\{\{M_{\tau}^{x}< A\}\cap\{M_{\tau}^{y}\geq A\}\r\}}\r.\r.
\nonumber\\
&&\qquad \qquad \l.\l.+  \hat{p}({\zblue{X}^{'}_{\theta(x)}},{\zblue{Y}^{'}_{\gamma}}) \mathbb{I}_{\l\{\{M_{\tau}^{x}\geq A\}\cap\{M_{\tau}^{y}< A\}\r\}} + \hat{p}({\zblue{X}^{'}_{\theta(x)}},{\zblue{Y}^{'}_{\theta(y)}}) \mathbb{I}_{\l\{\{M_{\tau}^{x}\geq A\}\cap\{M_{\tau}^{y}\geq A\}\r\}} \r) \r].
\nonumber
\ENA
Subtracting $v(\cdot)$ from $v'(\cdot)$ gives $\l|v'(\x)-v(\x)\r|\le \ldots$
\begin{align*}
\ldots &\le \sup_{\gamma\in[0,\tau]}\l|\Ebb_{\tau}^{\myx,\myy}\l[d(\tau)\l(\hat{v}({\zblue{X}^{'}_{\gamma}},{\zblue{Y}^{'}_{\gamma}}) \mathbb{I}_{\l\{\{M_{\tau}^{x}\geq A\}\cup\{M_{\tau}^{y}\geq A\}\r\}}
-  \hat{p}({\zblue{X}^{'}_{\gamma}},{\zblue{Y}^{'}_{\theta(y)}}) \mathbb{I}_{\l\{\{M_{\tau}^{x}< A\}\cap\{M_{\tau}^{y}\geq A\}\r\}}\r.\r.\r.
\nonumber
\\
&\qquad \qquad \l.\l.\l.-  \hat{p}({\zblue{X}^{'}_{\theta(x)}},{\zblue{Y}^{'}_{\gamma}}) \mathbb{I}_{\l\{\{M_{\tau}^{x}\geq A\}\cap\{M_{\tau}^{y}< A\}\r\}}
- \hat{p}({\zblue{X}^{'}_{\theta(x)}},{\zblue{Y}^{'}_{\theta(y)}}) \mathbb{I}_{\l\{\{M_{\tau}^{x}\geq A\}\cap\{M_{\tau}^{y}\geq A\}\r\}}\r)\r]\r|,
\nonumber
\\
\leq &\sup_{\gamma\in[0,\tau]}d(\tau)
\l[\Ebb_{\tau}^{\myx,\myy}\l|\hat{v}({\zblue{X}^{'}_{\gamma}},{\zblue{Y}^{'}_{\gamma}}) \mathbb{I}_{\l\{\{M_{\tau}^{x}\geq A\}\cup\{M_{\tau}^{y}\geq A\}\r\}}\r|
+
\Ebb_{\tau}^{\myx,\myy}\l|\hat{p}({\zblue{X}^{'}_{\gamma}},{\zblue{Y}^{'}_{\theta(y)}}) \mathbb{I}_{\l\{\{M_{\tau}^{x}< A\}\cap\{M_{\tau}^{y}\geq A\}\r\}}\r|\r.
\nonumber
\\
& \qquad \qquad \l.
+  \Ebb_{\tau}^{\myx,\myy}\l|\hat{p}({\zblue{X}^{'}_{\theta(x)}},{\zblue{Y}^{'}_{\gamma}}) \mathbb{I}_{\l\{\{M_{\tau}^{x}\geq A\}\cap\{M_{\tau}^{y}< A\}\r\}}\r|
+ \Ebb_{\tau}^{\myx,\myy}\l|\hat{p}({\zblue{X}^{'}_{\theta(x)}},{\zblue{Y}^{'}_{\theta(y)}}) \mathbb{I}_{\l\{\{M_{\tau}^{x}\geq A\}\cap\{M_{\tau}^{y}\geq A\}\r\}}\r|\r],
\nonumber
\\
\leq &\sup_{\gamma\in[0,\tau]}d(\tau)\l[\l\|\hat{v}\r\|_{\infty} \mathfrak{Q}\l(\{M_{\tau}^{x}\geq A\}\cup\{M_{\tau}^{y}\geq A\}\r)+ \l\|\hat{p}\r\|_{\infty} \mathfrak{Q}\l(\{M_{\tau}^{x}\geq A\}\cup\{M_{\tau}^{y}\geq A\}\r)\r],
\nonumber
\\
\leq&\sup_{\gamma\in[0,\tau]}d(\tau)\l[\l(\l\|\hat{v}\r\|_{\infty}+\l\|\hat{p}\r\|_{\infty} \r)\l(\mathfrak{Q}\l(M_{\tau}^{x}\geq A\r)+\mathfrak{Q}\l(M_{\tau}^{y}\geq A\r)\r)\r],
\nonumber
\\
\leq&\sup_{\gamma\in[0,\tau]}d(\tau)\l[\l(\l\|\hat{v}\r\|_{\infty}+\l\|\hat{p}\r\|_{\infty} \r)\l(\mathfrak{Q}\l(M_{\tau}^{0}\geq A-|x|\r)+\mathfrak{Q}\l(M_{\tau}^{0}\geq A-|y|\r)\r)\r],
\nonumber\\
{\buildrel (\text{i}) \over \leq}&\sup_{\gamma\in[0,\tau]}d(\tau)\l[\l(\l\|\hat{v}\r\|_{\infty}+\l\|\hat{p}\r\|_{\infty} \r)C'(\tau)\l({\zblue{e^{-(A-|x|)}+e^{-(A-|y|)}}}\r)\r],
\nonumber
\\
=&~C(\tau)\l(\l\|\hat{v}\r\|_{\infty}+\l\|\hat{p}\r\|_{\infty} \r)\l(e^{-(A-|x|)}+e^{-(A-|y|)}\r).
\end{align*}
Here, (i) is due to Theorem 25.18 of \cite{ken1999levy} and {\zblue{Markov}}'s inequality; $C(\tau)$ is a positive bounded constant that does not depend on $x_{\mymin}$, $x_{\mymax}$, $y_{\mymin}$, and $y_{\mymax}$.  This concludes the proof.

\section{Proof of Lemma~\ref{lemma:series_g}}
\label{app:series_g}
By the inverse Fourier transform $\mathfrak{F}^{-1}[\cdot]$ in \eqref{eq:ft_pair}
and the closed-form expression for $G(\bmeta, \Delta \tau)$ in \eqref{eq:vmf_G_closed}, we have
\begin{align}
\label{eq:g_ift}
g(\bfz, \Delta \tau)~=~&
\f{1}{(2\pi)^2} \, \int_{\Rbb^2}
e^{i\bmeta\cdot\bfz}e^{\Psi(\bmeta)\Delta \tau}  ~\md \bmeta
~=~
\f{1}{(2\pi)^2} \, \int_{\Rbb^2}
e^{- \frac{1}{2} \bmeta^{\top} \bfC\bmeta +
i\l(\bmbeta +\bfz\r) \cdot\bmeta + \theta} \, e^{\lambda \Gamma\l( \bmeta \r) \Delta \tau} ~ \md \bmeta,\nonumber\\
{\text{where }} & \bfC = \Delta \tau \,\tilde{\bfC},
\quad
\bmbeta =  \Delta \tau\,\tilde{\bmbeta} ,
\quad
\theta = -(r+\lambda) \Delta \tau.
\end{align}
Following the approach developed in \cite{dangJacksonSues2016, zhang2023monotone, berthe2019shannon}, we expand the term $e^{\lambda \Gamma(\bmeta) \Delta \tau}$ in \eqref{eq:g_ift} in a Taylor series,  noting that

\begin{align}
\l(\Gamma(\bmeta)\r)^k &= \l(\int_{\Rbb^2} f(\bfs) \exp(i \bmeta\cdot \bfs) \md \bfs\r)^k\nonumber\\
&=\l(\int_{\Rbb^2} f(\bfs_1) \exp(i \bmeta\cdot \bfs_1) \md \bfs_1\r)\l(\int_{\Rbb^2} f(\bfs_2) \exp(i \bmeta\cdot \bfs_2) \md \bfs_2\r)\ldots\l(\int_{\Rbb^2} f(\bfs_k) \exp(i \bmeta\cdot \bfs_k) \md \bfs_k\r)
\nonumber\\
&= \int_{\Rbb^2} \ldots \int_{\Rbb^2}
             f(\bfs_{1})f(\bfs_{2})\ldots f(\bfs_{k}) \exp\l(i\bmeta \cdot \bfs_1\r)\exp\l(i\bmeta \cdot \bfs_2\r)\ldots \exp\l(i\bmeta \cdot \bfs_k\r)
            \md \bfs_1 \md \bfs_2 \ldots \md \bfs_k,
\nonumber\\
&= \int_{\Rbb^2} \ldots \int_{\Rbb^2}
            \prod_{\ell = 1}^{k} f(\bfs_{\ell}) \exp\l(i\bmeta \cdot \bmS_k\r)
            \md \bfs_1 \md \bfs_2 \ldots \md \bfs_k.
\end{align}
Here, $\bfs_{\ell}=[s_1,s_2]_{\ell}$ is the $\ell$-th column vector, and each pair of $\bfs_{i}$ and $\bfs_{j}$ being independent and identically distributed (i.i.d) for $i\neq j$, $\bmS_k = \sum_{\ell=1}^{k} \bfs_{\ell}=\sum_{\ell=1}^{k}[s_1, s_2]_{\ell}$, with ${\bmS}_0 = [0,0]$, and
for $k = 0$, $\l(\Gamma(\bmeta)\r)^0 = 1$. Then, we have the Taylor series for $e^{\lambda\Gamma(\bmeta)\Delta \tau}$ as follows
\EQA
\label{eq:egamma_Taylor}
e^{\lambda\Gamma(\bmeta)\Delta \tau} = \sum_{k=0}^{\infty} \f{\l(\lambda\Delta \tau\r)^k}{k!}(\Gamma(\bmeta))^{k}=\sum_{k=0}^{\infty} \f{\l(\lambda\Delta \tau\r)^k}{k!}\int_{\Rbb^2} \ldots \int_{\Rbb^2}
            \prod_{\ell = 1}^{k} f(\bfs_{\ell}) \exp\l(i\bmeta \cdot \bmS_k\r)
            \md \bfs_1 \md \bfs_2 \ldots \md \bfs_k.
\ENA
We now substitute equation \eqref{eq:egamma_Taylor} into the Green's function $g(\bfz,\Delta \tau)$ in \eqref{eq:g_ift}, which is expressed through substitutions as

\begin{align}
    &g(\bfz, \Delta \tau) =
\f{1}{(2\pi)^2}
\sum_{k=0}^{\infty} \f{\l(\lambda\Delta \tau\r)^k}{k!}
\int_{\Rbb^2}
e^{- \frac{1}{2} \bmeta^{\top} \mathbf{C}\bmeta +
i\l(\bmbeta +\bfz\r) \cdot\bmeta  + \theta} \,
\int_{\Rbb^2} \ldots \int_{\Rbb^2}
            \prod_{\ell = 1}^{k} f(\bfs_{\ell}) \exp\l(i\bmeta \cdot \bmS_k\r)
            \md \bfs_1 \md \bfs_2 \ldots \md \bfs_k ~ \md \bmeta \nonumber
\\
&{\buildrel (\text{i}) \over =}~ \f{1}{(2\pi)^2}
\sum_{k=0}^{\infty} \f{\l(\lambda\Delta \tau\r)^k}{k!}\,
\int_{\Rbb^2} \ldots \int_{\Rbb^2}
            \prod_{\ell = 1}^{k} f(\bfs_{\ell})\int_{\Rbb^2}
e^{- \frac{1}{2} \bmeta^{\top} \mathbf{C}\bmeta +
i\l(\bmbeta +\bfz+\bmS_k\r) \cdot\bmeta  + \theta} ~ \md \bmeta\,
            \md \bfs_1 \md \bfs_2 \ldots \md \bfs_k  {\buildrel (\text{ii}) \over =} \ldots \nonumber
\\
\nonumber
&\f{1}{ 2\pi\s{\det(\bfC)} }
\sum_{k=0}^{\infty} \f{\l(\lambda \Delta \tau\r)^{k}}{k!}
\int_{\Rbb^2} \ldots \int_{\Rbb^2}
            	\exp\l(\theta  -
	\f{\l(\bmbeta + \bfz +  \bmS_k\r)^{\top}\bfC^{-1}(\bmbeta + \bfz +  \bmS_k)}{2}\r)
    \l(\prod_{\ell = 1}^{k} f(\bfs_\ell)\r) \md \bfs_1 \ldots \md \bfs_k .
\end{align}
Here, (i) is due to the Fubini’s theorem, in (ii), we apply the result for the multidimensional Gaussian-type integral, i.e $\int_{\Rbb^{n}}\exp(-\frac{1}{2}\bfx^{\top}\bfA\bfx+\boldsymbol{b}^{\top}\bfx+c)~\md\bfx=\s{\det(2\pi\bfA^{-1})}e^{\frac{1}{2}\boldsymbol{b}^{\top}\bfA^{-1}\boldsymbol{b}+c}$, and the determinant, $\det(\bfC)=\det(\Delta \tau\,\tilde{\bfC})=(\Delta \tau)^2\det(\tilde{\bfC})=(\Delta \tau)^2\sigma_{\myx}^2\sigma_{\myy}^2(1-\rho^2)$.

\section{Proof of  Corollary \ref{cor:twodis}}
\label{app:merton}
   Recalling \eqref{eq:g_proof_sum}, we have $g(\bfz;\Delta t)=\f{\exp\l(\theta -
	\f{\l(\bmbeta + \bfz\r)^{\top}\bfC^{-1}(\bmbeta + \bfz)}{2}\r)}
{2\pi\s{ \det(\bfC) }}\ldots$
    \begin{align}
        \label{eq:reform_g}
        \ldots&+\f{e^{\theta}}{2 \pi\s{ \det(\bfC)}}
\sum_{k=1}^{\infty} \f{\l(\lambda \Delta \tau\r)^{k}}{k!}
\underbrace{\int_{\Rbb^2} \ldots \int_{\Rbb^2}
            	\exp\l( -
	\f{\l(\bmbeta + \bfz +  \bmS_k\r)^{\top}\bfC^{-1}(\bmbeta + \bfz +  \bmS_k)}{2}\r)
      \l(\prod_{\ell = 1}^{k} f(\bfs_\ell)\r) \md \bfs_1 \ldots \md \bfs_k}_{E_k}
\nonumber\\
&=
\f{\exp\l(\theta -
	\f{\l(\bmbeta + \bfz\r)^{\top}\bfC^{-1}(\bmbeta + \bfz)}{2}\r)}
{2\pi\s{ \det(\bfC) }} +
\f{e^{\theta}}{2 \pi\s{ \det(\bfC)}}
\sum_{k=1}^{\infty} \f{\l(\lambda \Delta \tau\r)^{k}}{k!}~E_k.
    \end{align}
    Here, the term $E_k$ in \eqref{eq:reform_g} is clearly non-negative and can be computed as
    \EQA
    \label{eq:E_k_1}
    E_k=\int_{\Rbb^2}\exp\l( -
	\f{\l(\bmbeta + \bfz +  \bfs\r)^{\top}\bfC^{-1}(\bmbeta + \bfz +  \bfs)}{2}\r)f_{\hat{\bmxi}_k}(\bfs)~\md\bfs,
    \ENA
    where $f_{\hat{\bmxi}_k}(\bfs)$ is the PDF of the random variable $\hat{\bmxi}_k=\sum_{\ell=1}^{k}\ln(\bmxi)_{\ell}=\sum_{\ell=1}^{k}[\ln(\xi_{\myx}),\ln(\xi_{\myy})]_{\ell}$ which is the sum of i.i.d random variables for fixed $k$.
 For the  Merton case, we have $\hat{\bmxi}_k\sim \text{Normal}\l(k\bmtmu, k\bfC_{M}\r)$ with the PDF
    \EQA
    \label{eq:PDF_sum}
    f_{\hat{\bmxi}_k}(\bfs) = \f{\exp\l(-
	\f{\l(\bfs - k\bmtmu\r)^{\top}(k\bfC_{\mathcal{M}})^{-1}(\bfs - k\bmtmu)}{2}\r)}
{2\pi\s{ \det(k\bfC_{\mathcal{M}}) }}.
    \ENA
   By substituting the equation \eqref{eq:PDF_sum} into \eqref{eq:E_k_1}, we have
       \begin{align}
    \label{eq:E_k_2}
    E_k&=\int_{\Rbb^2}\exp\l( -
	\f{\l(\bmbeta + \bfz +  \bfs\r)^{\top}\bfC^{-1}(\bmbeta + \bfz +  \bfs)}{2}\r)\f{\exp\l(-
	\f{\l(\bfs - k\bmtmu\r)^{\top}(k\bfC_{\mathcal{M}})^{-1}(\bfs - k\bmtmu)}{2}\r)}
{2\pi\s{ \det(k\bfC_{\mathcal{M}}) }}~\md\bfs \nonumber\\
&=\frac{1}{2\pi\s{ \det(k\bfC_{\mathcal{M}})}}\int_{\Rbb^2}\exp\l( -
	\f{\l(\bmbeta + \bfz +  \bfs\r)^{\top}\bfC^{-1}(\bmbeta + \bfz +  \bfs)+
	\l(\bfs - k\bmtmu\r)^{\top}(k\bfC_{\mathcal{M}})^{-1}(\bfs - k\bmtmu)}{2}\r) ~\md\bfs\nonumber\\
 &{\buildrel (\text{i}) \over =}\frac{1}{2\pi\s{ \det(k\bfC_{\mathcal{M}})}}\int_{\Rbb^2}\exp\l(-\frac{\bfs^{\top}\l(\bfCI+(k\bfC_{\mathcal{M}})^{-1}\r)\bfs}{2}+\l((k\bmtmu)^{\top}(k\bfC_{\mathcal{M}})^{-1}-(\bfbz)^{\top}\bfCI\r)\bfs\ldots\r.\nonumber\\
  &\qquad\l. \ldots -\f{(\bfbz)^{\top}\bfCI(\bfbz)+(k\bmtmu)^{\top}(k\bfC_{\mathcal{M}})^{-1}(k\bmtmu)}{2}\r)~\md\bfs
\end{align}
    Here, in (i), we use matrix multiplication distributive and associative properties. For simplicity, we adopt the following notational convention: $\bfA=\bfCI+(k\bfC_{\mathcal{M}})^{-1}$, which is positive semi-definite and symmetric, and $\bfal=\bfbz$. Then, equation \eqref{eq:E_k_2} becomes
           \begin{align}
    \label{eq:E_k_3}
    E_k
 &=\frac{1}{2\pi\s{ \det(k\bfC_{\mathcal{M}})}}\int_{\Rbb^2}\exp\l(-\frac{\bfs^{\top}\bfA\bfs}{2}+\l((k\bmtmu)^{\top}(k\bfC_{\mathcal{M}})^{-1}-\bfal^{\top}\bfCI\r)\bfs\ldots\r.\nonumber\\
  &\qquad\l. \ldots -\f{\bfal^{\top}\bfCI\bfal+(k\bmtmu)^{\top}(k\bfC_{\mathcal{M}})^{-1}(k\bmtmu)}{2}\r)~\md\bfs\nonumber\\
  &{\buildrel (\text{i}) \over =}\frac{\s{\det\l(\bfA^{-1}\r)}}{\s{ \det(k\bfC_{\mathcal{M}})}}\exp\l(\frac{\l((k\bmtmu)^{\top}(k\bfC_{\mathcal{M}})^{-1}-\bfal^{\top}\bfCI\r)\bfA^{-1}\l((k\bmtmu)^{\top}(k\bfC_{\mathcal{M}})^{-1}-\bfal^{\top}\bfCI\r)^{\top}}{2}\ldots\r.\nonumber\\
    &\qquad\l.\ldots -\f{\bfal^{\top}\bfCI\bfal+(k\bmtmu)^{\top}(k\bfC_{\mathcal{M}})^{-1}(k\bmtmu)}{2}\r)\nonumber\\
  &=\frac{\s{\det\l(\bfA^{-1}\r)}}{\s{ \det(k\bfC_{\mathcal{M}})}}\exp\l(\frac{\l((k\bmtmu)^{\top}(\bfA-\bfCI)-\bfal^{\top}\l(\bfA-(k\bfC_{\mathcal{M}})^{-1}\r)\r)\bfA^{-1}\l((k\bmtmu)^{\top}(k\bfC_{\mathcal{M}})^{-1}-\bfal^{\top}\bfCI\r)^{\top}}{2}\ldots\r.\nonumber\\
    &\qquad\l.\ldots -\f{\bfal^{\top}\bfCI\bfal+(k\bmtmu)^{\top}(k\bfC_{\mathcal{M}})^{-1}(k\bmtmu)}{2}\r)\nonumber\\
  &{\buildrel (\text{ii}) \over =}\frac{\s{\det\l(\bfA^{-1}\r)}}{\s{ \det(k\bfC_{\mathcal{M}})}}\exp\l(\frac{1}{2}\l((k\bmtmu)^{\top}\bfA\bfA^{-1}(k\bfC_{\mathcal{M}})^{-1}(k\bmtmu)-(k\bmtmu)^{\top}\bfCI\bfA^{-1}(k\bfC_{\mathcal{M}})^{-1}(k\bmtmu)\ldots\r.\r.\nonumber\\
  &\qquad\l.\l.\ldots +\bfal^{\top}\bfA\bfA^{-1}\bfCI\bfal-\bfal^{\top}(k\bfC_{\mathcal{M}})^{-1}\bfA^{-1}\bfCI\bfal\r) -\f{\bfal^{\top}\bfCI\bfal+(k\bmtmu)^{\top}(k\bfC_{\mathcal{M}})^{-1}(k\bmtmu)}{2}\r)\nonumber\\
  &{\buildrel (\text{iii}) \over =}\frac{\s{\det\l(\bfC(\bfC+k\bfC_{\mathcal{M}})^{-1}(k\bfC_{\mathcal{M}})\r)}}{\s{ \det(k\bfC_{\mathcal{M}})}}\exp\l(-\frac{(\bfal+k\bmtmu)^{\top}(\bfC+k\bfC_{\mathcal{M}})^{-1}(\bfal+k\bmtmu)}{2}\r)\nonumber\\
  &{\buildrel (\text{iv}) \over =}\frac{\s{\det(\bfC)}\exp\l(-\frac{(\bfbz+k\bmtmu)^{\top}(\bfC+k\bfC_{\mathcal{M}})^{-1}(\bfbz+k\bmtmu)}{2}\r)}{\s{\det(\bfC+k\bfC_{\mathcal{M}})}}.
\end{align}
Here, in (i), we apply the result $\int_{\Rbb^{n}}\exp(-\frac{1}{2}\bfx^{\top}\bfA\bfx+\boldsymbol{b}^{\top}\bfx+c)~\md\bfx=\s{\det(2\pi\bfA^{-1})}e^{\frac{1}{2}\boldsymbol{b}^{\top}\bfA^{-1}\boldsymbol{b}+c}$; (ii) is due to matrix multiplication distributive and associative properties; in (iii), we use the equality for inverse matrix: $(\bfA^{-1}+\boldsymbol{B}^{-1})^{-1}=\bfA(\bfA+\boldsymbol{B})^{-1}\boldsymbol{B}$, and (iv) is due to the determinant of a product of matrices, i.e $\det(\bfA\boldsymbol{B})=\det(\bfA)\det(\boldsymbol{B})$. Using \eqref{eq:reform_g} and \eqref{eq:E_k_3} together with further simplifications gives us the desired result.

{\apnum{
\section{Proof of Lemma~\ref{lemma:bounded_g_scaled}}
\label{app:bounded_g_scaled}
We inspect each term in the full series definition of 
$ \widetilde{g}(x_{n-l},y_{j-d}, \Delta \tau)
=
\Delta x\,\Delta y\,  \odot  g(x_{n-l},y_{j-d}, \Delta \tau)$.
Recall that $g(x_{n-l},y_{j-d}, \Delta \tau)$,
is given by the infinite series representation in Corollary~\ref{cor:twodis}.
Thus, with the notation $\mathbf{z}_{n-l,j-d} = [x_{n-l}, y_{j-d}]= [(n-l) \Delta x,(j-d)\Delta y]$, folding $\Delta x\,\Delta y$ into the denominator of $g(x_{n-l},y_{j-d}, \Delta \tau)$ yields: $\widetilde{g}(x_{n-l},y_{j-d}, \Delta \tau) =
\sum_{k=0}^{K} \widetilde{g}^{(k)}(x_{n-l},y_{j-d}, \Delta\tau) =  \ldots$
\[
\ldots
\;=\;
\sum_{k=0}^{\infty}
  \frac{(\lambda\,\Delta\tau)^{k}}{k!}
  \cdot
  \underbrace{\frac{\Delta x\Delta y}{2\pi \,\sqrt{\det(\mathbf{C}+k\,\mathbf{C}_{\mathcal{M}})}}}_{A_k}
  \exp\!\Bigl(\theta
    - \tfrac12 \bigl(\boldsymbol{\beta}+\mathbf{z}_{n-l,j-d}+k\,\boldsymbol{\tilde\mu}\bigr)^{\top}
                         \bigl(\mathbf{C}+k\mathbf{C}_{\mathcal{M}}\bigr)^{-1}
                         \bigl(\boldsymbol{\beta}+\mathbf{z}_{n-l,j-d}+k\,\boldsymbol{\tilde\mu}\bigr)
     \Bigr)
  .
\]
For the rest of the proof, let $\hat{C}$  be a generic non-negative bounded constant, which may take different values from line to line.
\begin{itemize}
\item Case $k = 0$: we have $\mathbf{C}+0\,\mathbf{C}_{\mathcal{M}} = \mathbf{C} = (C_3\,h)\,\tilde{\mathbf{C}}$. Its determinant $\det(\mathbf{C}) = (C_3\,h)^2
    \det(\tilde{\mathbf{C}})$, with $\det(\tilde{\mathbf{C}})=\sigma_{\myx}^2\sigma_{\myy}^2(1-\rho^2)$ as given in  \eqref{eq:cov_m}. Thus, $\det(\mathbf{C}) \sim h^2$, so $\sqrt{\det(\mathbf{C})}\sim h$.
    Hence, multiplying by $\Delta x\,\Delta y \sim h^2$ and dividing by $\sqrt{\det(\mathbf{C})}\sim h$ yields the factor $A_0 \sim h$.

    Next, we consider the exponential term. Here, note that $\bigl(\mathbf{C}\bigr)^{-1} = \mathcal{O}(1/h)$ and $\boldsymbol{\beta} = \mathcal{O}(h)$.
    Regarding $\mathbf{z}_{n-l,j-d}$,
    under the Assumption~\ref{as:dis_parameter}, particularly,
    $P_x^{\dagger} = C'_1/h$ and $P_y^{\dagger} = C'_2/h$, we have $[n-l, j-d]$ ranging from  $\mathcal{O}(1)$  to $\mathcal{O}(1/h^2)$. Thus, the components of  $\mathbf{z}_{n-l,j-d}$ vary between $\mathcal{O}(h)$ to $\mathcal{O}(1/h)$.

    If $\mathbf{z}_{n-l,j-d} =  \mathcal{O}(h)$,
    the product $(\boldsymbol{\beta} + \mathbf{z}_{n-l,j-d})^{\top}\bigl(\mathbf{C}\bigr)^{-1}(\boldsymbol{\beta} + \mathbf{z}_{n-l,j-d}) \sim \hat{C}h$, noting $\mathbf{C}$ is a positive semi-definite matrix.
    Recalling $\theta \sim -h$, it follows that the exponential term
    is of order $\mathcal{O}(e^{-h})$.

    If $\mathbf{z}_{n-l,j-d} = \mathcal{O}(1/h)$, the  product $(\ldots)^{\top} \mathbf{C}^{-1}(\ldots) \sim \hat{C}/h^3$,
    resulting in an overall exponential term of $\mathcal{O}(e^{-1/h^3})$.

    Therefore, for the case $k=0$, the term
    $\widetilde{g}^{(0)}(x_{n-l},y_{j-d}, \Delta\tau)$ varies between  $\mathcal{O}(he^{-1/h^3})$
    and $\mathcal{O}(he^{-h})$.

\item Case $k \ge 1$: in $\mathbf{C}+k\mathbf{C}_{\mathcal{M}}$, $\mathbf{C} = \Delta \tau \tilde{\mathbf{C}}$, so $\mathbf{C} = \mathcal{O}(h)$;
    In addition, $k\mathbf{C}_{\mathcal{M}}$ is nonsingular for each $k\ge 1$
    and $k\mathbf{C}_{\mathcal{M}} = \mathcal{O}(1)$. Hence $\sqrt{\det(\mathbf{C}+k\,\mathbf{C}_{\mathcal{M}})}=\mathcal{O}(1)$, leading to the factor
    $A_k \sim h^2$.

    In the exponent, $\bigl(\mathbf{C}+k\mathbf{C}_{\mathcal{M}}\bigr)^{-1} = \mathcal{O}(1)$;
    also $\boldsymbol{\beta} = \mathcal{O}(h)$
    and $k\,\boldsymbol{\tilde{\mu}}$ is a constant vector for each fixed $k$.

    If $\mathbf{z}_{n-l,j-d} =  \mathcal{O}(h)$,
    $(\boldsymbol{\beta} + \mathbf{z}_{n-l,j-d})^{\top}\bigl(\mathbf{C}+k\mathbf{C}_{\mathcal{M}}\bigr)^{-1}(\boldsymbol{\beta} + \mathbf{z}_{n-l,j-d}) \sim \hat{C}h^2$,
    resulting in an overall exponential term of order $\mathcal{O}(e^{-h})$,
    noting $\theta\sim -h$.

    If $\mathbf{z}_{n-l,j-d} \sim \mathcal{O}(1/h)$, $(\ldots)^{\top}\bigl(\mathbf{C}+k\mathbf{C}_{\mathcal{M}}\bigr)^{-1}(\ldots) \sim \hat{C}/h^2$,
    so the exponential term is of order $\mathcal{O}(e^{-1/h^2})$.

    Meanwhile, $(\lambda\,\Delta\tau)^{k}\sim h^k$.
    So for each $k\ge1$, the term
    $\widetilde{g}^{(k)}(x_{n-l},y_{j-d}, \Delta\tau)$ varies between $\mathcal{O}\!\bigl(h^{2+k}e^{-1/h^2}\bigr)$ and $\mathcal{O}\!\bigl(h^{2+k}e^{-h}\bigr)$.
    Thus, the dominating term of $\sum_{k=1}^{\infty}\mathcal{O}\!\bigl(h^{2+k}\bigr)$ is $\mathcal{O}(h^{3}e^{-1/h^2})$   and  $\mathcal{O}(h^{3}e^{-h})$, respectively.
\end{itemize}
Combining the results yields $\widetilde{g}_{n-l,j-d}(\Delta\tau) = \mathcal{O}(he^{-h})$,
which simplifies to $\mathcal{O}(h)$. In conclusion, in any case, as $h\to0$, $\widetilde{g}(x_{n-l},y_{j-d}, \Delta\tau)\to0$, hence, it stays bounded.
It is trivial that $\widetilde{g}_{n-l, j-d}(\Delta\tau, K)$ is also bounded.
}}